\documentclass{article}

\usepackage[american]{babel}
\usepackage[T1]{fontenc}
\usepackage{csquotes}
\usepackage{amsmath, amssymb, amsthm, bm}
\allowdisplaybreaks[4]
\usepackage{dsfont}
\usepackage{latexsym}
\usepackage{mathrsfs}
\usepackage{enumitem}
\usepackage{mathtools}

\usepackage[a4paper, top=3cm, bottom=3.5cm, left=3.5cm, right=3.5cm]{geometry}

\usepackage[colorlinks=true,citecolor=blue,urlcolor=blue,linkcolor=blue]{hyperref}

\usepackage[
    backend=biber,
    style=alphabetic,
    giveninits=true,
    sorting=nyt,
    natbib=true,
    url=false,
    isbn=false,
    maxbibnames=10,
    maxalphanames=4,
    ]{biblatex}
\AtEveryBibitem{\ifentrytype{book}{\clearfield{pages}}{}}
\DeclareNameAlias{default}{family-given}
\renewbibmacro*{in:}{
  \ifentrytype{article}
    {}
    {\printtext{\bibstring{in}\intitlepunct}}}
\DeclareNolabel{\nolabel{\regexp{[\p{P}\p{S}\p{C}\p{Z}]+}}}
\DeclareSourcemap{
  \maps[datatype=bibtex]{
    \map{\step[fieldsource=doi, match=\regexp{\\_}, replace=\regexp{_}]}
    \map{\step[fieldsource=edition, match=\regexp{\A(\d+)(st|nd|rd|th)\Z}, replace=\regexp{$1}]}
  }
}
\bibliography{Biblio, Biblio_ASMF}

\newcommand{\R}{\mathbb{R}}
\newcommand{\N}{\mathbb{N}}
\newcommand{\E}{\mathbb{E}}
\newcommand{\V}{\mathbb{V}}
\newcommand{\BMO}{\mathrm{BMO}}
\newcommand{\A}{\mathcal{A}}
\newcommand{\X}{\mathcal{X}}
\newcommand{\F}{\mathcal{F}}
\newcommand{\Q}{\mathbb{Q}}
\renewcommand{\P}{\mathbb{P}}
\newcommand{\1}{\mathds{1}}

\newcommand{\tinybullet}{\mathbin{\vcenter{\hbox{\tiny$\bullet$}}}}
\newcommand{\ds}{\displaystyle}
\renewcommand{\d}{\,\mathrm{d}}
\newcommand{\itemref}[1]{$\textit{\ref{#1}}$}
\newcommand{\sss}{\mathchoice
  {\scriptscriptstyle}{\scriptscriptstyle}{\scriptscriptstyle}{\scriptscriptstyle}}

\def\qed{\unskip\nobreak\hfill$\Box$\par\addvspace{\medskipamount}}
\def\a{{\alpha}}
\def\r{{\rho}}
\def\l{{\lambda}}
\def\Om{{\Omega}}
\def\om{{\omega}}
\def\s{{\sigma}}
\def\g{{\gamma}}
\def\eps{{\varepsilon}}

\numberwithin{equation}{section}
\mathtoolsset{showonlyrefs}

\newtheorem{theorem}{Theorem}[section]
\newtheorem{corollary}[theorem]{Corollary}
\newtheorem{definition}[theorem]{Definition}
\newtheorem{lemma}[theorem]{Lemma}
\newtheorem{proposition}[theorem]{Proposition}
\newtheorem{example}[theorem]{Example}

\theoremstyle{remark}
\newtheorem{remark}[theorem]{Remark}
\newtheorem{exampletext}[theorem]{Example}

\DeclareMathOperator*{\esssup}{ess\,sup}
\DeclareMathOperator*{\essmax}{ess\,max}
\DeclareMathOperator*{\argmin}{arg\,min}
\DeclareMathOperator*{\argmax}{arg\,max}

\title{
Financial Resilience Evaluation: 
From Conditional Expectations to Dynamic Convex
Risk Measures\thanks{We are very grateful to 
Yacine A\"it-Sahalia, Jose Blanchet, Freddy Delbaen, Marco Frittelli, Sergio Ortobelli Lozza,
and seminar and conference participants at 
QFW 2026 in Bergamo, 
IMPMS 2026 in Palermo, 
SAC 2026 in Stockholm, 
the 2026 Linnaeus Workshop on Stochastic Analysis and Applications in V\"axj\"o, 
the 13th World Congress of the Bachelier Finance Society in Bologna,
the STAR seminar at the University of Oslo, SKKU Seoul,
SUSTech Shenzhen, the University of Amsterdam and ETH Z\"urich
for useful comments and discussions.
This research was funded in part by
the Netherlands Organization for Scientific Research (NWO) under grant NWO Vici 2020--2027 (Ferrari, Laeven, Zullino)
and by an UvA-AC3R Midsize Grant (Ferrari, Laeven).
Matteo Ferrari, Emanuela Rosazza Gianin, and Marco Zullino are members of GNAMPA-INdAM, Italy. 
}} 
\author{Matteo Ferrari \\
{\footnotesize Dept.~of Quantitative Economics}\\
{\footnotesize University of Amsterdam, The Netherlands}\\
{\footnotesize \texttt{M.Ferrari@uva.nl}}\\
\and Roger J.~A.~Laeven\footnote{Corresponding author.} \\
{\footnotesize Dept.~of Quantitative Economics}\\
{\footnotesize University of Amsterdam, CentER}\\
{\footnotesize and EURANDOM, The Netherlands}\\
{\footnotesize \texttt{R.J.A.Laeven@uva.nl}}\\
\and Emanuela Rosazza Gianin \\
{\footnotesize Dept.~of Statistics and Quantitative Methods}\\
{\footnotesize University of Milano-Bicocca, Italy}\\
{\footnotesize \texttt{emanuela.rosazza1@unimib.it}}\\
\and Marco Zullino \\
{\footnotesize Dept.~of Mathematics and Applications}\\
{\footnotesize University of Milano-Bicocca, Italy}\\
{\footnotesize \texttt{M.Zullino@campus.unimib.it}}\\
}

\date{\today}

\begin{document}

\maketitle

\begin{abstract}
    Financial resilience concerns whether, and how rapidly, a risky position recovers or deteriorates at a future time $t$, given the information available at time $s\leq t$.
    As a first step, the \emph{resilience rate} was introduced as the expected instantaneous rate of change of a price or risk-assessment process.
    Because this quantity captures only the conditional mean of future increments, 
    it cannot distinguish between positions with the same expected recovery but different conditional risk profiles.
    We obtain a richer characterization by evaluating such increments through a possibly nonlinear dynamic risk measure $\r$.
    For an It\^o process $\pi$ and a normalized, cash-additive $\rho$, we define the \emph{resilience evaluation} 
    as the $L^1$-limit of $\eps^{-1}\rho_s(\pi_{t+\eps}-\pi_t)$, as $\eps\to 0^+$, whenever it exists.
    For a convex $\r$ induced by BSDEs with Lipschitz or quadratic drivers, 
    we prove existence and a dual representation:
    it is the worst-case conditional expectation, over zero-penalty measure changes, of an effective drift that combines the physical drift of $\pi$ with the risk adjustment assigned by $\rho$ to the volatility of $\pi$.
    We establish attainment of the optimal scenario 
    and illustrate the scope of the construction and the role of the required assumptions through examples and counterexamples.


\end{abstract}

\newpage

\setcounter{tocdepth}{3}
\tableofcontents

\section{Introduction}
As risks resolve and new information becomes progressively available, assessments of financial positions change over time.
Dynamic risk measures formalize this evolution by assigning to each position $X$ a process $(\rho_t(X))_{t\in[0,T]}$ of conditional evaluations along a filtration ${\bm \F}:=\big(\F_t\big)_{t\in[0,T]}$ (see, e.g., \cite{Cvitanic+Karatzas_1999_On_dynamic_measures_of_risk, Delbaen_2006_Structure_m-stable_sets_particular_set_risk_neutral_measures,Frittelli+RosazzaGianin_2004_Dynamic_Convex_Risk_Measures, Bion-Nadal_2009_Time_consistent_dynamic_risk_processes,Follmer+Schied_2016_Stochastic_finance}).
Those induced by backward stochastic differential equations (BSDEs) form a nonlinear, time-consistent class the properties of which are governed by the BSDE driver (see \cite{Peng_1997_Backward_SDE_related_g-expectation,ElKaroui+Peng+Quenez_1997_Backward_stochastic_differential_equations_finance,RosazzaGianin_2006_Risk_measures_g-expectations,Barrieu+ElKaroui_2009_Pricing_hedging_optimally_designing_derivatives_minimization_risk_measures}).

While this literature analyzes the dynamic evolution of the \emph{level} of risk, a complementary and largely unexplored question concerns its local \emph{rate} of change.
This is the perspective of financial \emph{resilience}: the pace at which a risky position tends to recover or further deteriorate at a future time, possibly in response to adverse conditions or stress scenarios.\footnote{Empirical firm-level evidence links financial resilience to innovation and R\&D (\cite{Eiling+Laeven+Xu_2024_Coping_Unexpected_Forward-Looking_Measure_Firm_Resilience}).}
As a first theoretical contribution, \cite{Laeven+Ferrari+Rosazza+Zullino_2025_measuringfinancialresilienceusing} introduced the \emph{resilience rate} of a risk-assessment process $\pi$ as its instantaneous expected rate of change
\[
    \bm\dot\pi_{t|s}
    :=L^1\text{ \!-\!}\lim_{\varepsilon\to0^+}\frac1\eps\E\!\left[\pi_{t+\varepsilon}-\pi_{t}\,\middle|\,\F_s\right], \qquad 0\leq s \leq t < T.
\]
Negative values indicate an average tendency toward future risk reduction or recovery; positive values indicate an average tendency toward higher risk. 
Because this rate depends only on the conditional mean, it cannot distinguish between positions with the same expected recovery but different conditional risk profiles. 
To capture these differences, and the evaluator's risk attitude toward them, we replace the conditional expectation with a possibly nonlinear dynamic risk measure.

To fix ideas, consider two loss positions that recover at the same expected pace, the first one deterministically and the second one subject to random fluctuations.
They have the same resilience rate, i.e., the mean criterion regards them as equally resilient, despite the fact that the second one may well deteriorate over the next period.
A risk-averse evaluator, such as a supervisor monitoring the recovery of a capital shortfall after a shock, should rank them differently, and regard the second one as deteriorating when its fluctuations are large enough relative to its expected recovery.

We formalize this idea through a nested, two-layer construction.
The inner layer is a real-valued loss/risk-assessment process $\pi$ following Itô dynamics,
\[
    \d \pi_t=b_t\d t+\s_t\cdot\d W_t,\qquad t\in[0,T],
\]
whose local recovery or deterioration is being assessed.
Here, $b$ and $\sigma$ are sufficiently regular stochastic processes (see later) and $W$ is a multidimensional standard Brownian motion. 
The outer layer is a normalized, cash-additive dynamic risk measure $\rho=(\rho_s)_{s\in[0,T]}$, applied to short-time increments of $\pi$.
We study the existence and identification of the \textit{resilience evaluation}
\begin{equation}
\label{EQ:intro:gen_res}
    \mathcal D_s^\rho \pi_t
    := L^1\text{ \!-\!}\lim_{\eps\to 0^+}\frac1\eps\,\r_s(\pi_{t+\eps}-\pi_t),
    \qquad 0\le s\le t<T.
\end{equation}
Taking $\rho_s(\cdot)=\E[\,\cdot\,|\F_s]$ recovers the \textit{mean} resilience rate from \cite{Laeven+Ferrari+Rosazza+Zullino_2025_measuringfinancialresilienceusing}.

Three benchmarks illustrate the additional information carried by the short-time scaling that defines the resilience evaluation.
The conditional expectation yields
$\mathcal D_s^{\E}\pi_t=\E[b_t\,|\,\F_s]$, isolating the local drift of $\pi$.
The same scaling applied to the conditional variance isolates $\E\big[|\s_t|^2\big|\F_s\big]$.
The entropic risk measure $\mathfrak e^\g_s(X):=\g^{-1}\log\E[\exp(\g X)|\F_s]$, with risk-aversion coefficient $\g>0$, combines the two contributions:
\[
    \mathcal D^{\mathfrak e^\g}_s\pi_t=\E[b_t|\F_s]+\frac\g2\E\big[|\s_t|^2\big|\F_s\big].
\]
In the deterministic vs.\ random example above, and assuming degenerate $b$ and $\sigma$, the two positions have entropic resilience evaluations $b$ and $b+\g|\s|^2/2$, so that the second one is read as deteriorating as soon as $\g|\s|^2>-2b$.

When restricted to sufficiently integrable risks, the entropic risk measure is induced by the quadratic BSDE driver 
$z\mapsto \tfrac\g2|z|^2$.
This suggests considering, as outer evaluation $\r$, the entire, rich class of convex dynamic risk measures induced by BSDEs with driver $g$ that is normalized, independent of the $y$-variable, and convex in the $z$-variable.
We treat two regimes of growth for the driver: the Lipschitz and the quadratic case.
The key structural property we exploit is convexity: it grants $\rho$ a dual representation as a penalized worst-case expectation (see \cite{Barrieu+ElKaroui_2009_Pricing_hedging_optimally_designing_derivatives_minimization_risk_measures, Delbaen+Peng+RosazzaGianin_2010_Representation_penalty_term_dynamic_concave_utilities}), where the penalty term is the time integral of the convex conjugate $g^\ast$ of the BSDE driver $g$:
\[ 
    \rho_s(X) = \esssup_{\mu\in\mathcal B} \E_\mu\bigg[ X-\int_s^T g^\ast(r,\mu_r)\d r \bigg|\F_s \bigg], \qquad s\in[0,T].
\] 
Here, the expectation is computed with respect to the probability $\Q^\mu$ generated by the stochastic exponential of $\int \mu\cdot\d W$, and the dual domain $\mathcal B$ is a suitable class of BMO processes (\cite{Kazamaki_1994_Continuous_exponential_martingales_BMO}).
For quadratic drivers, we extend this representation  to risks with finite exponential moments of all orders, while retaining the fixed BMO dual class.

Our main result establishes that, under suitable integrability assumptions and a growth condition on $g^*$, the resilience evaluation exists and admits the dual representation
\begin{equation}
\label{EQ:intro:main}
    \mathcal D_{s}^\r\pi_t
    =\essmax_{\mu\in\partial^0 \!g}\E_\mu\big[b_t+g(t,\s_t)\big|\F_s\big],
    \qquad 0\leq s \leq t <T,
\end{equation}
for almost every $t$.
Here $\partial^0 \!g$ is the class of processes $\mu\in\mathcal B$ that take values in the zero-penalty set $\partial_z g(\,\cdot\,,0)
    =\big\{q\in \R^m:\ g^\ast(\,\cdot\,,q)=0\big\}$.
Thus the resilience evaluation is represented as a worst conditional expectation, over zero-penalty controls, of a local \textit{effective drift} that combines the physical drift of $\pi$ with the risk adjustment that $\r$ assigns to the volatility $\s$.
We prove that an optimal change of measure exists. 
At $s=t$, the representation reduces to $\mathcal D_t^\rho\pi_t=b_t+g(t,\sigma_t)$, which already extends the classical representation theorem for BSDE drivers, corresponding to $b=0$ and constant $\sigma=z$ (see \cite{Jiang_2008_Convexity_translation_invariance_subadditivity_g-expectations_related_risk_measures} in the Lipschitz case and \cite{Ma+Yao_2010_Quadratic_g-evaluationsexpectations_related_analysis,Zheng+Li_2018_Representation_theorem_generators_quadratic_BSDEs} in the quadratic case), to Itô processes with random coefficients.

The proof is delicate and proceeds in several steps.
It starts from the dual representation of $\rho_s(\pi_{t+\eps}-\pi_t)$ and rewrites the increment by Girsanov's theorem.
We then separate the penalty contribution on the short interval $[t,t+\eps]$ from the one accumulated on $[s,t]$.
The scaling factor $\eps^{-1}$ then amplifies any non-vanishing penalty and forces the dual processes toward the zero-penalty set, localizing the problem to the shrinking interval $[t,t+\eps]$; there, the Fenchel equality $\s\cdot\mu-g^\ast(\,\cdot\,,\mu)=g(\,\cdot\,,\s)$ converts the dual expression back into the driver $g$.
Making this rigorous requires restricting the dual domain, a measurable selection, and a pasting argument before passing to the $L^1$-limit.

We then illustrate the scope of the theory with a range of examples.
First, we consider the case where the inner process $\pi$ is the first component of the solution  of a BSDE with driver $g^\pi$.
The resilience evaluation makes the interaction between the two layers explicit, the inner driver $g^\pi$ governing the local drift of $\pi$ and the outer driver $g$ imposing a risk adjustment to the martingale component of $\pi$.
As a financial illustration, we apply this to a replicating portfolio in the Black--Scholes model, where $\pi$ is the price process of a European payoff.
We also examine linear, sublinear and convex outer drivers, including entropic and robust entropic measures and drivers with subquadratic penalties.

A few counterexamples also clarify when the limit is ill-posed.
For Brownian increments, both Value at Risk and Expected Shortfall scale like $\sqrt\eps$, so their evaluations divided by $\eps$ diverge as $\eps\to0^+$.
Dropping normalization or outer cash-additivity can likewise make the limit ill-posed, even for simple It\^o processes.
\smallskip

The outline of the paper is as follows.
Section~\ref{sec:prel} collects the necessary preliminaries on functional spaces, dynamic risk measures, and BSDEs.
Section~\ref{SEC:res_eval} introduces the resilience evaluation, establishes the dual representation, and discusses its interpretations.
Section~\ref{SEC:examples} presents examples and counterexamples.
Appendix~\ref{SEC:app_taylor} contains conditional asymptotic expansions.
Appendix~\ref{SEC:app_duality} collects technical results on convex duality, measurable selections, and changes of measure, including the extension of the dual representation to unbounded risks.
Appendix~\ref{SEC:app_determ_driver} verifies the growth condition on the convex conjugate for deterministic drivers.
Appendix~\ref{SEC:app_lipschitz} adapts the proof of the main theorem to the Lipschitz case.

\section{Preliminaries}\label{sec:prel}
In this section, we introduce the main functional spaces, present the definitions of  dynamic risk measures and their properties, and recall the required results on BSDEs, together with their relation with dynamic risk measures.

\subsection{Functional setting}
\label{SEC:funct_spaces}
For $m\in\N$, we denote by $a\cdot b$ the Euclidean scalar product of $a,b\in\R^m$, and by $|a|$ the induced norm.
Assume that $f:\R^m\to\R$ is a convex function and let $z_0\in\R^m$. 
Then $f$ is continuous at $z_0$ and its subdifferential at $z_0$
\[
    \partial f(z_0):=\big\{q\in\R^m \ : \ f(z)\geq f(z_0)+q\cdot(z-z_0), \forall\, z\in\R^m\big\}
\]
is a non-empty, compact and convex set in $\R^m$ by \cite[Theorems~10.1 and~23.4]{Rockafellar_1997_Convex_analysis}.
Moreover, the convex conjugate of $f$, denoted by 
\[
    f^\ast:\R^m\ni q \mapsto \sup_{z\in\R^m}\{q\cdot z - f(z)\}\in\R\cup\{+\infty\}, 
\]
is proper, convex, lower semi-continuous and satisfies $f^{\ast\ast}=f$ by \cite[Theorem~12.2]{Rockafellar_1997_Convex_analysis}.
In addition, $f(z)+f^\ast(q)\geq q\cdot z$ for all $z,q\in\R^m$, with equality if and only if $q\in\partial f(z)$; see \cite[Theorem~23.5]{Rockafellar_1997_Convex_analysis}.

For a topological space $E$, we let $\mathscr B(E)$ denote its Borel $\s$-algebra.
Moreover, $\ell_1$ denotes the Lebesgue measure on $\mathscr B(\R)$, or its restriction to Borel subsets of~$\R$.
We use the symbol $\otimes$ for the product of both $\sigma$-algebras and $\s$-finite measures.

If $(A,\Sigma, \nu)$ is a measure space, and $p\in[1,+\infty]$, we denote by $L^p(A,\Sigma,\nu;\R^m)$ the Banach space of (equivalence classes of) functions $A\to\R^m$ that are $\Sigma/\mathscr B(\R^m)$-measurable and $p$-integrable with respect to $\nu$ (or $\nu$-essentially bounded, if $p=\infty$).
We fix a complete probability space $(\Om,\F,\P)$.
As usual, probabilities restricted to sub-$\s$-algebras of $\F$ will be denoted by the same symbol.
When $A=\Om$, $\Sigma=\F$, $\nu=\P$, or $m=1$, the respective argument will be implied. 
For instance, we simply write
\[
    L^p:=L^{p}(\Om,\F,\mathbb P;\R),
    \qquad 
    L^p_t:=L^p\big(\Om\times[0,t],\F\otimes\mathscr B([0,t]),\P\otimes\ell_1\big), \qquad \forall \,t>0,
\]
for the respective space of random variables and jointly measurable stochastic processes.
The norms on $L^p$ and $L^p_t$ are denoted by $\|\,\cdot\,\|_{L^p}$, $\|\,\cdot\,\|_{L^p_t}$, respectively.
We let $L^{\mathrm{exp}}(\Sigma)$ be the linear subspace of $L^1(\Sigma)$ consisting of random variables $X$ with finite exponential moments of all orders, i.e.,~$\E\big[\exp(c|X|)\big]<+\infty$ for all $c>0$.\footnote{We note that $L^{\mathrm{exp}}$ is the Orlicz heart associated with the function $\R\ni x \mapsto e^{|x|}-1.$}

Let $m\in\N$ and assume that $W:\Omega\times[0,+\infty)\to\R^m$ is a standard $m$-dimensional Brownian motion.
Henceforth, we work with a fixed time horizon $T\in(0,+\infty)$.
We let ${\bm \F}:=\big(\F_t\big)_{t\in[0,T]}$ be the complete and right-continuous filtration on $(\Om,\F,\P)$ generated by $W$, namely $\F_t:=\bigcap_{s\in(t,T]}\s\big(W_r \ : \ r\in[0,s]\big)\vee \mathscr N$ for $t\in[0,T)$,
and $\F_T:=\s\big(W_r \ : \ r\in[0,T]\big)\vee \mathscr N$, where $\mathscr N$ is the family of all events in $\F$ of $\P$-probability zero.
We denote by $\mathcal P$ the predictable $\s$-algebra on $\Om\times[0,T]$, i.e.,~the one generated by continuous (equivalently, left-continuous or càglàd) $\bm\F$-adapted processes $\Om\times[0,T]\to\R$.
With a slight abuse of notation, we use the same symbols $\bm\F$ and $\mathcal P$ when the time domain is a subset of $[0,T]$.

Unless otherwise stated, equalities and inequalities between random variables should be understood $\mathbb{P}$\text{-a.s.}, whereas equalities and inequalities for processes must be interpreted $\mathbb{P}\otimes \ell_1$-a.e.
Consider a function $\psi:\Om\times[0,T]\times \R^m\ni(\om,t,z)\mapsto \psi(\om,t,z)\in \R$.
If we say that $\psi$ satisfies $\P\otimes\ell_1$-a.e.\ a property in $z\in \R^m$, we mean that, for $\P\otimes\ell_1$-a.e.\ $(\om,t)\in\Om\times[0,T]$, the function $\R^m\ni z\mapsto \psi(\om,t,z)$ satisfies that property.
Also, for fixed $z\in \R^m$, we denote by $\psi(\,\cdot\,,z)$ the function $\Om\times[0,T]\ni(\om,t)\mapsto \psi(\om,t,z)$.

For any $p\in[1,+\infty)$ and $t\in(0,T]$, we define the following classes of stochastic processes:
{\small\begin{align}
    \mathcal S_t^p&:=\left\{Y:\Om\times[0,t]\to \R \ \textit{continuous  
    predictable } : \ \E\left[\sup_{s\in[0,t]}|Y_s|^p\right]<+\infty\right\},\\
    \mathcal S_t^{\mathrm{exp}}&:=\Bigg\{Y:\Om\times[0,t]\to \R \  \textit{continuous 
    predictable }: \ \E\bigg[\exp\Big(c\sup_{s\in[0,t]}|Y_s|\Big)\bigg]<+\infty, \ \forall\, c>0\Bigg\},\\
    \mathcal S_t^\infty&:=\left\{Y:\Om\times[0,t]\to \R \  \textit{continuous 
    predictable }: \ \bigg\|\sup_{s\in[0,t]}|Y_s|\bigg\|_{L^\infty}<+\infty\right\},\\
    \mathcal H^p_t &:=
    \left\{
    Z:\Om\times[0,t]\to\R^m\ \textit{predictable}:
    \E\left[
    \left(\int_0^t |Z_s|^2\d s\right)^{p/2}
    \right]<+\infty
    \right\},\\[-3ex]
    \BMO_t&:=\Bigg\{Z:\Om\times[0,t]\to \R^m \   \textit{predictable }:\ \bigg\|\sup_{s\in[0,t]}\E\bigg[\int_s^t|{Z_r}|^2\d r \bigg|\F_s\bigg]\bigg\|_{L^\infty}<+\infty\Bigg\}.
\end{align}}
For ease of notation, we use the same symbols for the normed spaces obtained as quotients of the above under the equivalence relation induced by their respective natural (semi)norms.
With this convention, a process in $\mathcal S^p_t$, $\mathcal S^{\mathrm{exp}}_t$ or $\mathcal S^\infty_t$ is unique up to indistinguishability (and, for each equivalence class, we select the continuous representative), while processes in $\mathcal H^p_t$, $\BMO_t$ are unique up to a $\P\otimes\ell_1$-null set (and, for each equivalence class, we select the appropriately measurable representative).

For any continuous local martingale $M$ with $M_0=0$, we denote its stochastic exponential by ${\mathcal E(M) := \exp\big(M-\frac12\langle M\rangle\big)}$, where $\langle M\rangle$ is its quadratic variation process. 
By \cite[Theorem~1.2]{Kazamaki_1994_Continuous_exponential_martingales_BMO}, $\mathcal E(M)$ is a strictly positive local martingale and hence a supermartingale.
If $\mu\in\BMO_T$, then we denote its stochastic integral by $(\mu\tinybullet W):=\int_0^{\,\cdot}\mu\cdot\d W$ and the respective stochastic exponential by $\mathcal E^\mu:= \mathcal E(\mu\tinybullet W)$, which is a uniformly integrable martingale by \cite[Theorem~2.3]{Kazamaki_1994_Continuous_exponential_martingales_BMO}. 
Finally, we define the probability measure $\Q^\mu\sim\P$ by
\begin{equation}
\label{EQ:equivalent_mg_measure}
    \frac{\d\Q^\mu}{\d\P}
    := \mathcal E^\mu_T
    = \exp\left(
    \int_0^T\mu_r\cdot\d W_r
    -\frac12\int_0^T|\mu_r|^2\d r
    \right),
\end{equation}
and we denote by $\E_\mu[\,\cdot\,|\F_s]$ the conditional expectation under $\Q^\mu$. 
By the Girsanov theorem (see, e.g.,\ \cite[Section~3.5, Theorem~5.1]{Karatzas+Shreve_1991_Brownian_motion_stochastic_calculus}) the process $W^\mu_t :=
    W_t-\int_0^t\mu_r\d r$, for $t\in[0,T]$, is a Brownian motion under $\Q^\mu$.

Finally, we recall \cite[Proposition~9]{Laeven+Ferrari+Rosazza+Zullino_2025_measuringfinancialresilienceusing}, a result that generalizes \cite[Proposition~2.2]{Jiang_2008_Convexity_translation_invariance_subadditivity_g-expectations_related_risk_measures} and that will be used extensively in the paper.
\begin{lemma}
\label{LEM:integral_average}
     Let $\psi\in L^q_T$ for some $q\in[1,+\infty)$. 
     For $\ell_1$-a.e.\ $t\in[0,T)$ and all $s\in[0,t]$, we have:
    \begin{align} 
    \label{EQ:PROP:integral_average_determ_time}
        \psi_t=L^q \text{ \!-\!}\lim_{\varepsilon\to0^+}\frac{1}{\varepsilon}\int_{t}^{t+\varepsilon}\psi_r\d r, 
        \qquad 
        \E[\psi_{t}|\F_s] =L^q \text{ \!-\!}\lim_{\varepsilon\to0^+}\E\left[\left.\frac{1}{\varepsilon}\int_{t}^{t+\varepsilon}\psi_r\d r\right|\F_s\right].
    \end{align}
\end{lemma}

\subsection{Dynamic risk measures}
\label{SEC:Dyn_risk_measures}
We adopt throughout the actuarial loss convention: the argument of a risk measure represents a loss, liability, exposure, or risk value, so that larger values are less desirable.
Accordingly, monotonicity is increasing and cash-additivity carries a positive sign;
see \cite{Barrieu+ElKaroui_2009_Pricing_hedging_optimally_designing_derivatives_minimization_risk_measures,Delbaen_2006_Structure_m-stable_sets_particular_set_risk_neutral_measures,Frittelli+RosazzaGianin_2004_Dynamic_Convex_Risk_Measures}.

\begin{definition}[Dynamic risk measure]
\label{DEF:dynamic_risk_measure}
    Let $\X_T$ be a vector subspace of $L^1(\F_T)$ containing constants, and define $\X_t:=\X_T \cap L^1(\F_t)$ for $t\in[0,T)$. 
    A dynamic risk measure on $\X_T$ is a family $\rho=(\rho_t)_{t\in[0,T]}$ such that, for any $t\in[0,T]$, $\r_t:\X_T\to \X_t$ satisfies the following monotonicity property:
    For  any $X,Y\in \X_T$, if $X\leq Y$, then $\rho_t(X)\leq\rho_t(Y)$.
\end{definition}

Let us now fix $\X_T$ as above.
We present a non-exhaustive list of well-known axioms for a dynamic risk measure $\r$ on $\X_T$, to be understood for every $t\in[0,T]$.
{\begin{enumerate}[noitemsep, topsep=0pt,leftmargin=1em]
    \item 
        Normalization: $\rho_t(0)=0$.
    \item 
        Cash-additivity: For $X\in \X_T$, if $h\in \X_t$, then $\rho_t(X+h)=\rho_t(X)+h$.
    \item 
        Positive homogeneity: $\rho_t(\alpha X)=\alpha\rho_t(X)$, for $X\in \X_T$ and $\a \geq 0$.
    \item 
        Convexity: $\rho_t\big(\lambda X+(1-\lambda)Y\big)\leq \lambda\rho_t(X)+(1-\lambda)\rho_t(Y)$, for $X,Y\in \X_T$ and $\lambda\in[0,1]$.
    \item Coherence: cash-additivity + convexity + positive homogeneity.
\end{enumerate}}
For further details on the financial interpretation of the properties listed above, interested readers are referred to \cite{Artzner+Delbaen+Eber+Heath+Ku_2007_Coherent_multiperiod_risk_adjusted_values_Bellmans_principle, Delbaen_2002_Coherent_risk_measures_general_probability_spaces, Follmer+Schied_2002_Convex_measures_risk_trading_constraints, Frittelli+RosazzaGianin_2002_Putting_order_risk_measures,Barrieu+ElKaroui_2009_Pricing_hedging_optimally_designing_derivatives_minimization_risk_measures,Delbaen_2012_Monetary_Utility_Functions,Bellini+Laeven+RosazzaGianin_2018_Robust_return_risk_measures}. 
Recall that positive homogeneity implies normalization, so coherent risk measures are normalized.

If $\rho$ is a dynamic risk measure on $\X_T$, and $X\in \X_T$, we will often work with the naturally defined $\bm\F$-adapted stochastic process ${\rho(X):\Om\times[0,T]\ni(\om,t)\mapsto\rho_t(X)(\om)\in \R}$. 
This notation highlights the connection to BSDEs.

\subsection{Backward stochastic differential equations}
\label{SEC:BSDEs}
We recall the BSDE framework used to construct the dynamic risk measures central to this paper.
\begin{definition}
\label{DEF:solution_BSDE_jumps}
    We call a driver a $\mathcal P\otimes\mathscr B(\R)\otimes\mathscr B(\R^m)$-measurable function 
    \[
        g:\Om\times[0,T]\times\R\times \R^m\ni(\om,t,y,z)\mapsto g(\om,t,y,z)\in \R.
    \]
    If $g$ is a driver, $T'\in(0,T]$, and $X:\Om\to\R$ is $\F_{T'}$-measurable, then a solution to the BSDE with parameters $(g,T',X)$ is a pair $(Y,Z)\in\mathcal S_{T'}^2\times \mathcal H_{T'}^2$ such that, for any $t\in[0,T']$:
    \begin{equation}
    \label{EQ:BSDE}
        Y_t=X+\int_t^{T'}g(s, Y_s, Z_s)\d s - \int_t^{T'} Z_s\cdot \d W_s.
    \end{equation}
\end{definition}

We use two standard BSDE well-posedness results. 
If $g(\,\cdot\,,0,0)\in L^2_T$ and $g$ is $\P\otimes\ell_1$-a.e. uniformly Lipschitz in $(y,z)$, then for every $X\in L^2(\F_T)$, the BSDE~\eqref{EQ:BSDE} with parameters $(g,T,X)$ has a unique solution in $\mathcal S^2_T\times\mathcal H^2_T$; see \cite{Pardoux+Peng_1990_Adapted_solution_backward_stochastic_differential_equation} and \cite[Theorems~4.2.1,~4.3.1]{Zhang_2017_Backward_stochastic_differential_equations}.
If, $\P\otimes\ell_1$-a.e.,  $g$ is independent of $y$, convex in $z$, satisfies $g(\,\cdot\,,0)=0$, and obeys ${|g(\,\cdot\,,z)|\leq k(1+|z|^2)}$ for some constant $k>0$, then for every $X\in L^{\mathrm{exp}}(\F_T)$ there is a unique solution with $Y\in\mathcal S^{\mathrm{exp}}_T$ and $Z\in\mathcal H^p_T$ for every $p\geq1$; see \cite[Corollary~6]{Briand+Hu_2008_Quadratic_BSDEs_convex_generators_unbounded_terminal_conditions}. For bounded $X$, the solution also satisfies $(Y,Z)\in\mathcal S^\infty_T\times\BMO_T$: existence follows from \cite[Lemma~2]{Briand+Hu_2008_Quadratic_BSDEs_convex_generators_unbounded_terminal_conditions}, while uniqueness and regularity follow from the standard estimates and linearization arguments; see \cite[Theorems~7.2.1,~7.3.1]{Zhang_2017_Backward_stochastic_differential_equations} and \cite{Kobylanski_2000_Backward_stochastic_differential_equations_partial_differential_equations_quadratic_growth,Laeven+Stadje_2014_Robust_portfolio_choice_indifference_valuation}.
\smallskip

In the respective regimes, set $\mathcal X_T=L^2(\F_T)$ or $\mathcal X_T=L^{\mathrm{exp}}(\F_T)$ and define for all $t\in[0,T]$ ${\rho_t:\mathcal X_{T}\ni X \mapsto Y^X_t\in \mathcal X_t}$, where $(Y^X,Z^X)$ solves~\eqref{EQ:BSDE}. BSDE comparison ensures that $\rho$ is a dynamic risk measure on $\mathcal X_T$; see \cite[Theorem~4.4.1]{Zhang_2017_Backward_stochastic_differential_equations} and \cite[Theorem~5]{Briand+Hu_2008_Quadratic_BSDEs_convex_generators_unbounded_terminal_conditions}.
We say that $\rho$ is induced (on $\mathcal X_T$) by the (family of BSDEs with) driver $g$.
Conversely, we say that a given dynamic risk measure $\rho$ is BSDE-induced if there exists a driver $g$ satisfying any of the conditions above and such that $\rho$ is induced by $g$.
There are sufficient conditions under which dynamic 
risk measures are BSDE-induced;
see \cite{Peng_2004_Nonlinear_expectations_nonlinear_evaluations_risk_measures,Peng_2005_dynamically_consistent_nonlinear_evaluations,RosazzaGianin_2006_Risk_measures_g-expectations, Coquet+Hu+Memin+Peng_2002_Filtration-consistent_nonlinear_expectations_related_g_expectations,Hu+Ma+Peng+Yao_2008_Representation_theorems_quadratic_F_consistent_nonlinear_expectations}.
    
As is well-known, the properties of BSDE-induced dynamic risk measures are mainly dictated by the corresponding drivers.
In the following list, we make explicit the relation between the driver $g$ and the induced dynamic risk measure (see, for instance, \cite{Barrieu+ElKaroui_2009_Pricing_hedging_optimally_designing_derivatives_minimization_risk_measures,Jiang_2008_Convexity_translation_invariance_subadditivity_g-expectations_related_risk_measures,Laeven+RosazzaGianin+Zullino_2023_Dynamic_return_star-shaped_risk_measures_BSDEs,RosazzaGianin_2006_Risk_measures_g-expectations,Laeven+RosazzaGianin+Zullino_2025_Geometric_bsdes}). 
Let $g$ be a driver in either of the two regimes just described.
Let $\rho$ denote the dynamic risk measure induced by $g$.
Then the following properties hold.
\begin{enumerate}[noitemsep, topsep=0pt, leftmargin=1.3em]
    \item Normalization: 
        $\rho$ is normalized if and only if $g(\,\cdot\,,0,0)=0$.
    \item Cash-additivity: 
        $\rho$ is cash-additive if 
        $g$ is independent of $y$, namely, $\P\otimes\ell_1$-a.e., we have  $g(\,\cdot\,,y,z)=g(\,\cdot\,,0,z)$ for all $(y,z)\in\R\times\R^m$.
        In this case, we identify $g$ with a driver defined on $\Om\times[0,T]\times\R^m$.
    \item Positive homogeneity: 
            If $g$ is $\P\otimes\ell_1$-a.e.\ positively homogeneous in $(y,z)$,
            then $\rho$ is positively homogeneous.
    \item Convexity: 
            If $g$ is $\P\otimes\ell_1$-a.e.\ convex in $(y,z)$, then $\rho$ is convex.\label{PAGE:convex}
    \item Coherence:
        If $g$ is independent of $y$
        and is $\P\otimes\ell_1$-a.e. convex and positively homogeneous in $z$, then $\rho$ is coherent.
\end{enumerate}

\smallskip
The converse implications, namely the derivation of structural properties of the driver from the corresponding properties of the induced risk measure, hold under additional assumptions ensuring the local representation of the driver.
We refer to \cite{Jiang_2008_Convexity_translation_invariance_subadditivity_g-expectations_related_risk_measures} in the Lipschitz case and \cite{Ma+Yao_2010_Quadratic_g-evaluationsexpectations_related_analysis, Zheng+Li_2018_Representation_theorem_generators_quadratic_BSDEs,Hu+Ma+Peng+Yao_2008_Representation_theorems_quadratic_F_consistent_nonlinear_expectations} in the quadratic case.
Applications of BSDEs to (e)valuation, portfolio choice, and optimal stopping are prevalent; e.g., \cite{Lazrak+Quenez_2003_Generalized_stochastic_differential_utility,Mania+Schweizer_2005_Dynamic_exponential_utility_indifference_valuation,Kratschmer+Ladkau+Laeven+Schoenmakers+Stadje_2018_Optimal_stopping_under_uncertainty_drift_jump_intensity} and the references therein.
\section{Resilience Evaluation}
\label{SEC:res_eval}

This section forms the core of the paper.
We first define the resilience evaluation, relate it to the resilience rate of \cite{Laeven+Ferrari+Rosazza+Zullino_2025_measuringfinancialresilienceusing}, and discuss the interpretations that follow from the definition.
We then examine conditional variance and the entropic risk measure as benchmarks.
Next, we state the representation theorem for BSDE-induced convex risk measures and its coherent specialization.
We study the attainment of the resulting essential supremum and discuss the interpretations arising from the dual representation.
For clarity of exposition, the proofs are postponed to the end of the section.

\subsection{Definition}
\label{SEC:definition}
The construction is formulated for a general dynamic risk measure and a continuous It\^o process.
Later on, both objects will be specialized in several directions.

Assume that $\X_T$ is a vector subspace of $L^1(\F_T)$ containing the constants, and set $\X_t:=\X_T\cap L^1(\F_t)$ for all $t\in[0,T]$.
Fix a normalized and cash-additive dynamic risk measure $\rho=(\rho_t)_{t\in[0,T]}$ on $\X_T$.
Let $\pi=(\pi_t)_{t\in[0,T]}$ be a real-valued It\^o process with dynamics
\begin{equation}
\label{EQ:pi_Ito}
    \d \pi_t=b_t\d t+\s_t\cdot\d W_t,
    \qquad t\in[0,T],
\end{equation}
where $b\in L^1_T$ is $\bm\F$-adapted and $\s\in\mathcal H^2_T$.
Further assume that, for every $t\in[0,T)$, there exists $h\in(0,T-t]$ such that the increment $\Delta_\eps\pi_t:=\pi_{t+\eps}-\pi_t $ is in $\X_T$ for all $\eps\in(0,h]$.

Consistently with the actuarial loss convention for $\rho$, $\pi$ is interpreted as a source of risk: an increase in $\pi$ is adverse, whereas a decrease is favorable. In particular, when $\pi$ is a market price process, this corresponds to the viewpoint of an agent holding a short position.

\begin{definition}
\label{DEF:res_eval}
    Let $0\leq s\leq t < T$.
    Whenever the following limit exists in $L^1(\F_s)$, we define the \textbf{resilience evaluation} of $\pi_t$ through $\r_s$ by
    \[
        \mathcal D^\r_s\pi_t
        :=
        L^1\text{ \!-\!}\lim_{\varepsilon\to0^+}
        \frac1\eps\r_s(\pi_{t+\eps}-\pi_t).
    \]
\end{definition}

We first illustrate the definition through the conditional expectation. 
This elementary benchmark also identifies the connection between the resilience evaluation $\mathcal D^\rho\pi$ and the resilience rate $\bm\dot\pi$ introduced in \cite{Laeven+Ferrari+Rosazza+Zullino_2025_measuringfinancialresilienceusing}.

\begin{example}[Resilience rate]
\label{EX:resilience_rate}
    Suppose that the dynamic risk measure $\r$ is the conditional expectation: $\r_s(X)=\E[X|\F_s]$, for $X\in L^1(\F_T)$ and $s\in[0,T]$.
    Let $\pi$ be as above.
    Then the increments $\Delta_\eps\pi_t$ belong to $L^1(\F_T)$ and the random variable
    $\rho_s(\Delta_\eps\pi_t)$ is well-defined in $L^1(\F_s)$, for all
    $0\leq s\leq t<t+\eps\leq T$.
    By direct inspection,
    \begin{align}
        \frac1\eps\r_s(\pi_{t+\eps}-\pi_t)
        =
        \frac1\eps
        \E\left[
            \left.
            \int_t^{t+\eps}b_r\d r
            +\int_t^{t+\eps}\s_r\cdot\d W_r
            \right|\F_s
        \right] =
        \frac1\eps
        \E\left[
            \left.
            \int_t^{t+\eps}b_r\d r
            \right|\F_s
        \right],
    \end{align}
    where we used the martingale property of the It\^o integral.
    In view of Lemma~\ref{LEM:integral_average}, the last expression converges in $L^1$,
    as $\eps\to0^+$, to $\E[b_t|\F_s]$ for $\ell_1$-a.e. $t\in[0,T)$ and all
    $s\in[0,t]$.
    Therefore,
    \begin{equation}
    \label{EQ:res_rate}
        \mathcal D_s^\E\pi_t
        =
        \E[b_t|\F_s],
        \qquad
        \ell_1\text{-a.e. }t\in[0,T),\ \forall\,s\in[0,t].
    \end{equation}
    Suppose, in addition, that $\pi$ is the first component of the solution $(\pi,Z^\pi)$ to a BSDE with parameters $(g^\pi,T,\pi_{\sss T})$, where, for
    instance, $\pi_{\sss T}\in L^2(\F_T)$ and $g^\pi$ is $\P\otimes\ell_1$-a.e. uniformly Lipschitz in $(y,z)$ with $g^\pi(\,\cdot\,,0,0)\in L^2_T$.
    Then $\pi$ can be interpreted as the dynamic risk assessment of the risk
    $\pi_{\sss T}$ and
    \begin{equation}
    \label{EQ:pi_BSDE_relation}
        b=-g^\pi(\,\cdot\,,\pi,Z^\pi),
        \qquad
        \s=Z^\pi.
    \end{equation}
    In this case, $b\in L^2_T$ and $\s\in\mathcal H^2_T$, and \eqref{EQ:res_rate}
    becomes $\mathcal D_s^\E\pi_t
        =
        -\E\big[g^\pi(t,\pi_t,Z^\pi_t)\big|\F_s\big]$.
    This coincides with the conditional resilience rate
    $\bm\dot\pi_{t|s}$ of the risk $\pi_{\sss T}$, as defined in
    \cite{Laeven+Ferrari+Rosazza+Zullino_2025_measuringfinancialresilienceusing}.\footnote{We note that in \cite{Laeven+Ferrari+Rosazza+Zullino_2025_measuringfinancialresilienceusing}, $t$ may, but need not be, a stopping time. 
    Here, we restrict to deterministic times, leaving extensions to stopping times to future work.}
\end{example}

As in the previous example, we will later specialize the space of risks $\X_T$, the outer dynamic risk measure $\rho$, and the inner price/risk assessment process $\pi$ to several specific settings.
In the general results of  Section~\ref{SEC:resil_convex}, we shall assume that $\rho$ is induced by a family of BSDEs with driver $g$, and the domain $\mathcal X_T$ is chosen accordingly.
Concerning the inner process $\pi$, we keep it as general as possible, imposing on its coefficients $b$ and $\s$ only the additional regularity needed to ensure that its increments belong to the chosen ambient space $\X_T$.
Since we have specified only the dynamics for $\pi$, we may either choose $\pi$ as the solution to a forward SDE with given initial condition $\pi_0$, or, as in Example~\ref{EX:resilience_rate}, let $\pi$ be the first component of a solution $(\pi,Z^\pi)$ to a BSDE with parameters $(g^\pi,T,\pi_{\sss T})$.
In the latter case, the relation between the It\^o coefficients and the BSDE parameters is given by~\eqref{EQ:pi_BSDE_relation}.

\begin{remark}[Self-evaluation has zero resilience]
\label{REM:self_evaluation}
    Suppose that the dynamic risk measure $\rho$ is induced by a normalized BSDE driver $g$ independent of $y$, and let $\pi_t:=\rho_t(X)$ for all $t\in[0,T]$ and some $X\in\X_T$.
    Thus, $(\pi,Z^\pi)$ is the solution to the BSDE with parameters $(g,T,X)$ and, in particular, the relations~\eqref{EQ:pi_BSDE_relation} hold.
    Then, for every $0\leq s\leq t<t+\eps\leq T$, time-consistency and cash-additivity yield
    \[   
        \rho_s(\pi_{t+\eps}-\pi_t)
        \!=\!
        \rho_s\big(\rho_t(\pi_{t+\eps}-\pi_t)\big)
        \!=\!
        \rho_s\big(\rho_t(\rho_{t+\eps}(X))-\rho_t(X)\big)
        \!=\!
        \rho_s\big(\rho_t(X)-\rho_t(X)\big)
        \!=\!0.
    \]
    Hence,
    $
        \mathcal D^\rho_s\big(\r_t(X)\big)=0,
        \ 0\leq s\leq t<T.
    $
    Thus, when the inner process is the dynamic evaluation of a claim through the same BSDE-induced risk measure used as the outer evaluator, its resilience evaluation vanishes identically.
    In other words, $\pi=\rho(X)$ is a martingale under the nonlinear expectation $\rho$, i.e.,\ $\rho_s(\pi_t)=\pi_s$ for $s\leq t$, in the same way as discounted price processes are martingales under the pricing measure, and a martingale has no local tendency to recover or deteriorate.
    The entropic Example~\ref{EX:entropy} of Section~\ref{SEC:EX:convex} recovers this situation.
\end{remark}

\subsection{First interpretations}
\label{SEC:res_eval_interpretation}
The following observations clarify the meaning of Definition~\ref{DEF:res_eval} before imposing the BSDE assumptions.

\begin{itemize}[leftmargin=1em]
    \item
    \label{IT:res_interp_risk_rate}
    \textbf{Risk rate of a future change.}
        The random variable $\rho_s(\pi_{t+\eps}-\pi_t)$ is the risk, assessed at time $s$, of the future change in exposure generated over $[t,t+\eps]$.
        The rescaling factor $\eps^{-1}$ converts this risk into a rate per unit of time.
        Thus $\mathcal D^\rho_s\pi_t$ is the per-unit-time present risk of an infinitesimal future change in $\pi$.\footnote{This differs from the rate of change of the risk level $u\mapsto\rho_s(\pi_u)$, which for $s<t$ also depends on the loss already accumulated at time $t$: the resilience evaluation isolates the risk generated by the future evolution of $\pi$, independently of its current level.}

    \item
    \label{IT:res_interp_sign}
    \textbf{Sign interpretation: recovery versus deterioration.}
        Under the increasing loss convention of Definition~\ref{DEF:dynamic_risk_measure}, $\pi$ represents an exposure, liability, or risk assessment for which an increase is adverse.
        Therefore,  $\mathcal D^\rho_s\pi_t$  inherits the same reading, and its sign carries a direct financial meaning:
        \[
            \mathcal D^\rho_s\pi_t<0
            \ \Longleftrightarrow\
            \text{\!\emph{local prospect of recovery};} 
            \quad
            \mathcal D^\rho_s\pi_t>0
            \ \Longleftrightarrow\
            \text{\!\emph{local prospect of deterioration}.}
        \]
        The sign is that of a risk-adjusted rate. 
        For $\rho_s=\E[\,\cdot\,|\F_s]$ it is the sign of the expected drift; for a risk-averse $\rho$, a positive value may also be due to the volatility of $\pi$ alone: a position with no expected deterioration can be read as deteriorating because its future changes are uncertain.
        
    \item
    \label{IT:res_interp_curve}
    \textbf{Derivative in the exposure duration.}
        The idea in the first bullet can be formalized as follows. 
        For fixed $0\leq s\leq t<T$, set $\varphi_{s,t}(\eps):=\rho_s(\pi_{t+\eps}-\pi_t)$ for sufficiently small $\eps\geq 0$.
        Normalization gives $\varphi_{s,t}(0)=0$.
        The map $\varphi_{s,t}$ gives the present risk of a future exposure to an increment in $\pi$ as a function of the duration of this exposure.
        The resilience evaluation is precisely the right derivative $\varphi'_{s,t}(0)$ in $L^1$, whenever it exists.
        
    \item
    \label{IT:res_interp_driver}
    \textbf{Extension of the representation formula for BSDE drivers.}
        Suppose that $\rho$ is induced by a driver $g:\Omega\times[0,T]\times\R^m\to\R$.
        Representation results for BSDE generators yield
        \begin{equation}
        \label{EQ:Jiang_driver_representation}
            g(s,z)
            =
            L^1\text{ \!-\!}\lim_{\eps\to0^+}
            \frac1\eps
            \rho_s\big(z\cdot(W_{s+\eps}-W_s)\big), \qquad \, \ell_1\text{-a.e. } s\in[0,T), \ z\in\R^m;
        \end{equation}
        see \cite[Lemma~2.1]{Jiang_2008_Convexity_translation_invariance_subadditivity_g-expectations_related_risk_measures} in the Lipschitz case and \cite{Ma+Yao_2010_Quadratic_g-evaluationsexpectations_related_analysis,Zheng+Li_2018_Representation_theorem_generators_quadratic_BSDEs} in the quadratic case.
        This limit corresponds to the special setting $s = t$, $b = 0$, and $\s = z$ in Definition~\ref{DEF:res_eval}.
        Therefore, the resilience evaluation can be viewed as a generalization of this representation formula; see Remark~\ref{REM:s=t} and Theorem~\ref{TH:resilience_convex}.
        
    \item
    \label{IT:res_interp_present}
    \textbf{First-order expansion at $\mathbf{s=t}$.}
        Normalization and cash-additivity yield $\r_t(\pi_t)=\pi_t$ and $\r_t(\pi_{t+\eps}-\pi_t)=\r_t(\pi_{t+\eps})-\pi_t$.
        Therefore,
        \begin{equation}
        \label{EQ:s=t}
            \mathcal D^\rho_t\pi_t = L^1\text{ \!-\!}\lim_{\eps\to0^+} \frac1\eps\big[\rho_t(\pi_{t+\eps})-\rho_t(\pi_t)\big], 
        \end{equation}
        whenever either limit exists. 
        Thus, $\mathcal D^\r_t\pi_t$ is also the right derivative at $t$ of $u\mapsto \r_t(\pi_u)$.
        It follows that we can write a first-order expansion in time: $\r_t(\pi_{t+\eps})=\pi_t+\eps\,\mathcal D^\r_t\pi_t+o(\eps)$, as $\eps\to 0^+$, where the $o$ is intended in $L^1$, or in $\P$-probability.
        Therefore, $\mathcal D^\rho_t\pi_t$ serves as a proxy for the computation of the instantaneous change in the risk associated with $\pi_t$.

    \item
    \label{IT:res_special_coherent} 
    \textbf{Risk of the rate under positive homogeneity.} 
        If $\rho$ is positively homogeneous, then 
        \[ 
            \mathcal D^\rho_s\pi_t = L^1\text{ \!-\!}\lim_{\eps\to0^+} \rho_s\left(\frac{\pi_{t+\eps}-\pi_t}{\eps}\right), 
        \] whenever either limit exists. 
        In this case, the resilience evaluation can be read as the present risk of the future infinitesimal rate of change of $\pi$.
\end{itemize}

Further interpretations of the resilience evaluation, specific to the BSDE representation, are postponed to Section \ref{SEC:representation_interpretation}.
We next establish the variance and entropic benchmarks.

\subsection{Conditional variance}
\label{SEC:variance}
We momentarily step outside the framework of Definition~\ref{DEF:dynamic_risk_measure}.
Although the conditional variance defines a normalized and convex map $\V(\,\cdot\,|\F_s):L^\infty(\F_T)\longrightarrow L^\infty(\F_s)$, it is neither monotone nor cash-additive, being instead cash-invariant: $\V(X+m|\F_s)=\V(X|\F_s)$, for $X\in L^\infty(\F_T)$, $m\in L^\infty(\F_s)$.
Nevertheless, the infinitesimal scaling underlying Definition~\ref{DEF:res_eval} remains meaningful and isolates information that is not captured by the conditional mean, namely the local quadratic variation of the process.

This observation is consistent with the interpretation in
\cite[Section~7.2]{ElKaroui+Ravanelli_2009_Cash_sub-additive_risk_measures_interest_rate_ambiguity}.
For a solution $(Y,Z)$ to a BSDE, $|Z|$ represents the local volatility of $Y$: letting $g$ be the driver of the BSDE
\begin{align}
    \V(\!\d Y_t|\F_t)
    :=\E\left[\left.(\!\d Y_t-\E[\!\d Y_t|\F_t])^2\right|\F_t\right]
    =\E\left[\left.(Z_t\cdot\d W_t)^2\right|\F_t\right]
    =|Z_t|^2\d t,
\end{align}
where we used $\E[\!\d Y_t|\F_t]=\E[-g(t,Y_t,Z_t)\d t + Z_t\cdot\d W_t|\F_t]=-g(t,Y_t,Z_t)\d t$.
We now formalize this result for a general It\^o process as in \eqref{EQ:pi_Ito}.
\begin{theorem}
\label{TH:resilience_Variance}
    Assume that $b\in L^2_T$.
    For $\ell_1$-a.e. $t\in[0,T)$ and all $s\in[0,t]$, we have:
    \[
        \frac1\eps\V(\pi_{t+\eps}-\pi_t|\F_s)\longrightarrow\E\left[\left.|\s_t|^2\right|\F_s\right], \qquad \text{in }L^{1} .
    \]
\end{theorem}
\begin{proof}
    Fix $t\in[0,T)$, $\eps>0$ such that $t+\eps\leq T$, and $s\in[0,t]$.
    Set $A_\eps:=\int_t^{t+\eps}b_r\d r$ and $M_\eps:=\int_t^{t+\eps}\s_r\cdot\d W_r$, so that $\pi_{t+\eps}-\pi_t=A_\eps+M_\eps$.
    The martingale property gives $\E[M_\eps|\F_s]=0$. 
    Hence,
    \[ 
    \begin{aligned} 
        \frac1\eps\V(\pi_{t+\eps}-\pi_t|\F_s)
        :=\,&\frac1\eps\E\left[\left.\left(\pi_{t+\eps}-\pi_{t}\right)^2\right|\F_s\right]-
        \frac1\eps\big(\E\left[\pi_{t+\eps}-\pi_t|\F_s\right]\big)^2\\
        =&\frac1\eps\E[A_\eps^2|\F_s] 
        -\frac1\eps\big(\E[A_\eps|\F_s]\big)^2 
        +\frac2\eps\E[A_\eps M_\eps|\F_s] 
        +\frac1\eps\E[M_\eps^2|\F_s]. 
    \end{aligned} 
    \]
    By Lemma~\ref{LEM:integral_average}, for $\ell_1$-a.e. $t\in[0,T)$, $B_\eps:=\frac1\eps A_\eps \longrightarrow b_t$ in $L^2$.
    In particular, $(B_\eps)_\eps$ is bounded in $L^2$, and therefore 
    \[
        \left\| \frac1\eps\E[A_\eps^2|\F_s] \right\|_{L^1} = \eps\,\E\big[B^2_\eps\big] 
        \longrightarrow0, 
        \qquad 
        \left\| \frac1\eps\big(\E[A_\eps|\F_s]\big)^2 \right\|_{L^1} \leq \eps\,\E\big[B^2_\eps\big] 
        \longrightarrow0, 
    \] 
    where we used the conditional Jensen's inequality.
    Moreover, by conditional Jensen's inequality, Cauchy--Schwarz, and It\^o's isometry, 
    \begin{align} 
        \left\| \frac1\eps\E[A_\eps M_\eps|\F_s] \right\|_{L^1} 
        \leq \E[|B_\eps M_\eps|]
        \leq \|B_\eps\|_{L^2} \left( \E\left[\int_t^{t+\eps}|\s_r|^2\d r\right] \right)^{1/2} \longrightarrow0. 
    \end{align}
    Finally, since $M_\eps^2-\int_t^{t+\eps}|\s_r|^2\d r$ is a martingale, \[ 
        \frac1\eps\E[M_\eps^2|\F_s] = \E\left[ \left. \frac1\eps\int_t^{t+\eps}|\s_r|^2\d r \right|\F_s \right]
        \longrightarrow 
        \E\left[\left.|\s_t|^2\right|\F_s\right] 
        \qquad\text{in }L^1,
    \] 
    where the convergence is due to Lemma~\ref{LEM:integral_average} because $|\s|^2\in L^1_T$.
    Combining the preceding convergences proves the claim. 
\end{proof}

Because $\mathbb V(\,\cdot\,|\F_s)$ is not monotone, the sign interpretation of Section~\ref{SEC:res_eval_interpretation} does not apply to it directly: by Theorem~\ref{TH:resilience_Variance}, the limit is always non-negative, $\E[|\s_t|^2|\F_s]\geq0$, and therefore carries no directional information about recovery or deterioration. 
It should instead be read purely in terms of magnitude: the larger $|\sigma_t|$, the more dispersed, hence the less safe, the future increment of $\pi$ is expected to be, irrespective of its expected direction. 
This is one example of the information not captured by the resilience rate of Example~\ref{EX:resilience_rate}.

\subsection{Entropic risk measure}
\label{SEC:entropic_RM}
For a (fixed) risk-aversion coefficient $\gamma>0$, we define the entropic risk measure (e.g., \cite{Barrieu+ElKaroui_2009_Pricing_hedging_optimally_designing_derivatives_minimization_risk_measures,Laeven+Stadje_2013_Entropy_coherent_entropy_convex_measures_risk}) as
\[
    \mathfrak e^\gamma_s(X)=\frac1\gamma\log\E\left[\left.e^{\gamma X}\right|\F_s\right], \qquad s\in[0,T],
\]
which is finite $\P$-a.s.\ for every real-valued $\F_T$-measurable random variable $X$ such that $\E[e^{\gamma X}]<+\infty$.
This natural domain need not be a vector space and the outputs need not be integrable, so $\mathfrak e^\gamma$ is not necessarily a dynamic risk measure in the sense of Definition~\ref{DEF:dynamic_risk_measure}. 
Nevertheless, the limiting procedure underlying Definition~\ref{DEF:res_eval} remains meaningful.
Accordingly, whenever the quantities involved are well-defined in  $L^1(\F_s)$, we call
\[
    \mathcal D^{\mathfrak e^\gamma}_s\pi_t
    :=
    L^1\text{-}\lim_{\eps\to0^+}
    \frac1\eps
    \mathfrak e^\gamma_s(\pi_{t+\eps}-\pi_t)
\]
the entropic resilience evaluation of $\pi_t$ at time $s$.

\begin{theorem}
\label{TH:resilience_entropy}
    Let $p\in(2,+\infty)$ and $\g>0$.
    Assume that $b,|\s|\in L^p_T$, and that there exists $q>p/(p-2)$ such that
    \begin{equation}
    \label{EQ:TH:resilience_entropy:HP_expon_integ}
        \E\left[\exp\Big(4q\g\int_0^T|b_r|\d r \Big)\right]<+\infty, \qquad 
        \E\left[\exp\Big(32q^2\g^2\int_0^T|\s_r|^2\d r \Big)\right]<+\infty.
    \end{equation}
    Then, for $\ell_1$-a.e. $t\in[0,T)$ and all $s\in[0,t]$, the entropic resilience evaluation of $\pi_t$ at time $s$ is well-defined in $L^1$ and given by the formula:
    \begin{equation}
    \label{EQ:formula_res_entropy}
        \mathcal D^{\mathfrak e^\gamma}_s\pi_t=\E[b_t|\F_s]+\frac\gamma2\E\big[|\s_t|^2\big|\F_s\big].
    \end{equation}
\end{theorem}
A trivially sufficient assumption for~\eqref{EQ:TH:resilience_entropy:HP_expon_integ} is $b,|\s|\in L^\infty_T$.

The formula \eqref{EQ:formula_res_entropy} combines the previous two benchmarks. 
The drift term contributes through the sign of $\E[b_t|\F_s]$: positively (deterioration) where $\E[b_t|\F_s]>0$, negatively (recovery) where $\E[b_t|\F_s]<0$, as in Example~\ref{EX:resilience_rate}.
The variance term, by Theorem~\ref{TH:resilience_Variance}, is always non-negative:
local volatility can only push $\mathcal D^{\mathfrak e^\gamma}_s\pi_t$ towards deterioration, never towards recovery. 
The risk-aversion coefficient $\gamma$ calibrates the relative weight of the two effects. 
For small values of $\gamma$, the direction of the drift matters, as for a risk-neutral evaluator. 
For large $\gamma$, the variance term dominates: on the event $\{\E[|\s_t|^2|\F_s]>0\}$, we have $\mathcal D^{\mathfrak e^\gamma}_s\pi_t>0$ as soon as $\gamma>2\,(\E[b_t|\F_s])^-/\E[|\s_t|^2|\F_s]$, provided that~\eqref{EQ:TH:resilience_entropy:HP_expon_integ} holds for such $\gamma$ (e.g.,\ when $b,|\s|\in L^\infty_T$).
An increasingly risk-averse investor eventually reads any position with non-degenerate volatility as deteriorating, on account of its uncertainty alone.

To prove the previous theorem, we first need a preliminary general result, the proof of which is provided in Appendix~\ref{SEC:app_taylor}.
\begin{proposition}
\label{PROP:Taylor_exp}
    Let $(X_n)_{n\in\N}$ be a sequence of random variables satisfying the following assumptions.
    \begin{enumerate}[label=(\roman*)]
        \item\label{IT:PROP:Taylor_exp:conv} There exists $p>2$ such that $X_n\longrightarrow0$ in $L^p$, as $n\to+\infty$.
        \item\label{IT:PROP:Taylor_exp:unif_bound} There exist $q>p/(p-2)$ and $K\geq 0$ such that $\E\left[e^{2q|X_n|}\right]\leq K$ for all $n\in\N$.
    \end{enumerate} 
    Then for any sub-$\s$-algebra $\mathcal G\subseteq \F$, we have 
    \[
        \log\E\left[\left.e^{X_n}\right|\mathcal G\right] = \E[X_n|\mathcal G] + \frac12\V(X_n|\mathcal G)+R_n, \qquad\forall\,n\in\N,
    \]
    where $R_n\in L^1(\mathcal G)$, for $n\in\N$, and  $\|R_n\|_{L^1}=o\left(\|X_n\|^2_{L^p}\right)$ as $n\to\infty$.
\end{proposition}

\begin{proof}[\textbf{Proof of Theorem~\ref{TH:resilience_entropy}.}]
    Fix $t\in[0,T)$ in a full $\ell_1$-measure set such that the conclusions of Example~\ref{EX:resilience_rate}, Theorem~\ref{TH:resilience_Variance}, and Lemma \ref{LEM:integral_average} applied to $r\mapsto\E[|b_r|^p]$, $r\mapsto\E[|\s_r|^p]$ hold at $t$.
    Let $(\eps_n)_{n\in\N}$ be an arbitrary positive infinitesimal sequence (such that $t+\eps_n\leq T$ for all $n\in\N$).
    For $n\in\N$ and $h\in[0,T-t]$, set
    \[
        B_n:=\int_t^{t+\eps_n}b_r\d r,
        \qquad
        \widehat M_h:=\int_t^{t+h}\s_r\cdot\d W_r,
        \qquad 
        M_n:=\widehat M_{\eps_n},
    \]
    so that $X_n:=B_n+M_n$ and $\widehat M$ is a continuous local martingale, starting from $0$, with respect to the shifted filtration $(\F_{t+h})_{h\in[0,T-t]}$.
    Jensen's inequality and the Burkholder--Davis--Gundy inequality yield
    \[
    \begin{aligned}
        \frac{\|X_n\|_{L^p}^2}{\eps_n}
        \leq
        2\eps_n
        \left(
            \frac1{\eps_n}
            \int_t^{t+\eps_n}\E[|b_r|^p]\d r
        \right)^{2/p}
        +2C_p
        \left(
            \frac1{\eps_n}
            \int_t^{t+\eps_n}\E[|\s_r|^p]\d r
        \right)^{2/p}.
    \end{aligned}
    \]
    Hence, $X_n\longrightarrow0\quad\text{in }L^p$ and $\sup_n\eps_n^{-1}\|X_n\|_{L^p}^2<+\infty$.
    
    We next verify the exponential-integrability assumption of Proposition~\ref{PROP:Taylor_exp}.

    As recalled in Section~\ref{SEC:funct_spaces}, for any continuous local martingale $N$ with $N_0=0$, $\mathcal E(N)$ is a supermartingale, hence $\E[\mathcal E(N)]\leq1$ at all times.
    Combining this fact with the Cauchy--Schwarz inequality, we have for any $\l\in\R$, everywhere in time:
    \begin{equation}
    \label{EQ:TH:entropy:stoch_expon_bound_begin}
    \begin{aligned}
        \E\left[e^{\l N}\right]
        \leq \left(\E\left[e^{2\l N-2\l^2 \langle N\rangle}\right]\right)^{1/2}\left(\E\left[e^{2\l^2 \langle N\rangle}\right]\right)^{1/2}
        \leq \left(\E\left[e^{2\l^2 \langle N\rangle}\right]\right)^{1/2},
    \end{aligned}
    \end{equation}
    where the first factor in the middle  equals $\left(\E\left[\mathcal E(2\l N)\right]\right)^{1/2}\leq 1$.
    Using the trivial inequality $e^{|x|}\leq e^x+e^{-x}$, for $x\in\R$, we infer
    \begin{equation}
    \label{EQ:TH:entropy:stoch_expon_bound}
        \E\left[e^{\l |N|}\right] 
        \leq \E\left[e^{\l N}\right] + \E\left[e^{-\l N}\right]
        \leq 2\left(\E\left[e^{2\l^2 \langle N\rangle}\right]\right)^{1/2}, \qquad \forall\,\l\in\R.
    \end{equation}
    Applying this estimate with $N=\widehat M$ and $\l=4q\g$, we obtain
    \[
    \begin{aligned}
        \E[e^{2q\g|X_n|}]
        &\leq
        \left(\E[e^{4q\g|B_n|}]\right)^{1/2}
        \left(\E[e^{4q\g|M_n|}]\right)^{1/2} \\
        &\leq
        \sqrt2
        \left(
            \E\left[
                e^{4q\g\int_0^T|b_r|\d r}
            \right]
        \right)^{1/2}
        \left(
            \E\left[
                e^{32q^2\g^2\int_0^T|\s_r|^2\d r}
            \right]
        \right)^{1/4}.
    \end{aligned}
    \]
    The right-hand side is independent of \(n\), so Proposition~\ref{PROP:Taylor_exp}, applied to \((\g X_n)_{n\in\N}\), gives, for every \(s\in[0,t]\),
    \[
        \log\E[e^{\g X_n}|\F_s]
        =
        \g\E[X_n|\F_s]
        +\frac{\g^2}{2}\V(X_n|\F_s)
        +R_{n,s},
    \]
    where $\|R_{n,s}\|_{L^1}
        =
        o(\|X_n\|_{L^p}^2)$.
    Dividing by \(\g\eps_n\), we obtain
    \[
        \frac1{\eps_n}\mathfrak e^\gamma_s(X_n)
        =
        \frac1{\eps_n}\E[X_n\mid\F_s]
        +
        \frac{\g}{2\eps_n}\V(X_n\mid\F_s)
        +
        \frac{R_{n,s}}{\g\eps_n}.
    \]
    The last term converges to \(0\) in \(L^1\), since $\eps_n^{-1}\|R_{n,s}\|_{L^1}
        =
        \eps_n^{-1}\|X_n\|_{L^p}^2\,o(1)$.
    The first two terms converge respectively to $\E[b_t\mid\F_s]$ and $\frac\g2\E[|\s_t|^2\mid\F_s]$ by Example~\ref{EX:resilience_rate} and Theorem~\ref{TH:resilience_Variance}.
    The conclusion follows from the arbitrariness of $(\eps_n)_{n\in\N}$.
\end{proof}

\subsection{BSDE-induced convex risk measures}
\label{SEC:resil_convex}
On $L^{\mathrm{exp}}(\F_T)$, the entropic risk measure is induced by the quadratic BSDE driver $z\mapsto \frac{\gamma}{2}|z|^2$. 
The preceding subsection shows that this driver enters the resilience evaluation as an adjustment to the local volatility of $\pi$. 
This motivates considering, much more generally, an outer normalized, cash-additive, and convex dynamic risk measure induced by a family of BSDEs.

\subsubsection{Convex drivers and dual representation}
Throughout this section, we fix a driver $g$  that satisfies either of the following conditions.
\begin{enumerate}[label=$\mathbf{(C_L)}$]
    \item 
    \label{IT:g_lip}
        Let $g:\Om\times[0,T]\times\R^m\to \R$ be a driver such that
        \begin{itemize}[label=\scriptsize$\bullet$, noitemsep,leftmargin=2em,topsep=0pt]
            \item $g$ is $\P\otimes\ell_1$-a.e. convex in $z\in\R^m$.
            \item $g(\,\cdot\,,0)=0$, $\P\otimes\ell_1$-a.e.
            \item There exists $L>0$ such that $g$ is $\P\otimes\ell_1$-a.e. $L$-Lipschitz continuous in $z\in\R^m$.
        \end{itemize}
\end{enumerate}
\begin{enumerate}[label=$\mathbf{(C_Q)}$]
    \item 
    \label{IT:g_quad}
        Let $g:\Om\times[0,T]\times\R^m\to \R$ be a driver such that
        \begin{itemize}[label=\scriptsize$\bullet$, noitemsep,leftmargin=2em,topsep=0pt]
            \item $g$ is $\P\otimes\ell_1$-a.e. convex in $z\in\R^m$.
            \item $g(\,\cdot\,,0)=0$, $\P\otimes\ell_1$-a.e.
            \item There exists $k> 0$ such that $\P\otimes\ell_1$-a.e., for all $z\in\R^m$: $|g(\,\cdot\,,z)|\leq k(1+|z|^2)$.
        \end{itemize}
\end{enumerate}

We first present a technical lemma that provides the predictable selectors $g^\ast,\Pi,\Theta$ for the convex conjugate of $g$, for the orthogonal projector onto $\partial_z g(\,\cdot\,,0)$, and for the subdifferential $\partial_zg(\,\cdot\,,z)$, respectively.
The proof is provided in Appendix~\ref{SEC:app_duality}.

\begin{lemma}
\label{LEM:g*}
    There exist  $\mathcal P\otimes\mathscr B(\R^m)$-measurable maps 
    \[
        g^\ast:\Om\times[0,T]\times\R^m\to \R\cup\{+\infty\},
        \qquad \qquad  
        \Pi,\Theta:\Om\times[0,T]\times\R^m\to \R^m
    \]
    such that the following properties hold $\P\otimes\ell_1$-a.e.: 
    \begin{enumerate}[label=(\roman*), noitemsep,leftmargin=2em]
        \item 
        \label{IT:LEM:g*:def}
            The map $q\mapsto g^\ast(\,\cdot\,, q)$ is the convex conjugate of $z\mapsto g(\,\cdot\,,z)$.
            Moreover:
            \begin{enumerate}[label=(\alph*), noitemsep, leftmargin=1.5em]
                \item   
                \label{IT:LEM:g*:prop_nonnega}
                    It is proper, convex, lower semi-continuous and non-negative.
                    In addition, 
                    \begin{itemize}[leftmargin=*, label=\small$\bullet$, topsep=0pt, noitemsep]
                        \item if $g$ satisfies~\itemref{IT:g_quad}, then $\ds g^\ast(\,\cdot\,,q)\geq \frac{|q|^2}{4k}-k$ for all $q\in\R^m$.
                        \item if $g$ satisfies~\itemref{IT:g_lip}, then $\ds g^\ast(\,\cdot\,,q)=+\infty$ for all $|q|>L$.
                    \end{itemize}

                \item 
                \label{IT:LEM:g*:g=g**}
                    It holds that $g^{\ast\ast}(\,\cdot\,,z) =g(\,\cdot\,,z)$ for all $z\in\R^m$.

                \item 
                \label{IT:LEM:g*:subgrad_0}
                    $\partial_z g(\,\cdot\,,0)$ is a non-empty, closed, convex, $\P\otimes\ell_1$-uniformly bounded set, and it coincides with the zero-set of $q\mapsto g^\ast(\,\cdot\,,q)$.
                    In particular, for any $q\in \partial_z g(\,\cdot\,,0)$, we have $|q|\leq L$ if $g$ satisfies~\itemref{IT:g_lip} and $|q|\leq 2k$ if $g$ satisfies~\itemref{IT:g_quad}. 
            \end{enumerate}

        \item 
        \label{IT:LEM:g*:projector}
            The map $q\mapsto \Pi(\,\cdot\,,q)$ is the orthogonal projector onto $\partial_z g (\,\cdot\,,0)$, namely:
            \begin{equation}
            \label{EQ:LEM:g*:def_proj}
                \Pi(\,\cdot\,,q)\in \partial_z g(\,\cdot\,,0),
                \qquad
                \big|q-\Pi(\,\cdot\,,q)\big|
                = \operatorname{dist}\big(q, \partial_z g(\,\cdot\,,0)\big).
            \end{equation}

        \item
        \label{IT:LEM:g*:selector_subgrad}
            For $z\in\R^m$, we have $\Theta(\,\cdot\,,z)\in\partial_z g(\,\cdot\,,z)$, hence $g(\,\cdot\,,z) = \Theta(\,\cdot\,,z)\cdot z - g^\ast\big(\,\cdot\,,\Theta(\,\cdot\,,z)\big)$.
            Moreover, 
            $|\Theta(\,\cdot\,,z)|\leq L$ if $g$ satisfies~\itemref{IT:g_lip}, and $|\Theta(\,\cdot\,,z)|\leq 5k(1+|z|)$ if $g$ satisfies~\itemref{IT:g_quad}.
    \end{enumerate}
\end{lemma}

The following lemma gathers some integrability estimates under the probability measure $\Q^\mu$ defined in~\eqref{EQ:equivalent_mg_measure}.
The proof is provided in Appendix~\ref{SEC:app_duality}.
\begin{lemma}
\label{LEM:cond_expec_convex}
    For $\mu\in\BMO_T$ and $0\leq a\leq b\leq T$, define
    \[
        \mathcal E_{a,b}^\mu
        := 
        \frac{\mathcal E_b^\mu}{\mathcal E_a^\mu}
        =
        \exp\bigg(
            \int_a^b\mu_r\cdot\d W_r
            -
            \frac12\int_a^b|\mu_r|^2\d r
        \bigg).
    \]
    \begin{enumerate}[label=(\roman*), noitemsep]
        \item
        \label{IT:LEM:cond_expec_convex:Lexp_BMO}
            For every fixed $\mu\in\BMO_T$, and $p\in[1,+\infty)$, we have $L^{\mathrm{exp}}
                \subset
                L^{\mathrm{exp}}(\Q^\mu)
                \subset
                L^p(\Q^\mu)$.
        
        \item
        \label{IT:LEM:cond_expec_convex:stoch_exponen}
            Let $K>0$ and let $\mathcal U\subseteq\BMO_T$ be such that $|\mu|\leq K$, $\P\otimes\ell_1\text{-a.e.}$, for all $\mu\in\mathcal U$.
            Then, for every $p>1$, and $0\leq a\leq b\leq T$,
            \[
                \E\big[(\mathcal E_{a,b}^\mu)^p\big|\F_a\big]
                \leq
                \exp\Big(\frac p2(p-1)K^2T\Big),
                \qquad \forall\,\mu\in\mathcal U.
            \]

        \item
        \label{IT:LEM:cond_expec_convex:embedding}
            Under the assumptions of~\itemref{IT:LEM:cond_expec_convex:stoch_exponen},
            for every $1\leq u<p$, and every $s\in[0,T]$,
            \[
                \big(\E_\mu[|X|^u|\F_s]\big)^{1/u}
                \leq
                \exp\bigg(\frac{K^2T}{2(p-u)}\bigg)
                \big(\E[|X|^p|\F_s]\big)^{1/p},
                \qquad
                \forall\,\mu\in\mathcal U,\ X\in L^p.
            \]
            In particular, $L^p$ is continuously embedded in $L^u(\Q^\mu)$,
            uniformly in $\mu\in\mathcal U$.
    \end{enumerate}
\end{lemma}

We now turn to the properties of the dynamic risk measure $\r$ induced by the driver $g$, and, in particular, to its dual representation.
The assertions of the following proposition are mainly standard.
The only point that needs a separate argument is the extension of the dual representation to unbounded terminal conditions with fixed dual class $\BMO_T$.
Related dual representations for unbounded terminal conditions are available in \cite[Theorems~3.10,~4.5,~4.6]{Drapeau+Kupper+RosazzaGianin+Tangpi_2016_Dual_representation_minimal_supersolutions_convex_BSDEs} and \cite[Theorem~3.2(i)]{Fan+Hu+Tang_2025_Unbounded_dynamic_concave_utilities_BSDEs}.
The former works with a fixed dual class characterized by essentially bounded Girsanov densities, which is in general not comparable with $\BMO_T$, whereas in the latter the admissible dual class may depend on the terminal condition.
Thus, neither result directly yields the representation over the fixed class $\BMO_T$ established below.
\begin{proposition}
\label{PROP:dual_representation}
    Assume that $g$ satisfies either~\itemref{IT:g_lip} or~\itemref{IT:g_quad}.
    Then the risk measure $\r$ induced by $g$ is normalized, cash-additive, convex and  enjoys the following dual representation for all $s\in[0,T]$: 
    \begin{equation}
    \label{EQ:dual_representation_convex}
        \r_s(X)=\esssup_{\mu\in\mathcal B}\E_{\mu}\bigg[X-\int_s^Tg^\ast(r,\mu_r)\d r\bigg|\F_s\bigg],\qquad 
        \forall\, X\in \mathcal X_T,
    \end{equation}
    where:
    \begin{itemize}[label=\scriptsize$\bullet$, noitemsep,topsep=0pt]
        \item 
            If $g$ satisfies~\ref{IT:g_lip}, we have $\mathcal X_T= L^2(\F_T)$ and $\mathcal B$ is the set of all predictable processes $\mu:\Om\times[0,T]\to\R^m$ such that $|\mu|\leq L$, $\P\otimes\ell_1$-a.e.
        \item 
            If $g$ satisfies~\ref{IT:g_quad}, we have $\mathcal X_T= L^{\mathrm{exp}}(\F_T)$ and
            $\mathcal B:=\BMO_T$.
    \end{itemize}
    The conditional expectation in~\eqref{EQ:dual_representation_convex}
    is understood in the extended sense as a $[-\infty,+\infty)$-valued random variable.
\end{proposition}

\subsubsection{Main representation theorem}
\label{SEC:main_result}

We introduce one further condition on $g^\ast$, which will be later required in Theorem~\ref{TH:resilience_convex} for some $0\leq s \leq t<T$.

\begin{enumerate}[label=\textup{(\textbf{H$_{s,t}$})}, noitemsep]
    \item
    \label{IT:GA}
        There is a continuous, strictly increasing, convex function $\psi:[0,+\infty)\to[0,+\infty)$, with $\psi(0)=0$, such that
        \begin{equation}
        \label{EQ:IT:GA}
            g^\ast(\,\cdot\,,q)
            \geq 
            \psi\big(\operatorname{dist}^2\big(q, \partial_zg (\,\cdot\,,0)\big)\big)
            \quad \forall\, q\in\R^m, 
            \quad 
            \P\otimes\ell_1\text{-a.e. on }\Om\times(s,t).
        \end{equation}
\end{enumerate}

Recall from Lemma~\ref{LEM:g*}\itemref{IT:LEM:g*:def}\itemref{IT:LEM:g*:subgrad_0} that $\partial_zg(\,\cdot\,,0)$ coincides with the zero-set of $g^\ast$ in the dual representation~\eqref{EQ:dual_representation_convex}. 
\begin{remark}
\label{REM:GA}
    We discuss the role and implications of the assumption~\itemref{IT:GA}.
    Examples of drivers satisfying it appear in Section~\ref{SEC:examples}; item~\itemref{IT:REM:GA:mild} below gives a driver for which it fails.
    
    \begin{enumerate}[label=($\roman*$), leftmargin=1.5em, itemsep=3pt, topsep=3pt]
    \item 
    \label{IT:REM:GA:minimal}
        The regularity imposed on $\psi$ causes no loss of generality. 
        The essential requirement is only the existence of a convex function $\psi_0:[0,+\infty)\to[0,+\infty]$ such that $\psi_0(x)=0$ if and only if $x=0$, for which~\eqref{EQ:IT:GA} holds with $\psi_0$ in place of $\psi$.\par\smallskip
        Indeed, define the generalized inverse $\phi_0(y) := \sup\{x\geq0:\psi_0(x)\leq y\}$, for $y\geq0$.
        Then $\phi_0$ is finite, non-decreasing, concave, with $\phi_0(0)=0$ and $\phi_0(y)\to0$ as $y\to0^+$.
        The function $\phi(y):=\phi_0(y)+ y$ is continuous, strictly increasing, concave and maps $[0,+\infty)$ onto itself. 
        Consequently, its inverse $\psi:=\phi^{-1}$ satisfies the regularities stated in~\itemref{IT:GA} and, in addition, $\psi\leq\psi_0$.
        If~\eqref{EQ:IT:GA} holds for $\psi_0$, then it holds for $\psi$ too.

    \item
    \label{IT:REM:GA:bounded}
        The condition is binding only on bounded neighborhoods of $\partial_z g(\,\cdot\,,0)$.
        Under~\itemref{IT:g_lip}, \eqref{EQ:IT:GA} is automatic for $|q|>R:=L$, since there $g^\ast=+\infty$; in this case we set $\psi(x):=cx$ with an arbitrary $c>0$.
        Under~\itemref{IT:g_quad}, Lemma~\ref{LEM:g*}\itemref{IT:LEM:g*:def}\itemref{IT:LEM:g*:prop_nonnega} gives $g^\ast(\,\cdot\,,q)\geq |q|^2/(4k)-k$ and $\operatorname{dist}(q,\partial_zg(\,\cdot\,,0))\leq|q|+2k$; hence, for $\psi(x):=cx$ with $0<c<1/(4k)$, there exists $R>0$ such that $\psi\big(\operatorname{dist}^2(q,\partial_zg(\,\cdot\,,0))\big)\leq c(|q|+2k)^2\leq g^\ast(\,\cdot\,,q)$ for all $|q|>R$.
        Therefore, in both cases, it suffices to verify~\eqref{EQ:IT:GA} for $|q|\leq R$, with some $\psi_1$.
        Indeed, \eqref{EQ:IT:GA} then holds for all $q\in\R^m$ with the largest convex minorant $\psi_0$ of $\min\{\psi_1,\psi\}$.
        Moreover, $\psi_0$ vanishes only at $0$: for fixed $x>0$, since $\psi_1(y)/y$ is non-decreasing, the affine function $y\mapsto\delta(y-x/2)$, with $\delta:=\min\{2\psi_1(x/2)/x,\,c\}>0$, lies below $\min\{\psi_1,\psi\}$, so that $\psi_0(x)\geq\delta x/2>0$.
        Finally, $\psi_0$ can be regularized as in~\itemref{IT:REM:GA:minimal}.

    \item 
    \label{IT:REM:GA:s_neq_t}
        The assumption~\itemref{IT:GA} is relevant only when $s\neq t$; otherwise, it is tautological.

    \item 
    \label{IT:REM:GA:pos_hom}
        If $g$ is positively homogeneous,  then $g^\ast$ is $\P\otimes\ell_1$-a.e. equal to $0$ on $\partial_zg(\,\cdot\,,0)$, and to $+\infty$ otherwise.
        Hence~\itemref{IT:GA} is automatic, for instance with $\psi(x)=x$.

    \item 
    \label{IT:REM:GA:det_cons}
        If $g$ is deterministic and constant in time (in which case, we simply write $g:\R^m\to\R$), then the condition~\itemref{IT:GA} is satisfied for any $0\leq  s\leq  t < T$.
        A proof is provided in Appendix~\ref{SEC:app_determ_driver}.

    \item
    \label{IT:REM:GA:mild}
        Condition~\itemref{IT:GA} is quite mild. 
        The previous points show that, in order to violate it, we need to assume $s\neq t$, and $g$ has to be a convex, not positively homogeneous function depending on $r\in[0,T]$ or on $\om\in\Om$.
        For instance, let $g^\ast(r,q):=h(r)|q|$ for $|q|\leq1$ and $g^\ast(r,q):=+\infty$ for $|q|>1$, $r\in(s,t)$, where $h(r)\to0^+$ as $r\to t^-$.
        Then $\partial_zg(r,0)=\{0\}$ and a function $\psi$ as in the assumptions cannot exist. 
        If it existed, then the inequality $\psi(\operatorname{dist}^2(q, \partial_zg(r,0)))=\psi(|q|^2)\leq h(r)|q|$ would lead to a contradiction for any $q\neq 0$ by letting $r\to t^-$.
        A driver with this convex conjugate is $g(r,z)=(|z|-h(r))^+$, where, for instance, $h(r)=|t-r|$.
        
    \item
    \label{IT:REM:GA:concave}
        Let $\psi$ be as in~\itemref{IT:GA} and set $ \phi:=\psi^{-1}$.
        Then $\phi:[0,+\infty)\to[0,+\infty)$ is continuous, strictly increasing and concave, with $\phi(0)=0$, and condition~\eqref{EQ:IT:GA} is equivalently written as 
        \begin{equation} 
        \label{EQ:IT:GA_concave} 
            \operatorname{dist}^2 \big(q,\partial_zg(\,\cdot\,,0)\big) 
            \leq 
            \phi\big(g^\ast(\,\cdot\,,q)\big)
            \quad \forall\, q\in\R^m,
            \quad 
            \P\otimes\ell_1\text{-a.e. on }\Om\times(s,t).
        \end{equation}
        The formulation~\eqref{EQ:IT:GA} makes the growth requirement on $g^\ast$ transparent, whereas the equivalent concave formulation~\eqref{EQ:IT:GA_concave} is useful in the proof. 
        In particular, concavity and $\phi(0)=0$ imply  $\alpha\,\phi(x/\alpha) \leq (1\vee\alpha)\phi(x),$ for $x\geq0$ and $\alpha>0$.
    \end{enumerate}
\end{remark}

Having specified the  assumptions for  $\r$, we now state the main theorem.
\begin{theorem}
\label{TH:resilience_convex}
    Assume that $g,\pi$ satisfy either of the following:
    \begin{enumerate}[label=(\alph*), noitemsep]
        \item 
        \label{IT:TH:convex:lip}
            $g$ satisfies~\itemref{IT:g_lip} and $b\in L^2_T$;
        \item 
        \label{IT:TH:convex:quad}
            $g$ satisfies~\itemref{IT:g_quad} and $b,|\s|\in L^\infty_T$;
    \end{enumerate}
    and let $\r$ be the dynamic risk measure induced by $g$.
    There exists a Borel set $\mathcal T\subseteq [0,T)$ with full $\ell_1$-measure such that, if $t\in\mathcal T$, $s\in[0,t]$ and~\itemref{IT:GA} is satisfied, then 
    the resilience evaluation of $\pi_t$ through $\r_s$ is well-defined in $L^1$ and given by the formula:
    \begin{gather}
    \label{EQ:TH:resilience_convex:thesis}
        \mathcal D^\r_s\pi_t=\esssup_{\mu\in\partial^0 \!g}\E_{\mu}\left[\left.b_t+g(t,\s_t)\right|\F_s\right],\\
        \partial^0 \!g:=
        \big\{
            \mu:\Om\times[0,T]\to\R^m  \ \text{predictable}\,:\ \mu\in\partial_{z}g(\,\cdot\,,0)\ \P\otimes\ell_1\text{-a.e.}
        \big\}\subseteq\BMO_T,\quad 
    \label{EQ:A^0}
    \end{gather}
    where $\E_\mu$ is the expectation with respect to the probability measure $\Q^\mu$ defined by~\eqref{EQ:equivalent_mg_measure}.
    Moreover, under~\itemref{IT:TH:convex:lip},  $\mathcal D^\r_s\pi_t\in L^2$ with 
    \begin{equation}
    \label{EQ:TH:convex_lip_bounds_resilience}
        \|\mathcal D^\r_s\pi_t\|_{L^2}
        \leq e^{L^2T/2}\big(\|b_t\|_{L^2}+L\|\s_t\|_{L^2}\big),
    \end{equation}
    while under~\itemref{IT:TH:convex:quad},  $\mathcal D^\r_s\pi_t\in L^\infty$ with 
    \begin{equation}
    \label{EQ:TH:convex_quadratic_bounds_resilience}
        -\|b\|_{L^\infty_T}-2k\|\s\|_{L^\infty_T} 
        \leq \mathcal D^\r_s\pi_t
        \leq \|b\|_{L^\infty_T}+k\big(1+\|\s\|^2_{L^\infty_T} \big).
    \end{equation}
\end{theorem}

The proof of Theorem~\ref{TH:resilience_convex} is deferred to the end of this subsection, after we present its main consequences, the attainment result, and the interpretation of the representation formula.

    If $z\mapsto g(\om,t,z)$ is differentiable at $z=0$ for $\P\otimes\ell_1$-a.e. $(\om,t)$, then $\partial_zg(\,\cdot\,,0)=\{\nabla_z g(\,\cdot\,,0)\}$, so that $\partial^0 \!g$ consists of a single equivalence class.
    Consequently, the essential supremum in~\eqref{EQ:TH:resilience_convex:thesis} is attained at a predictable version $\mu$ of $\nabla_z g(\,\cdot\,,0)$ and reduces to the conditional expectation under $\Q^\mu$.

\begin{remark}
\label{REM:As_or_Ast}
    The representation in Theorem~\ref{TH:resilience_convex} can
    equivalently be indexed by three different classes of dual processes.
    More precisely, for $s<t$ and $Y_t:=b_t+g(t,\s_t)$, we have
    \begin{equation}
    \label{EQ:REM:As_or_Ast:equivalence}
        \esssup_{\mu\in\partial^0\!g}
            \E_\mu[Y_t\mid\F_s]
        =
        \esssup_{\mu\in\A_s}
            \E_\mu[Y_t\mid\F_s]
        =
        \esssup_{\theta\in\A^{\mathrm{loc}}_{s,t}}
            \E_\theta[Y_t\mid\F_s].
    \end{equation}
    The three classes play different roles.
    The global zero-penalty class $\partial^0\!g$ is the most natural one for stating the representation theorem.
    The normalized class
    \begin{equation}
    \label{EQ:A_s}
    \begin{aligned}
        \A_s := \big\{
            \mu:\Om\times[0,T]\to\R^m
            \text{ predictable}:
            \ &\mu=0
            \quad \P\otimes\ell_1\text{-a.e. on }\Om\times[0,s],\\
            &\mu\in\partial_zg(\,\cdot\,,0)
            \quad \P\otimes\ell_1\text{-a.e. on }\Om\times(s,T]
        \big\}
    \end{aligned}
    \end{equation}
    is convenient for the proof of the main theorem.
    By Lemma~\ref{LEM:g*}\itemref{IT:LEM:g*:def}\itemref{IT:LEM:g*:subgrad_0}, $\A_s\subseteq \mathcal B$.
    The local class
    \[
        \A^{\mathrm{loc}}_{s,t}
        :=
        \big\{
            \theta:\Om\times(s,t)\to\R^m
            \text{ predictable}:
            \ \theta\in\partial_zg(\,\cdot\,,0)
            \quad\P\otimes\ell_1\text{-a.e. on }\Om\times(s,t)
        \big\}
    \]
    is the minimal class determining the value of the $\esssup$, since the conditional expectation depends only on the dual process over the interval $(s,t)$.
    Indeed, let $\theta\in\A^{\mathrm{loc}}_{s,t}$, and denote by
    $\theta\1_{(s,t)}$ its extension by zero outside $(s,t)$.
    Define
    \[
        \widehat\theta^{\,s}
        :=
        \theta\1_{(s,t)}
        +
        \Pi(\,\cdot\,,0)\1_{[t,T]},
        \qquad
        \widehat\theta^{\,0}
        :=
        \Pi(\,\cdot\,,0)\1_{[0,s]}
        +
        \theta\1_{(s,t)}
        +
        \Pi(\,\cdot\,,0)\1_{[t,T]}.
    \]
    Items \itemref{IT:LEM:g*:def}\itemref{IT:LEM:g*:subgrad_0} and  \itemref{IT:LEM:g*:projector} in Lemma~\ref{LEM:g*} imply that $\widehat\theta^{\,s}\in\A_s$ and $\widehat\theta^{\,0}\in\partial^0\!g$.
    We may therefore define $\E_\theta[Y_t\mid\F_s]
        :=
        \E_{\widehat\theta^{\,s}}[Y_t\mid\F_s]$.
    Lemma~\ref{LEM:mu_vs_nu}\itemref{IT:LEM:mu_vs_nu:equality} applied with $a=t$ shows that the same value is obtained using $\widehat\theta^{\,0}$, or any other extension in $\mathcal B$ that coincides with $\theta$ on $\Om\times(s,t)$.
    Conversely, the restriction to $\Om\times(s,t)$ of every process in $\A_s$ or $\partial^0\!g$ belongs to $\A^{\mathrm{loc}}_{s,t}$.
    This proves~\eqref{EQ:REM:As_or_Ast:equivalence}.
\end{remark}

Assumption~\itemref{IT:GA} and Remark~\ref{REM:As_or_Ast} explain the zero-penalty class in Theorem~\ref{TH:resilience_convex}.
    After scaling by \(\eps^{-1}\), controls with a non-vanishing penalty on \((s,t)\) cannot contribute to the limit; condition~\eqref{EQ:IT:GA} then brings the relevant controls close to \(\partial_zg(\,\cdot\,,0)\) on that interval. 
    On \((t,t+\eps)\), the penalty instead combines with the volatility term to produce \(g(t,\s_t)\) through Fenchel duality.

\begin{remark}
\label{REM:s=t}
    Consistently with Remark~\ref{REM:GA}\itemref{IT:REM:GA:s_neq_t}, neither~\eqref{EQ:IT:GA} nor the essential supremum in~\eqref{EQ:TH:resilience_convex:thesis} is required when $s =t$. 
        Namely, it follows directly from the statement that, under either~\itemref{IT:TH:convex:lip} or~\itemref{IT:TH:convex:quad}, the resilience evaluation $\mathcal D_t^\r\pi_t$ is well-defined in $L^1$ for $\ell_1$-a.e. $t\in[0,T)$, and we have
        \[
            \mathcal D_t^\r\pi_t=b_t + g(t,\s_t), \qquad\text{ in }L^1, \quad\ell_1\text{-a.e. }t\in[0,T).
        \]
\end{remark}

By specializing the driver $g$ to the positively homogeneous case (thus inducing a coherent risk measure, see Section~\ref{SEC:BSDEs}), we obtain a more explicit representation for the resilience evaluation.
Its proof is given at the end of this subsection, after the proof of the main theorem.
\begin{corollary}[Coherent outer risk measure]
\label{COR:from_convex_to_coherent}
    Assume that $g$ satisfies~\itemref{IT:g_lip}, that $b\in L^2_T$, and that $g$ is $\P\otimes\ell_1$-a.e. positively homogeneous in $z$.
    Let $\r$ be the dynamic risk measure induced by $g$.
    Then the full $\ell_1$-measure Borel set $\mathcal T$ in Theorem~\ref{TH:resilience_convex} can be chosen so that, for every $t\in\mathcal T$ and $s\in[0,t]$, the resilience evaluation of $\pi_t$ through $\r_s$ is well-defined in $L^1$, belongs to $L^2$, and is given by
    \begin{gather}
    \label{EQ:COR:resilience_coherent}
        \mathcal D^\r_s\pi_t
        =
        \esssup_{\mu\in\partial^0_t \!g}
        \E_\mu\left[
            b_t+\mu_t\cdot\s_t
            \,\big|\,\F_s
        \right],\\
        \partial^0_t \!g
        :=
        \big\{
            \mu\in\partial^0 \!g \ : \ 
            \mu_t\in\partial_zg(t,0)
            \ \P\text{-a.s.}
        \big\}.
    \label{EQ:COR:admissible_set_A}
    \end{gather}
\end{corollary}

The condition imposed on $\mu_t$ in
\eqref{EQ:COR:admissible_set_A} only selects a predictable representative of $\mu$.
Indeed, changing a predictable process at the deterministic time $t$ does not alter its stochastic integral and hence does not alter the measure $\Q^\mu$.
Therefore, recalling Remark~\ref{REM:As_or_Ast}, the restriction of $\mu$ to $\Om\times(s,t)$ determines the zero-penalty scenario, whereas its value at $\Om\times\{t\}$ can be chosen independently to attain the maximum of the sublinear driver $g$.

\subsubsection{Attainment of the essential supremum}
\label{SEC:attainment}
This subsection is devoted to the attainment of the essential supremum in formula~\eqref{EQ:TH:resilience_convex:thesis}.
We first need a technical lemma to construct the support function of the zero-penalty set $\partial_z g(\,\cdot\,,0)$ together with a measurable maximizer.
The proof is provided in Appendix~\ref{SEC:app_duality}.
\begin{lemma}
\label{LEM:h}
    There exist $\mathcal P\otimes\mathscr B(\R^m)$-measurable maps 
    \[
        h:\Om\times[0,T]\times\R^m\to \R, 
        \qquad\qquad 
        \hat\mu:\Omega\times[0,T]\times\R^m\to\R^m,
    \]
    that satisfy the following properties $\P\otimes\ell_1$-a.e. 
    \begin{enumerate}[label=(\roman*), noitemsep]
        \item 
        \label{IT:LEM:h:def}
            $z\mapsto h(\,\cdot\,, z)$ is the support function of $\partial_z g (\,\cdot\,,0)$, namely
            \begin{equation}
            \label{EQ:LEM:h:def}
                h(\,\cdot\,, z)
                =\sup_{q\in \partial_z g (\,\cdot\,,0)}q\cdot z,
                \qquad \forall\, z\in\R^m.
            \end{equation}
        \item
        \label{IT:LEM:h:selector}
            For every $z\in\R^m$,
            \[
                \hat\mu(\,\cdot\,,z)\in\argmax_{q\in \partial_z g (\,\cdot\,,0)}q\cdot z,
            \]
            namely $\hat\mu(\,\cdot\,,z)\in \partial_z g(\,\cdot\,,0)$ and $h(\,\cdot\,,z)=\hat\mu(\,\cdot\,,z)\cdot z$.
        \item   
        \label{IT:LEM:h:lip}
            $z\mapsto h(\,\cdot\,,z)$ is sublinear and Lipschitz continuous, with Lipschitz constant $L$ if $g$ satisfies~\ref{IT:g_lip} and $2k$ if $g$ satisfies~\ref{IT:g_quad}.
    \end{enumerate}
\end{lemma}

\begin{theorem}
\label{TH:attainment}
    Assume the hypotheses of Theorem~\ref{TH:resilience_convex}, and let $\mathcal T$ be the full $\ell_1$-measure Borel set provided therein.
    Fix $t\in\mathcal T$ and $s\in[0,t]$ such that~\itemref{IT:GA} holds.
    Let $h$ be the support function of $\partial_zg(\,\cdot\,,0)$ introduced in Lemma~\ref{LEM:h}, and let $(V,\Xi)$ be the unique solution of the BSDE with parameters $(h,t,b_t+g(t,\s_t))$:
    \begin{equation}
    \label{EQ:BSDE_h}
        V_r
        = b_t+g(t,\s_t)+\int_r^t h(u,\Xi_u)\,\d u-\int_r^t \Xi_u\cdot\d W_u,
        \qquad r\in[0,t].
    \end{equation}
    Then,
    \begin{equation}
    \label{EQ:TH:attainment:thesis}
        \mathcal D^\r_s\pi_t
        =
        V_s
        =
        \essmax_{\mu\in\partial^0\!g}
        \E_\mu\left[b_t+g(t,\s_t)\mid\F_s\right]
        =
        \E_{\bar\mu}\!\left[
            b_t+g(t,\s_t)
            \,\middle|\,
            \F_s
        \right],
        \qquad \P\text{-a.s.},
    \end{equation}
    where $\bar\mu\in\partial^0\!g$ can be chosen such that $h(\,\cdot\,,\Xi)=\bar\mu\cdot\Xi$, $\P\otimes\ell_1$-a.e. on $\Om\times(s,t)$.
    
    \noindent If, in addition, the hypotheses of
    Corollary~\ref{COR:from_convex_to_coherent} hold and $\mathcal T$ is chosen as therein, then $\bar\mu$ can be chosen in $\partial_t^0\!g$ so as to attain the essential supremum in~\eqref{EQ:COR:resilience_coherent}: 
    \begin{equation} 
    \label{EQ:TH:attainment:coherent} 
        \mathcal D^\r_s\pi_t 
        = \E_{\bar\mu}\!\left[ b_t+\bar\mu_t\cdot\s_t |\F_s \right],
        \qquad \P\text{-a.s.} 
    \end{equation}
\end{theorem}

\begin{proof}
    Fix $t\in\mathcal T$ and $s\in[0,t]$ satisfying~\ref{IT:GA}, and set $Y_t:=b_t+g(t,\s_t)$.
    Under~\itemref{IT:TH:convex:lip}, we have $Y_t\in L^2(\F_t)$ for $t\in\mathcal T$;
    under~\itemref{IT:TH:convex:quad}, we have $Y_t\in L^\infty(\F_t)$.
    
    Let us consider the BSDE \eqref{EQ:BSDE_h} with parameters $(h,t,Y_t)$, where $h$ is the function introduced in Lemma~\ref{LEM:h}.
    Thanks to Lemma~\ref{LEM:h}\itemref{IT:LEM:h:lip}, this BSDE is well-posed and admits a unique solution $(V,\Xi)$ in the class $\mathcal S^2_t\times\mathcal H^2_t$ under~\itemref{IT:TH:convex:lip} and in $\mathcal S^\infty_t\times\BMO_t$ under~\itemref{IT:TH:convex:quad}.
    
    Let $\hat\mu$ be the measurable selector given by Lemma~\ref{LEM:h}\itemref{IT:LEM:h:selector}, and define
    \begin{equation}
    \label{EQ:TH:attainment:def_mu_ast}
        \bar\mu
        :=
        \hat\mu(\,\cdot\,,\Xi)\1_{(s,t]}
        +
        \hat\mu(\,\cdot\,,0)\1_{[0,s]\cup(t,T]}.
    \end{equation}
    {\emergencystretch=3em
    Since $\Xi$ is predictable and $\hat\mu$ is
    $\mathcal P\otimes\mathscr B(\R^m)$-measurable, $\bar\mu$ is predictable.
    Moreover, Lemma~\ref{LEM:h}\itemref{IT:LEM:h:selector} yields $\bar\mu\in\partial^0 \!g$.\par}
    
    Fix an arbitrary $\mu\in\partial^0 \!g$, and define $W^\mu_r:=W_r-\int_0^r \mu\d\ell_1$ for $r\in[0,T]$.
    Then $W^\mu$ is a Brownian motion under $\Q^\mu$.
    Evaluating~\eqref{EQ:BSDE_h} at $r=s$ and rewriting the stochastic integral in terms of $W^\mu$, we obtain
    \[
        V_s
        =
        Y_t
        +\int_s^t
            \big(h(u,\Xi_u)-\mu_u\cdot\Xi_u\big)\,\d u
        -\int_s^t\Xi_u\cdot\d W^\mu_u.
    \]
    Taking the $\Q^\mu$-conditional expectation with respect to $\F_s$ gives
    \begin{equation}
    \label{EQ:PROP:attainment:upperbound}
        V_s
        =
        \E_\mu\!\left[
            Y_t+
            \int_s^t
                \big(h(u,\Xi_u)-\mu_u\cdot\Xi_u\big)\,\d u
            \,\bigg|\,\F_s
        \right].
    \end{equation}
    Since $\mu\in\partial_zg(\,\cdot\,,0)$,
    $\P\otimes\ell_1$-a.e. on $\Om\times(s,t)$,
    Lemma~\ref{LEM:h}\itemref{IT:LEM:h:def} implies $h(\,\cdot\,,\Xi)-\mu\cdot\Xi\geq0$, $\P\otimes\ell_1$-a.e. on $\Om\times(s,t)$.
    Hence, $V_s\geq\E_\mu[Y_t|\F_s]$.
    Taking the essential supremum over $\mu\in\partial^0 \!g$ and using~\eqref{EQ:TH:resilience_convex:thesis}, we obtain $V_s
        \geq
        \esssup_{\mu\in\partial^0 \!g}
            \E_\mu[Y_t|\F_s]
        =
        \mathcal D^\r_s\pi_t$.
    
    For $\mu=\bar\mu$, Lemma~\ref{LEM:h}\itemref{IT:LEM:h:selector} and the definition~\eqref{EQ:TH:attainment:def_mu_ast} give $h(\,\cdot\,,\Xi)-\bar\mu\cdot\Xi=0$, $\P\otimes\ell_1$-a.e. on $\Om\times(s,t)$.
    Therefore, $V_s=\E_{\bar\mu}[Y_t|\F_s]$.
    Since $\bar\mu\in\partial^0 \!g$, it follows that
    \[
        V_s
        =
        \E_{\bar\mu}[Y_t|\F_s]
        \leq
        \esssup_{\mu\in\partial^0 \!g}
            \E_\mu[Y_t|\F_s]
        =
        \mathcal D^\r_s\pi_t
        \leq V_s.
    \]
    This proves~\eqref{EQ:TH:attainment:thesis}.
    
    Assume now that the hypotheses of Corollary~\ref{COR:from_convex_to_coherent} hold and that $\mathcal T$ is chosen as in that corollary. 
    Let $q^\s:=\Theta(\,\cdot\,,\s)$, where $\Theta$ is the predictable selector introduced in Lemma~\ref{LEM:g*}. 
    By the choice of $\mathcal T$ in the proof of Corollary~\ref{COR:from_convex_to_coherent}, for every $t\in\mathcal T$
    \begin{equation} 
    \label{EQ:TH:attainment:coherent_selector} 
        q^\s_t\in\partial_zg(t,0), 
        \qquad q^\s_t\cdot\s_t=g(t,\s_t), 
        \qquad \P\text{-a.s.} 
    \end{equation} 
    Define $\widetilde\mu := \bar\mu\,\1_{[0,T]\setminus\{t\}} + q^\s_t\1_{\{t\}}$.
    Since $\Om\times\{t\}$ is predictable, $\widetilde\mu$ is predictable. Moreover, $\widetilde\mu=\bar\mu$ $\P\otimes\ell_1$-a.e.\ and hence $\widetilde\mu\in\partial^0\!g$. 
    In addition, \eqref{EQ:TH:attainment:coherent_selector} yields $\widetilde\mu_t=q^\s_t\in\partial_zg(t,0)$ $\P$-a.s. so that $\widetilde\mu\in\partial_t^0\!g$. 
    Since $\widetilde\mu$ and $\bar\mu$ differ only at the deterministic time $t$, their stochastic exponentials coincide, and therefore $\Q^{\widetilde\mu}=\Q^{\bar\mu}$.
    Consequently, \eqref{EQ:TH:attainment:thesis} and \eqref{EQ:TH:attainment:coherent_selector} give 
    \[ 
        \mathcal D^\r_s\pi_t = \E_{\widetilde\mu}\!\left[ \left.b_t+g(t,\s_t)\right|\F_s \right] = \E_{\widetilde\mu}\!\left[ \left.b_t+\widetilde\mu_t\cdot\s_t\right|\F_s \right]. 
    \] 
    Since $\widetilde\mu\in\partial_t^0\!g$, formula~\eqref{EQ:COR:resilience_coherent} shows that $\widetilde\mu$ attains the essential supremum therein. 
    Relabeling $\widetilde\mu$ as $\bar\mu$ proves \eqref{EQ:TH:attainment:coherent}.
\end{proof}

\subsubsection{Interpretation of the representation formula}
\label{SEC:representation_interpretation}

We collect here the main interpretations of the representation results established above, which complement the ones provided in Section~\ref{SEC:res_eval_interpretation}.

\begin{itemize}[leftmargin=1em] 
    \item 
    \textbf{Local effective drift.}
        Theorem~\ref{TH:resilience_convex} and Remark~\ref{REM:s=t} extend the drift/volatility decomposition of Theorem~\ref{TH:resilience_entropy} to a general convex driver. 
        For $t\in\mathcal T$,
        \begin{equation}
        \label{EQ:s=t_th}
            \mathcal D^\r_t\pi_t
            =
            b_t+g(t,\s_t)
        \end{equation}
        is the local \emph{effective drift} of $\pi$ under the evaluation $\r$: the term $g(t,\s_t)$ is the adjustment assigned by $\r$ to the local volatility $\s_t$, which is added to the physical drift $b_t$.
        Unlike in the entropic case of Section \ref{SEC:entropic_RM}, here $g(t,\s_t)$ can have either sign, thus recovery or deterioration is determined by the sign of the sum.
        A negative adjustment can occur only if $0\notin\partial_zg(\,\cdot\,,0)$, that is, if the physical measure is not among the zero-penalty scenarios, so that in some directions $\r$ is less conservative than the conditional expectation under $\P$.
        This is the case, for instance, for a linear driver $g(z)=\beta\cdot z$, which evaluates positions under the pricing measure $\Q^\beta$: there, $\beta\cdot\s_t$ is the risk premium attached to the exposure $\s_t$, and it can have either sign.
        If instead $0\in\partial_zg(\,\cdot\,,0)$, equivalently $g\geq0$, the adjustment is non-negative and $\r$ can only detect additional deterioration (see the last item below).

    \item
    \label{IT:res_interp_forecast}
    \textbf{Present forecast of future resilience.}
        For $t\in\mathcal T$ and $s\in[0,t]$ satisfying~\ref{IT:GA}, Theorem~\ref{TH:resilience_convex}, Remark~\ref{REM:s=t}, and Theorem~\ref{TH:attainment} yield
        \begin{equation}
        \label{EQ:essmax:inter}
            \mathcal D^\rho_s\pi_t 
            = \essmax_{\mu\in\partial^0 \!g} \E_\mu\!\left[ \mathcal D^\rho_t\pi_t |\F_s \right], 
            \qquad 
            \P\text{-a.s.} 
        \end{equation}
        Recalling \eqref{EQ:s=t}, $\mathcal D^\rho_s\pi_t$ is the time-$s$ worst-case forecast, over zero-penalty scenarios, of the rate of change of $u\mapsto\rho_t(\pi_u)$ at $t$.

    \item
    \textbf{Drift under the optimal scenario and convex adjustment.}
        Let $\bar\mu\in\partial^0\!g$ attain the maximum in  \eqref{EQ:essmax:inter}.  
        For $\ell_1$-a.e. $t$, a predictable representative of $\bar\mu$ can be chosen such that  $\bar\mu_t\in\partial_z g(t,0)$, $\P$-a.s. (see the details in the proof of Corollary~\ref{COR:from_convex_to_coherent}).
        Define, for $\P\text{-a.e. }\om\in\Om$
        \[
            \bar g_t(\om,z)
            :=
            g(\om,t,z)-\bar\mu_t(\om)\cdot z, \qquad \forall\, z\in\R^m.
        \]
        By the subgradient inequality, we have $\bar g_t(z)\geq0$ for all $z\in\R^m$ and $\P$-a.s.
        Therefore,
        \[
            \mathcal D^\rho_t\pi_t=b_t+\bar\mu_t\cdot\s_t + \bar g_t(\s_t), 
        \]
        where the term $b_t+\bar\mu_t\cdot\s_t$ is the drift of $\pi$ under $\Q^{\bar\mu}$, while $\bar g_t(\s_t)\geq0$ is the additional convex adjustment under $\Q^{\bar\mu}$.
        In the coherent setting, Theorem~\ref{TH:attainment} allows one to choose $\bar\mu$ such that  $\bar g_t(\s_t)=0$.
        In the entropic case, $\bar\mu=0$.

    \item
    \textbf{Coherent zero-penalty benchmark.}
        Set $C(\,\cdot\,):=\partial_zg(\,\cdot\,,0)$, ${g^0(\,\cdot\,,z):=\sup_{q\in C(\,\cdot\,)}q\cdot z}$.
        By Lemma~\ref{LEM:h}, $g^0$ is a uniformly Lipschitz sublinear driver; let $\rho^0$ be its induced risk measure.
        Since $g$ and $g^0$ have the same zero-penalty set, the subgradient inequality and
        Corollary~\ref{COR:from_convex_to_coherent} yield
        \[
            \mathcal D^\rho_s\pi_t
            \geq \mathcal D^{\rho^0}_s\pi_t
        \]
        for $t$ in a common full $\ell_1$-measure set and $s\in[0,t]$ satisfying~\itemref{IT:GA}.
        Thus the coherent driver gives the least conservative evaluation among convex drivers with the same zero-penalty set.
        If, in addition, $0\in C(\,\cdot\,)$
        $\P\otimes\ell_1$-a.e. (equivalently $g\geq 0$), then
        \[
            \mathcal D^\rho_s\pi_t
            \geq\mathcal D^{\rho^0}_s\pi_t
            \geq\E[b_t\mid\F_s]
            =\mathcal D^\E_s\pi_t,
        \]
        namely, the coherent benchmark is never more optimistic than the mean resilience rate whenever the physical measure belongs to the zero-penalty class.
\end{itemize}

\subsubsection{Proofs of Theorem~\ref{TH:resilience_convex} and Corollary \ref{COR:from_convex_to_coherent}}
\label{SEC:proofs}
Before proving Theorem~\ref{TH:resilience_convex}, we state two auxiliary lemmas. 
The first quantifies the sensitivity of conditional expectations under $\Q^\mu$ to perturbations of the dual control $\mu$. 
Its proof is given in Appendix~\ref{SEC:app_duality}.

\begin{lemma}
\label{LEM:mu_vs_nu}
    \begin{enumerate}[label=(\roman*), noitemsep]
        \item
        \label{IT:LEM:mu_vs_nu:bound_u}
            Let $K>0$, $u>1$, and set $v=u/(u-1)$. 
            There exists \({C_{u,K,T}>0}\) such that, for every \(\mu,\nu\in\BMO_T\) satisfying  $|\mu|\vee|\nu|\leq K$ $\P\otimes\ell_1$-a.e., \(X\in L^u(\Q^\mu)\) and every \(s\in[0,T]\),
            \[
                \big|\E_\mu[X|\F_s]-\E_\nu[X|\F_s]\big|^u
                \!\!
                \leq
                C_{u,K,T}
                \E_\mu[|X|^u|\F_s]
                \!
                \bigg(
                    \!\!
                    \E_\mu\bigg[
                        \bigg(
                        \!
                        \int_s^T|\mu_r-\nu_r|^2\d r\bigg)^v
                        \bigg|\F_s
                    \bigg]
                \bigg)^{\!\!u/(2v)}.
            \]
        \item
        \label{IT:LEM:mu_vs_nu:bound_infty}
            Assume \(s\in[0,T]\) and let \(\mu,\nu\in\BMO_T\) satisfy $\mu=\nu=0$, $\P\otimes\ell_1$-a.e. on $\Om\times[0,s]$.
            Then, for every \(X\in L^\infty\),
            \[
                \big|\E_\mu[X|\F_s]-\E_\nu[X|\F_s]\big|
                \leq
                \sqrt2\,\|X\|_{L^\infty}
                \bigg(
                    \E_\mu\bigg[
                        \int_s^T|\mu_r-\nu_r|^2\d r
                        \bigg|\F_s
                    \bigg]
                \bigg)^{1/2}.
            \]
        \item
        \label{IT:LEM:mu_vs_nu:equality}
            Assume \(\mu,\nu\in\BMO_T\) satisfy $\mu=\nu$, $\P\otimes\ell_1$-a.e. on $\Om\times(s,a)$, for some \(0\leq s<a\leq T\).
            Let \(X\) be a \([-\infty,+\infty)\)-valued, \(\F_a\)-measurable random variable such that
            \(X^+\in L^1(\Q^\mu)\cap L^1(\Q^\nu)\).
            Then $\E_\mu[X|\F_s]=\E_\nu[X|\F_s]$, the equality holding \(\P\)-a.s.\ as \([-\infty,+\infty)\)-valued random variables.
    \end{enumerate}
\end{lemma}

The next ingredient is a pasting lemma that allows us to control the $L^1$-norm of an essential supremum of conditional expectations.
\begin{lemma}
\label{LEM:pasting}
    Let $s\in[0,T]$ and let $\varnothing\neq\mathcal U\subseteq\BMO_T$ be stable under $\F_s$-pasting, namely
    \begin{equation}
    \label{EQ:LEM:pasting:pasting_U}
        \mu^1\1_A+\mu^2\1_{A^c}\in\mathcal U,
        \qquad
        \forall\,\mu^1,\mu^2\in\mathcal U,\quad A\in\F_s.
    \end{equation}
    Assume that, for some $K>0$ and all $\mu\in\mathcal U$, we have $\mu=0$ on $\Om\times[0,s]$ and $|\mu|\leq K$, $\P\otimes\ell_1$-a.e.
    Then for every non-negative $X\in L^2$,
    \[
       \left\|\esssup_{\mu\in\mathcal U }\E_\mu[X|\F_s]\right\|_{L^1}
       \leq e^{K^2T/2}\|X\|_{L^2}.
    \]
\end{lemma}

\begin{proof}
    By Lemma~\ref{LEM:cond_expec_convex}\itemref{IT:LEM:cond_expec_convex:embedding}, applied with $u=1$ and $p=2$, 
    \[
        \left\|\E_\mu[X|\F_s]\right\|_{L^1} \leq e^{K^2T/2}\|X\|_{L^2}, \qquad \forall\,\mu\in\mathcal U. 
    \]
    
    Let us denote by $\mathcal Q$ the set of all probabilities $\Q^\mu$ on $\mathcal F$ as $\mu$ varies in $\mathcal U$.
    Since $\mu=0$ on $\Om\times[0,s]$, we have $\Q^\mu=\P$ on $\F_s$ for all $\mu\in\mathcal U$.
    Moreover, we show that the set $\mathcal Q$ satisfies the following pasting property:
    \begin{align}
    \label{EQ:LEM:pasting_quadratic:pasting_Q}
        \left(\frac{\d \Q^1}{\d \P}\1_A + \frac{\d \Q^2}{\d \P}\1_{A^c}\right)\P\in\mathcal Q, \qquad \forall\, \Q^1,\Q^2\in\mathcal Q, \ A\in\mathcal \F_s.
    \end{align}
    Indeed, for $i=1,2$, if $\Q^i\in\mathcal Q$, there exists $\mu^i \in\mathcal U$ such that $\Q^i=\Q^{\mu^i}$. 
    Fix $A\in\mathcal F_s$. 
    The property~\eqref{EQ:LEM:pasting:pasting_U} in the hypothesis yields $\mu:=\mu^1\1_A+\mu^2\1_{A^c}\in\mathcal U$. 
    Hence, $\Q^\mu\in\mathcal Q$, where the density of $\Q^\mu$ is exactly $\mathcal E_T^\mu=\mathcal E_T^{\mu^1}\1_A + \mathcal E_T^{\mu^2}\1_{A^c}$,
    which shows the pasting property~\eqref{EQ:LEM:pasting_quadratic:pasting_Q} for the set $\mathcal Q$.

    These properties allow us to apply \cite[Lemma~3.4]{Nutz+Soner_2012_Superhedging_dynamic_risk_measures_volatility_uncertainty} and infer the existence of a sequence $(\mu^n)_{n\in\N}\in\mathcal U^\N$ such that
    \[
        Y_n:=\E_{\mu^n}[X|\F_s] \longrightarrow \esssup_{\mu\in\mathcal U}\E_{\mu}[X|\F_s]:=Y, \qquad \P\text{-a.s.},
    \]
    and $Y_n\leq Y_{n+1}$ for all $n\in\N$.
    Since $X\geq0$, the monotone convergence theorem yields 
    \[
        \|Y\|_{L^1}
        = \lim_{n\to\infty} \left\|\E_{\mu^n}[X|\F_s]\right\|_{L^1} 
        \leq e^{K^2T/2}\|X\|_{L^2}. 
    \]
\end{proof}

We are now ready for the proof of the main Theorem~\ref{TH:resilience_convex}.
We will give the details only under its assumption~\itemref{IT:TH:convex:quad}; the case~\itemref{IT:TH:convex:lip} is treated in Appendix~\ref{SEC:app_lipschitz}.
Although the representation in Theorem~\ref{TH:resilience_convex} is stated over the global zero-penalty class $\partial^0 \!g$, throughout the proof we work with the more convenient normalized class $\A_s$ defined in~\eqref{EQ:A_s}.
The two classes yield the same essential supremum by Remark~\ref{REM:As_or_Ast} when $s<t$; when $s=t$, this follows from the $\F_t$-measurability of $b_t+g(t,\s_t)$.

In order to provide an easy reference to the reader, let us introduce the notations that will be used throughout the proof.
For fixed $0\leq s\leq t<t+\eps\leq T$, we define
\begin{gather}
    \mathcal B_\eps:=\big\{\mu\in\mathcal B \ : \ \mu=0 \ \text{on }\Om\times[0,s], \ \mu\in\partial_zg(\,\cdot\,,0) \ \text{on }\Om\times[t+\eps,T], \ \P\otimes\ell_1\text{-a.e.}\big\},\\
    V_\eps(\mu)
    :=
    \ds \E_\mu\bigg[\frac1\eps\int_t^{t+\eps}(b_r + \s_r\cdot\mu_r)\d r - \frac1\eps\int_s^Tg^\ast(r,\mu_r)\d r \bigg|\F_s\bigg],\\
     G_{s,t}(\mu):=\E_{\mu}\left[\left.\int_s^{t}g^\ast(r,\mu_r)\d r \right|\F_s\right],
    \qquad 
    G_{t,t+\eps}(\mu):=\E_{\mu}\left[\left.\int_t^{t+\eps}g^\ast(r,\mu_r)\d r \right|\F_s\right],\\
    \mathcal C_\eps:=\big\{\mu\in\mathcal B_\eps \ : \  G_{s,t}(\mu)\leq \sqrt \eps, \ G_{t,t+\eps}(\mu)\leq \sqrt\eps, \ \P\text{-a.s.} \big\},\\
    Y_t:=b_t + g(t,\s_t),
    \qquad  \quad
    \langle Y_t\rangle_\eps :=\frac1\eps\int_t^{t+\eps}Y_r\d r,
    \qquad \quad
    J_{s,t}:=\esssup_{\mu\in\A_s}\E_\mu[Y_t|\F_s].
\end{gather}

Let us now outline the proof, which is organized in seven steps.
\begin{enumerate}[label=\textit{Step $\arabic*$.\!}, leftmargin=*, noitemsep, topsep=0pt]
    \item 
        We rewrite the increments of the process $\pi$ via Girsanov's theorem, then use the dual representation~\eqref{EQ:dual_representation_convex} for the risk measure $\r$, obtaining
        \[
            \frac1\eps\r_s(\pi_{t+\eps}-\pi_t)=\esssup_{\mu\in\mathcal B}V_\eps(\mu).
        \]

    \item
        We show that the essential supremum can be computed equivalently over $\mathcal B$ or $\mathcal B_\eps$:
        \[
            \esssup_{\mu\in\mathcal B}V_\eps(\mu)
            =\esssup_{\mu\in\mathcal B_\eps}V_\eps(\mu).
        \]
        If $\mu\in\mathcal B_\eps$, then $g^\ast(\,\cdot\,,\mu)=0$ on $\Om\times[t+\eps,T]$.
        Therefore, this step allows us to restrict the time integration in the penalty term in~\eqref{EQ:dual_representation_convex} from $[s,T]$ to $[s,t+\eps]$.

    \item
        We perform the main asymptotic localization.
        We split the penalty term into its contribution on $[s,t]$ and on $[t,t+\eps]$, and we show that the conditional payoff function $V_\eps(\mu)$ diverges to $-\infty$, as $\eps\to 0^+$, as soon as either of these two contributions is larger than $\sqrt\eps$.
        This is due to the normalization factor $1/\eps$ appearing in front of the penalty term.
        This means that, asymptotically, the leading contribution to the essential supremum comes from processes $\mu\in\mathcal C_\eps$.

    \item 
        We derive an upper bound for the essential supremum over $\mathcal C_\eps$:
        \[
            \esssup_{\mu\in\mathcal C_\eps} V_\eps(\mu)
            \leq \esssup_{\mu\in\mathcal C_\eps}
                \E_{\mu}\big[\langle Y_t\rangle_\eps \big|\F_s\big]
            \leq  \esssup_{\mu\in\A_s}\E_{\mu}\big[\langle Y_t\rangle_\eps \big|\F_s\big] 
                + C_1\,\varpi(\eps),
        \]
        where $\varpi(\eps)\to 0$ as $\eps\to 0^+$.
        For this comparison, both Lemma~\ref{LEM:mu_vs_nu}\itemref{IT:LEM:mu_vs_nu:bound_infty} and the assumption~\itemref{IT:GA} play a key role in proving that the error vanishes as $\eps\to 0^+$.

    \item 
        We prove the complementary lower bound
        by appropriately modifying processes in $\A_s$ on the interval $(t,t+\eps)$, and then using Lemma~\ref{LEM:mu_vs_nu}\itemref{IT:LEM:mu_vs_nu:bound_infty}:
        \[
            \esssup_{\mu\in\mathcal B_\eps}V_\eps(\mu)
            \geq
            \esssup_{\mu\in\A_s}
            \E_\mu[\langle Y_t\rangle_\eps|\F_s]
            -
            C_2\sqrt\eps.
        \]

    \item 
        We identify the limit of the common term $\esssup_{\mu\in\A_s}\E_\mu[\langle Y_t\rangle_\eps |\F_s]$ appearing in both the upper and the lower bounds.
        For $\ell_1$-a.e. $t$, the time averages $\langle Y_t\rangle_\eps$ converge in $L^2$ to $Y_t$.
        Using the pasting stability of $\A_s$ for Lemma~\ref{LEM:pasting}, we show that
        \[
            \esssup_{\mu\in\A_s}\E_\mu[\langle Y_t\rangle_\eps\mid\F_s]
            \longrightarrow \esssup_{\mu\in\A_s}\E_\mu[Y_t\mid\F_s], 
            \qquad\text{in }L^1, \text{ as }\eps\to 0^+.
        \]

    \item 
        Finally, we control separately the convergence of $\eps^{-1}\rho_s(\Delta_\eps\pi_t)$ to the candidate limit $\esssup_{\mu\in\A_s}\E_\mu[Y_t\mid\F_s]$ from above and below. 
        The positive part of the difference is controlled by means of the asymptotic restriction to $\mathcal C_\eps$ from \textit{Step} {3} and the upper bound from \textit{Step} {4}.
        The negative part is estimated by the lower bound from \textit{Step} {5}.
        Finally, the convergence of the common term from \textit{Step} {6} ensures that 
        \[
            \frac1\eps\rho_s(\Delta_\eps\pi_t)\longrightarrow \esssup_{\mu\in\A_s}\E_\mu[Y_t\mid\F_s],\qquad\text{in }L^1, \text{ as }\eps\to 0^+. 
        \] 
\end{enumerate}

\begin{proof}[\textbf{Proof of Theorem~\ref{TH:resilience_convex} under~\itemref{IT:TH:convex:quad}.}]
    Assume that $g$ satisfies~\ref{IT:g_quad},  that ${b,|\s|\in L^\infty_T}$, and that the risk measure $\r$ induced by $g$ has dual representation~\eqref{EQ:dual_representation_convex} with $\mathcal B:=\BMO_T$.

\smallskip
\textit{Step $1$ (Girsanov theorem)}.
    For every $\mu\in\BMO_T$,  $W^\mu_r:=W_r-\int_0^r\mu_u\d u$, for $r\in[0,T]$, is a $\Q^\mu$-Brownian motion, where $\Q^\mu$ is defined in~\eqref{EQ:equivalent_mg_measure}.
    
    For all $t\in[0,T)$ and $\eps\in(0,T-t]$, we denote by $\Delta_\eps\pi_t$ the increment $\pi_{t+\eps}-\pi_t$, and we notice that $\Delta_\eps\pi_t\in L^{\mathrm{exp}}(\F_T)$.
    Indeed, it is $\F_{t+\eps}$-measurable, hence $\F_T$-measurable, and
    \begin{align}
        \E\big[\exp(c|\Delta_\eps\pi_t|)\big]
        \leq 
        2\exp\big(\eps c\|b\|_{L^\infty_T} + \eps c^2\|\s\|_{L^\infty_T}^2\big)<+\infty,\qquad\forall\, c>0,
    \end{align}
    where we estimated the martingale part of the increment as in~\eqref{EQ:TH:entropy:stoch_expon_bound_begin},~\eqref{EQ:TH:entropy:stoch_expon_bound}.
        
    Using $W^\mu$, we can rewrite the increment:
    \begin{align}
        \Delta_\eps \pi_t
        =\int_t^{t+\eps}b_r\d r +\int_t^{t+\eps}\s_r \cdot\d W_r
        =\int_t^{t+\eps}(b_r+\s_r\cdot \mu_r)\d r +\int_t^{t+\eps}\s_r \cdot\d W^\mu_r.
    \end{align}
    
    For all $s\in[0,t]$, by Proposition~\ref{PROP:dual_representation}, and because the It\^o integral in the last line has vanishing $\F_s$-conditional expectation under $\Q^\mu$:
    \begin{equation}
    \label{EQ:TH:convex:step1_quadratic}
        \frac1\eps\r_s(\Delta_\eps\pi_t)
        =\esssup_{\mu\in\mathcal B}\E_\mu\left[\left.\frac1\eps\int_t^{t+\eps}(b_r + \s_r\cdot\mu_r)\d r - \frac1\eps\int_s^Tg^\ast(r,\mu_r)\d r \right|\F_s\right].
    \end{equation}
    The conditional expectation in the right-hand side will be henceforth denoted by $V_\eps(\mu)$.

\smallskip    
\textit{Step $2$ (Restriction from $\mathcal B$ to $\mathcal B_\eps$).}
    Let us fix $0\leq s\leq t<t+\eps\leq T$, define
    \[
        \mathcal B_\eps:=\big\{\mu\in\mathcal B \ : \ \mu=0\text{ on }\Om\times[0,s], \ \mu\in\partial_zg(\,\cdot\,,0)\text{ on }\Om\times[t+\eps,T], \ \P\otimes\ell_1\text{-a.e.}\big\},
    \]
    and show that the essential supremum in~\eqref{EQ:TH:convex:step1_quadratic} can be restricted to $\mathcal B_\eps$.
    
    Let us fix $\mu\in\mathcal B$ and consider its predictable metric projection $\Pi(\,\cdot\,,\mu)$ onto $\partial_z g (\,\cdot\,,0)$ via the $\mathcal P\otimes\mathscr B(\R^m)$-measurable map $\Pi$ from Lemma~\ref{LEM:g*}\itemref{IT:LEM:g*:projector}.
    Define $\tilde\mu:=\mu\1_{(s,t+\eps)} + \Pi(\,\cdot\,,\mu)\1_{[t+\eps,T]}$ so that $\tilde \mu\in\mathcal B_\eps$.
    Then,
    \begin{align}
        V_\eps(\mu)
        &=\E_\mu\left[\left.\frac1\eps\int_t^{t+\eps}(b_r + \s_r\cdot\tilde \mu_r)\d r - \frac1\eps\int_s^{t+\eps}g^\ast(r,\tilde \mu_r)\d r  - \frac1\eps\int_{t+\eps}^Tg^\ast(r,\mu_r)\d r \right|\F_s\right]\\
        &\leq \E_\mu\left[\left.\frac1\eps\int_t^{t+\eps}(b_r + \s_r\cdot\tilde \mu_r)\d r - \frac1\eps\int_s^{t+\eps}g^\ast(r,\tilde \mu_r)\d r\right|\F_s\right].
    \label{EQ:one_more_day}
    \end{align}
    For the equality we split the second integral and used the fact that $\tilde \mu=\mu$ on $\Om\times(s,t+\eps)$.
    For the inequality, we used the non-negativity of $g^\ast$.
    
    The difference between the two Lebesgue integrals in~\eqref{EQ:one_more_day} defines an $\F_{t+\eps}$-measurable random variable with values in $[-\infty,+\infty)$ and with positive part bounded by $\eps\big(\|b\|_{L^\infty_T}+k(1+\|\s\|^2_{L^\infty_T})\big)$: indeed, by $g^\ast\geq0$ and Fenchel's inequality $\s\cdot\tilde\mu-g^\ast(\,\cdot\,,\tilde\mu)\leq g(\,\cdot\,,\s)$, the difference is at most $\int_t^{t+\eps}\big(b_r+g(r,\s_r)\big)\d r$.
    Therefore,  we can apply Lemma~\ref{LEM:mu_vs_nu}\itemref{IT:LEM:mu_vs_nu:equality} to~\eqref{EQ:one_more_day} with $a:=t+\eps$ and $\nu:=\tilde \mu$:
    \[
        V_\eps(\mu)
        \leq \E_{\tilde \mu}\left[\left.\frac1\eps\int_t^{t+\eps}(b_r + \s_r\cdot\tilde \mu_r)\d r - \frac1\eps\int_s^{t+\eps}g^\ast(r,\tilde \mu_r)\d r\right|\F_s\right] 
        =V_\eps(\tilde\mu),
    \]
    where we implicitly used that $g^\ast(\,\cdot\,,\tilde\mu)=g^\ast(\,\cdot\,,\Pi(\,\cdot\,,\mu))=0$ on $\Om\times[t+\eps,T]$ because $\Pi(\,\cdot\,,\mu)\in\partial_zg(\,\cdot\,,0)$.
    Control the right-hand side with $\esssup_{\nu\in\mathcal B_\eps}V_\eps(\nu)$, then take the essential supremum with respect to $\mu\in\mathcal B$ to get 
    \[
            \esssup_{\mu\in\mathcal B}V_\eps(\mu)\leq \esssup_{\mu\in\mathcal B_\eps}V_\eps(\mu).
    \]
    Since $\mathcal B_\eps\subseteq \mathcal B$, the opposite inequality is trivial.
    Hence, we obtain the equality:
    \begin{equation}
    \label{EQ:TH:convex:step2_bis_quadratic}
        \frac1\eps\r_s(\Delta_\eps\pi_t)
        =\esssup_{\mu\in\mathcal B}V_\eps(\mu)
        =\esssup_{\mu\in\mathcal B_\eps}V_\eps(\mu).
    \end{equation}

\smallskip
\textit{Step $3$ (Asymptotic restriction from $\mathcal B_{\eps}$ to $\mathcal C_\eps$).}
    Fix $\mu\in\mathcal B_\eps$.
    Rewrite $V_\eps(\mu)$ by splitting the penalty term into an integral over $[s,t]$ and two equal contributions over $[t,t+\eps]$:
    \begin{equation}
    \label{EQ:TH:convex_quadratic:step_J_1}
    \begin{aligned}
        V_\eps(\mu)
        &=\E_\mu\left[\left.\frac1\eps\int_{t}^{t+\eps}\Big(b_r + \s_r\cdot \mu_r-\frac12g^\ast(r,\mu_r)\Big)\d r\right|\F_s\right]\\
        &\qquad-\frac1{2\eps}\E_\mu\left[\left.\int_{t}^{t+\eps}g^\ast(r,\mu_r)\d r\right|\F_s\right]
        -\frac1{\eps}\E_\mu\left[\left.\int_{s}^{t}g^\ast(r,\mu_r)\d r\right|\F_s\right].
    \end{aligned}
    \end{equation}

    The integrand in the first conditional expectation can be estimated as follows by recalling the lower bound from Lemma~\ref{LEM:g*}\itemref{IT:LEM:g*:def}\itemref{IT:LEM:g*:prop_nonnega}:
    \begin{align}
        b+\s\cdot \mu - \frac12g^\ast(\,\cdot\,,\mu)
        &\leq \|b\|_{L^\infty_T} + \|\s\|_{L^\infty_T}|\mu| - \frac12\Big(\frac1{4k}|\mu|^2 -k\Big)\\
        &\leq \|b\|_{L^\infty_T} + \sup_{x>0}\left\{\|\s\|_{L^\infty_T}x - \frac1{8k}x^2\right\} + \frac k2.
    \end{align}
    The supremum trivially equals $2k\|\s\|^2_{L^\infty_T}$.
    Consequently, the first conditional expectation in~\eqref{EQ:TH:convex_quadratic:step_J_1} is essentially bounded by $\|b\|_{L^\infty_T}+2k\|\s\|^2_{L^\infty_T}+ k/2$.
    
    Denote by $G_{t,t+\eps}(\mu)$ and $ G_{s,t}(\mu)$ the second and third conditional expectations in~\eqref{EQ:TH:convex_quadratic:step_J_1}, respectively.
    Hence, we obtain the estimate
    \begin{align}
        V_\eps(\mu)\leq C_0-\frac1{2\eps}G_{t,t+\eps}(\mu)-\frac1\eps G_{s,t}(\mu),
    \label{EQ:TH:convex_quadratic:step_J_2}
    \end{align}
    where $C_0:=\|b\|_{L^\infty_T}+2k\|\s\|^2_{L^\infty_T}+ k/2>0$ is a constant.

    We now define the asymptotically relevant class
    \begin{equation}
    \label{EQ:TH:convex_quadratic:def_K}
        \mathcal C_\eps
        :=\Big\{\mu\in\mathcal B_\eps\ :\   G_{s,t}(\mu)\le \sqrt\eps,\ \  G_{t,t+\eps}(\mu)\le \sqrt\eps, \ \P\text{-a.s.}\Big\}.
    \end{equation}
    Let us use $\Pi$ from Lemma~\ref{LEM:g*}\itemref{IT:LEM:g*:projector} to define the predictable process $\mu^0:=\Pi(\,\cdot\,,0)\1_{(s,T]} \in\mathcal C_\eps$.
    For $\mu\in\mathcal B_\eps$, set
    \[
        A_\mu:=
        \big\{G_{s,t}(\mu)\leq\sqrt\eps,\
          G_{t,t+\eps}(\mu)\leq\sqrt\eps\big\}\in\F_s,
        \qquad
        \widehat\mu:=
        \mu\1_{A_\mu}+\mu^0\1_{A_\mu^c}.
    \]
    Since $A_\mu\in\F_s$ and both controls vanish on $[0,s]$, $\widehat\mu$ is predictable. 
    Moreover, $|\widehat\mu|^2\leq|\mu|^2+|\mu^0|^2$, so $\widehat\mu\in\BMO_T$; it also satisfies the restrictions defining $\mathcal B_\eps$. 
    Since $A_\mu\in\F_s$ and $\mu,\mu^0$ vanish on $[0,s]$, we have $\E_{\widehat\mu}[\xi|\F_s]=\1_{A_\mu}\E_\mu[\xi|\F_s]+\1_{A_\mu^c}\E_{\mu^0}[\xi|\F_s]$ for every $\F_T$-measurable $\xi$ with $\xi^+\in L^1(\Q^\mu)\cap L^1(\Q^{\mu^0})$, as in the proof of Lemma~\ref{LEM:pasting}. Together with $g^\ast(\,\cdot\,,\mu^0)=0$ on $(s,T]$, this gives
    \[
        G_{s,t}(\widehat\mu)
          =\1_{A_\mu}G_{s,t}(\mu),
          \quad 
        G_{t,t+\eps}(\widehat\mu)
          =\1_{A_\mu}G_{t,t+\eps}(\mu),
      \quad
        \1_{A_\mu}V_\eps(\widehat\mu)
          =\1_{A_\mu}V_\eps(\mu).
    \]
    Hence, $\widehat\mu\in\mathcal C_\eps$. On $A_\mu^c$, at least one of the two penalties of $\mu$ exceeds $\sqrt\eps$, and \eqref{EQ:TH:convex_quadratic:step_J_2} yields $V_\eps(\mu)\leq C_0-1/(2\sqrt\eps)$ there.
    Thus, for every $\mu\in\mathcal B_\eps$, we have $V_\eps(\mu)\leq\esssup_{\nu\in\mathcal C_\eps}V_\eps(\nu)$ on $A_\mu$, while $V_\eps(\mu)\leq C_0-1/(2\sqrt\eps)$ on $A_\mu^c$.
    Consequently,
    \begin{equation}
    \label{EQ:TH:convex_quadratic:step_J_2_bis}
        \esssup_{\mu\in\mathcal B_\eps}V_\eps(\mu)
        \leq
        \bigg(\esssup_{\mu\in\mathcal C_\eps}V_\eps(\mu)\bigg)
        \vee\left(C_0-\frac1{2\sqrt\eps}\right).
    \end{equation}

\smallskip
\textit{Step $4$ (Upper bound on $\mathcal C_\eps$).}
    Thanks to the previous step, we fix $\mu\in\mathcal C_\eps$.
    Rewrite once again $V_\eps(\mu)$ by splitting the penalty term into an integral over $[s,t]$ and an integral over $[t,t+\eps]$. 
    While the second can be included with the other terms integrated over $[t,t+\eps]$, the first leads to $ G_{s,t}(\mu)$.
    \begin{equation}
        V_\eps(\mu)
        = \E_\mu\left[\left.\frac1\eps\int_t^{t+\eps}\big(b_r + \mu_r\cdot\s_r - g^\ast(r,\mu_r)\big)\d r\right|\F_s\right]-\frac1\eps G_{s,t}(\mu).
    \label{EQ:TH:convex_quadratic:upper_1}
    \end{equation}
    The definition of $g^\ast$ yields $g^\ast(\,\cdot\,,\mu)
        =\sup_{z\in\R^m}\big\{\mu\cdot z-g(\,\cdot\,,z)\big\}
        \geq \mu\cdot\s-g(\,\cdot\,,\s)$.
    Use this in~\eqref{EQ:TH:convex_quadratic:upper_1}, together with the non-negativity of $ G_{s,t}(\mu)$,  to obtain
    \begin{equation}
        V_\eps(\mu)
        \leq  \E_\mu\left[\left.\frac1\eps\int_t^{t+\eps}\big(b_r + g(r,\s_r)\big)\d r\right|\F_s\right]
        =:\E_\mu[\langle Y_t\rangle_\eps |\F_s].
    \label{EQ:TH:convex_quadratic:upper_2}
    \end{equation}

    Let us use the map $\Pi$ from Lemma~\ref{LEM:g*}\itemref{IT:LEM:g*:projector} to define  $\tilde\mu=\Pi(\,\cdot\,,\mu)\1_{(s,t+\eps)}+\mu\1_{[t+\eps,T]}$ so that $\tilde\mu\in\A_s$.
    Observe that $\langle Y_t\rangle_\eps$ belongs to $L^\infty$.
    Indeed, by the quadratic bound on $g$,
    \begin{align}
    \label{EQ:TH:convex:step3_quadratic_bis}
        \big|\langle Y_t\rangle_\eps\big|
        \leq
        \frac1\eps\int_t^{t+\eps}
            \big(|b_r|+|g(r,\s_r)|\big)\d r
        \leq
        \|b\|_{L^\infty_T}
        +k\big(1+\|\s\|_{L^\infty_T}^2\big)
        =:M.
    \end{align}
    Consequently, $\|\langle Y_t\rangle_\eps\|_{L^\infty}\leq M$.
    Therefore, we estimate the difference between the conditional expectations of $\langle Y_t\rangle_\eps $ with respect to $\mu$ and $\tilde\mu$ by means of Lemma~\ref{LEM:mu_vs_nu}\itemref{IT:LEM:mu_vs_nu:bound_infty} with $\nu:=\tilde\mu$ and $X=\langle Y_t\rangle_\eps $:
    \begin{equation}
        \E_{\mu}\big[\langle Y_t\rangle_\eps \big|\F_s\big]
        -\E_{\tilde\mu}\big[\langle Y_t\rangle_\eps \big|\F_s\big]
        \leq\sqrt 2 \|\langle Y_t\rangle_\eps \|_{L^\infty} 
        \!\!\left(\E_{\mu}\left[\left.\int_s^{t+\eps}|\mu_r-\tilde\mu_r|^2\d r \right|\F_s\right]\right)^{1/2}.
    \label{EQ:TH:convex_quadratic:step_upperbound_1}
    \end{equation}
    First, recall that $\tilde\mu=\Pi(\,\cdot\,,\mu)$ on $\Om\times(s,t+\eps)$.
    Then, we estimate the conditional quadratic variation in the right-hand side by splitting the integral over $(s,t)$ and $(t,t+\eps)$.
    On $(s,t)$, assume the condition \itemref{IT:GA}, then, by Remark \ref{REM:GA}\itemref{IT:REM:GA:concave}, \eqref{EQ:IT:GA_concave} is satisfied for some continuous, strictly increasing and concave function $\phi:[0,+\infty)\to[0,\infty)$  with $\phi(0)=0$.
    Thus 
    \[
        |\mu-\tilde\mu|^2
        =|\mu-\Pi(\,\cdot\,,\mu)|^2
        =\operatorname{dist}^2\big(\mu,\partial_z g(\,\cdot\,,0)\big)
        \leq \phi\big(g^\ast(\,\cdot\,,\mu)\big),
        \quad \P\otimes\ell_1\text{-a.e. on }\Om\times(s,t).
    \]
    We use it when $s\neq t$, with Jensen's inequality applied to $\phi$, first on $(s,t)$, and then conditionally with respect to $\mathcal F_s$:
    \begin{align}
        \E_\mu\left[\left.\int_s^t|\mu_r-\tilde\mu_r|^2\d r\right|\F_s\right]
        \leq (t-s)\,\phi\left(\frac{G_{s,t}(\mu)}{t-s}\right)
        \leq (1\vee T)\,\phi\big(G_{s,t}(\mu)\big),
    \end{align}
    where, in the end, we used $\a\,\phi(x/\a)\le (1\vee\a)\phi(x)$ from Remark~\ref{REM:GA}\itemref{IT:REM:GA:concave} with $\a:=t-s\leq T$ and $x:=G_{s,t}(\mu)$.
    The same estimate is trivially satisfied when $s=t$.
    On $(t,t+\eps)$, we use instead the estimates from Lemma~\ref{LEM:g*}\itemref{IT:LEM:g*:def}\itemref{IT:LEM:g*:prop_nonnega},\itemref{IT:LEM:g*:subgrad_0}:
    \[
        |\mu-\tilde\mu|^2
        =|\mu-\Pi(\,\cdot\,,\mu)|^2
        \leq 2|\mu|^2+2|\Pi(\,\cdot\,,\mu)|^2
        \leq 8k\,g^\ast(\,\cdot\,,\mu)+16k^2.
    \]
    Therefore, recombining the two intervals:
    \begin{align}
    \label{EQ:TH:convex_quadratic:step_upperbound_2}
        \E_{\mu}\left[\left.
            \int_s^{t+\eps}|\mu_r-\tilde\mu_r|^2\d r
        \right|\F_s\right]
        &\leq (1\vee T)\,\phi\big(G_{s,t}(\mu)\big)+8k\,G_{t,t+\eps}(\mu)+16k^2\eps\\
        &\leq
        C_{k,T}\big(\phi(\sqrt\eps)+\sqrt\eps+\eps\big).
    \end{align}
    For the last inequality, we used that $\mu\in\mathcal C_\eps$ with the increasing monotonicity of $\phi$, and set $C_{k,T}:=1\vee T\vee (8k) \vee (16k^2)>0$.
    Combining this with~\eqref{EQ:TH:convex_quadratic:step_upperbound_1} and using ${\|\langle Y_t\rangle_\eps \|_{L^\infty} \leq M}$, we find a constant $C_1:=M \sqrt{2C_{k,T}}>0$, depending only on $k,T,\|b\|_{L^\infty_T}$ and $\|\s\|_{L^\infty_T}$, such that
    \begin{equation}
        \E_{\mu}\big[\langle Y_t\rangle_\eps \big|\F_s\big]
        \leq  \E_{\tilde\mu}\big[\langle Y_t\rangle_\eps \big|\F_s\big] 
        + C_1\,\varpi(\eps)
        \leq  \esssup_{\nu\in\A_s}\E_{\nu}\big[\langle Y_t\rangle_\eps \big|\F_s\big] 
        + C_1\,\varpi(\eps),
    \end{equation}
    where $\varpi(\eps):=[\phi(\sqrt\eps)+\sqrt\eps+\eps]^{1/2}$ satisfies $\varpi(\eps)\to 0$ as $\eps\to 0^+$ thanks to Remark~\ref{REM:GA}\itemref{IT:REM:GA:concave}.
    Take the essential supremum over $\mu\in \mathcal C_\eps$ and use~\eqref{EQ:TH:convex_quadratic:upper_2} to get
    \begin{equation}
    \label{EQ:TH:convex_quadratic:upper_3}
        \esssup_{\mu\in\mathcal C_\eps} V_\eps(\mu)
        \leq \esssup_{\mu\in\mathcal C_\eps}
            \E_{\mu}\big[\langle Y_t\rangle_\eps \big|\F_s\big]
        \leq  \esssup_{\mu\in\A_s}\E_{\mu}\big[\langle Y_t\rangle_\eps \big|\F_s\big] 
        +  C_1\,\varpi(\eps).
    \end{equation}

\smallskip
\textit{Step $5$ (Lower bound on $\mathcal B_\eps$).}
    Set $q^\sigma:=\Theta(\,\cdot\,,\sigma)$.
    By Lemma~\ref{LEM:g*}\itemref{IT:LEM:g*:selector_subgrad}, we have  $\sigma\cdot q^\sigma-g^\ast(\,\cdot\,,q^\sigma)
        =
        g(\,\cdot\,,\sigma)$ $\P\otimes\ell_1$-a.e., and $|q^\sigma|\leq 5k(1+\|\sigma\|_{L^\infty_T})=:R$.
    Fix $\mu\in\A_s$ and define $\tilde\mu
        :=
        \mu\1_{[s,t]\cup[t+\eps,T]}
        +
        q^\sigma\1_{(t,t+\eps)}$.
    Since $\mu\in\partial_z g(\,\cdot\,,0)$ on $\Omega\times(s,T]$,
    we have $g^\ast(\,\cdot\,,\mu)=0$ on $\Omega\times(s,T]$.
    Moreover, $\tilde\mu$ is bounded, hence $\tilde\mu\in\BMO_T$, hence
    $\tilde\mu\in\mathcal B_\eps$.
    Since $\tilde\mu=q^\sigma$ on $\Omega\times(t,t+\eps)$, we obtain $V_\eps(\tilde\mu)
        =
        \E_{\tilde\mu}[\langle Y_t\rangle_\eps|\F_s]$.

    As in \textit{Step} {4}, Lemma~\ref{LEM:mu_vs_nu}\itemref{IT:LEM:mu_vs_nu:bound_infty}
    gives
    \begin{equation}
        \E_{\mu}[\langle Y_t\rangle_\eps|\F_s]
        -V_\eps(\tilde\mu)
        \leq
        \sqrt2 \|\langle Y_t\rangle_\eps\|^{}_{L^\infty}
        \sqrt{
            \E_\mu\bigg[
                \int_s^{t+\eps}
                |\mu_r-\tilde\mu_r|^2\d r
                \bigg|\F_s
            \bigg]
        }
        \leq
        M(R+\!2k)\sqrt{2\eps}.
    \end{equation}
    Therefore, for the constant $C_2:=M(R+2k)\sqrt 2>0$ depending only on
    $\|b\|_{L^\infty_T}$, $\|\sigma\|_{L^\infty_T}$ and $k$,
    \[
        \E_\mu[\langle Y_t\rangle_\eps|\F_s]
        \leq
        V_\eps(\tilde\mu)+C_2\sqrt\eps
        \leq
        \esssup_{\nu\in\mathcal B_\eps}V_\eps(\nu)
        +C_2\sqrt\eps.
    \]
    Taking the essential supremum over $\mu\in\A_s$ gives
    \begin{equation}
    \label{EQ:TH:convex_quadratic:step_lowerbound_2_new}
        \esssup_{\nu\in\mathcal B_\eps}V_\eps(\nu)
        \geq
        \esssup_{\mu\in\A_s}
        \E_\mu[\langle Y_t\rangle_\eps|\F_s]
        -
        C_2\sqrt\eps.
    \end{equation}

\smallskip
\textit{Step $6$ (Limit of the upper and lower bounds).}
    Since $Y\in L^2_T$, Lemma~\ref{LEM:integral_average} yields a Borel set $\mathcal T_1\subseteq[0,T)$ with full $\ell_1$-measure such that, for all $t\in\mathcal T_1$, $\langle Y_t\rangle_\eps \longrightarrow Y_t$ in $L^2$ as $\eps\to 0^+$.
    Fix $0\leq s\leq t<t+\eps\leq T$ with $t\in\mathcal T_1$.
    We prove that the term $\esssup_{\mu\in\A_s}\E_\mu[\langle Y_t\rangle_\eps |\F_s]$ appearing in the right-hand sides of both the upper bound~\eqref{EQ:TH:convex_quadratic:upper_3} and the lower bound~\eqref{EQ:TH:convex_quadratic:step_lowerbound_2_new} converges in $L^1$, as $\eps\to 0^+$, to  $J_{s,t}:=\esssup_{\mu\in\A_s}\E_\mu[Y_t|\F_s]$.
    
    Recall this general property: if $I$ is an arbitrary index set and $(x_i)_{i\in I}$, $(y_i)_{i\in I}$ are families of real-valued random variables, then $|\esssup_{i\in I} x_i - \esssup_{i\in I} y_i|\leq \esssup_{i\in I}|x_i-y_i|$.
    Using this:
    \begin{align}
        \left|
            \esssup_{\mu\in\A_s}
            \E_\mu[\langle Y_t\rangle_\eps\mid\F_s]
            -
            J_{s,t}
        \right|
        &\leq
        \esssup_{\mu\in\A_s}
        \E_\mu[
            |\langle Y_t\rangle_\eps-Y_t|
            \mid\F_s
        ].
    \end{align}
    We now apply Lemma~\ref{LEM:pasting} to the set $\A_s\subset \BMO_T$.
    First, every $\mu\in\A_s$ satisfies $\mu=0$ on $\Omega\times[0,s]$ by definition.
    Let us verify the pasting property~\eqref{EQ:LEM:pasting:pasting_U}.
    Fix $\mu^1,\mu^2\in\A_s$ and $A\in\mathcal F_s$, and define $\mu:=\mu^1\1_A+\mu^2\1_{A^c}$.
    Then $\mu=0$ on $\Omega\times[0,s]$.
    Moreover, since $\mu^1,\mu^2\in\partial_zg(\,\cdot\,,0)$ on $\Omega\times(s,T]$, the same holds for $\mu$.
    Thus $\mu\in\A_s$.
    Consequently, Lemma~\ref{LEM:pasting} applied with $X=|\langle Y_t\rangle_\eps-Y_t|$ gives
    \[
        \left\|
            \esssup_{\mu\in\A_s}
            \E_\mu[\langle Y_t\rangle_\eps\mid\F_s]
            -
            J_{s,t}
        \right\|_{L^1}
        \leq
        e^{2k^2T}
        \|\langle Y_t\rangle_\eps-Y_t\|_{L^2}.
    \]
    Since $t\in\mathcal T_1$, the right-hand side vanishes as
    $\eps\to0^+$.
    Hence,
    \begin{equation}
    \label{EQ:TH:convex_quadratic:step_limitupperbound_4}
        \esssup_{\mu\in\A_s}\E_\mu\big[\langle Y_t\rangle _{\eps}\big|\F_s\big]
        \longrightarrow J_{s,t}, \qquad \text{in }L^1, \text{ as }\eps\to 0^+.
    \end{equation}

\smallskip
\textit{Step $7$ (Convergence in $L^1$).}  
    We now employ~\eqref{EQ:TH:convex_quadratic:step_limitupperbound_4} to infer the convergence in $L^1$ of $\esssup_{\mu\in\mathcal B_\eps}V_\eps(\mu)$, by means of the upper and lower bounds~\eqref{EQ:TH:convex_quadratic:upper_3},~\eqref{EQ:TH:convex_quadratic:step_lowerbound_2_new}.

    Fix $t\in\mathcal T_1$ and $s\in[0,t]$.
    Then for all $\eps\in(0,T-t]$, by the lower bound~\eqref{EQ:TH:convex_quadratic:step_lowerbound_2_new}:
    \begin{align}
        \left(
            \esssup_{\mu\in\mathcal B_\eps}V_\eps(\mu)-J_{s,t}
        \right)^-
        \leq
        \left|
            \esssup_{\mu\in\A_s}\E_\mu\big[\langle Y_t\rangle_\eps \big|\F_s\big]
            -J_{s,t}
        \right|
        +C_2\sqrt{\eps}.
    \end{align}
    Taking expectations and using the convergence~\eqref{EQ:TH:convex_quadratic:step_limitupperbound_4} proved in the previous step, we conclude that
    \begin{equation}
    \label{EQ:TH:convex_quadratic:L1_1}
        \left(
        \esssup_{\mu\in\mathcal B_\eps}V_\eps(\mu)-J_{s,t}
        \right)^-
        \longrightarrow0,
        \qquad \text{in }L^1,\text{ as }\eps\to0^+.
    \end{equation}
    
    Before showing a similar convergence for the positive part, we need to restrict the set of times $\mathcal T_1$.
    Let us define the predictable process $q^0:=\Pi(\,\cdot\,,0)$ thanks to Lemma~\ref{LEM:g*}\itemref{IT:LEM:g*:projector}.
    Then also $q^0\in\partial_z g (\,\cdot\,,0)$ and $|q^0|\leq 2k$, $\P\otimes\ell_1$-a.e.
    From these properties and from the assumptions~\ref{IT:g_quad} on $g$, Fubini's theorem yields a Borel set $\mathcal T_2\subseteq[0,T]$ with full $\ell_1$-measure such that, for $t\in\mathcal T_2$
    \[
        q^0_t\in\partial_z g (t,0), \quad |q^0_t|\leq 2k, \quad g(t,0)=0, \quad |g(t,\s_t)|\leq k(1+|\s_t|^2), \qquad \P\text{-a.s.}
    \]
    Hence, for $t\in\mathcal T_2$, the subgradient inequality yields
    \[
        -2k\|\s\|_{L^\infty_T}
        \leq \s_t\cdot q^0_t
        =\s_t\cdot q^0_t - g(t,0)
        \leq  g(t,\s_t)
        \leq k(1+|\s_t|^2)
        \leq k(1+\|\s\|_{L^\infty_T}^2).
    \]
    Consequently, $Y_t\in L^\infty$ with lower bound $m:=-\|b\|_{L^\infty_T}-2k\|\s\|_{L^\infty_T} \leq Y_t$ and upper bound $M:=\|b\|_{L^\infty_T} + k(1+\|\s\|_{L^\infty_T}^2)\geq Y_t$.
    Then for all $s\in[0,t]$ and $\mu\in\A_s$ we have $m\leq \E_\mu[Y_t|\F_s]\leq M$.
    Take the essential supremum over $\mu\in\A_s$ to infer $J_{s,t}\in L^\infty$ with 
    \begin{equation}
    \label{EQ:TH:convex_quadratic:L1_1_b}
        m\leq J_{s,t}\leq M.
    \end{equation}
    
    Define $\mathcal T:=\mathcal T_1\cap \mathcal T_2$, so that $\mathcal T\subseteq[0,T)$ is Borel and has full $\ell_1$-measure.
    Fix $t\in\mathcal T$ and $s\in[0,t]$.
    By \eqref{EQ:TH:convex_quadratic:step_J_2_bis} and $(x\vee y)^+\leq x^++y^+$
    \begin{equation}
    \label{EQ:TH:convex_quadratic:L1_2}
        \left(
            \esssup_{\mu\in\mathcal B_\eps}V_\eps(\mu)-J_{s,t}
        \right)^+
        \leq 
            \left(
                \esssup_{\mu\in\mathcal C_\eps}V_\eps(\mu)-J_{s,t}
            \right)^+
             +
            \left(C_0-\frac1{2\sqrt\eps}-J_{s,t}\right)^+.
    \end{equation}
    Since $-1/(2\sqrt\eps)\to-\infty$, and $J_{s,t}$ is uniformly bounded by~\eqref{EQ:TH:convex_quadratic:L1_1_b}, the last positive part in~\eqref{EQ:TH:convex_quadratic:L1_2} is identically $0$ for sufficiently small $\eps>0$.
    Additionally, the first positive part can be controlled by the upper bound~\eqref{EQ:TH:convex_quadratic:upper_3}:
    \begin{align}
        \left(
            \esssup_{\mu\in\mathcal C_\eps}V_\eps(\mu)-J_{s,t}
        \right)^+
        \leq
        \left|
            \esssup_{\mu\in\A_s}\E_\mu\big[\langle Y_t\rangle_\eps \big|\F_s\big]
            -J_{s,t}
        \right|
        +C_1\,\varpi(\eps).
    \end{align}
    Therefore, taking expectations in \eqref{EQ:TH:convex_quadratic:L1_2} and using the convergence~\eqref{EQ:TH:convex_quadratic:step_limitupperbound_4} from \textit{Step} {$6$}:
    \begin{equation}
    \label{EQ:TH:convex_quadratic:L1_3}
        \left(
        \esssup_{\mu\in\mathcal B_\eps}V_\eps(\mu)-J_{s,t}
        \right)^+
        \longrightarrow0,
        \qquad \text{in }L^1,\text{ as }\eps\to0^+.
    \end{equation}

    Finally, recalling~\eqref{EQ:TH:convex:step2_bis_quadratic} from \textit{Step} {2} and using $|x|=x^++x^-$, we have
    \begin{equation}
    \label{EQ:TH:convex_quadratic:L1_3_bis}
    \begin{aligned}
        \left|
            \frac1\eps\r_s(\Delta_\eps\pi_t)-J_{s,t}
        \right|
        =
        \left(
            \esssup_{\mu\in\mathcal B_\eps}V_\eps(\mu)-J_{s,t}
        \right)^+
        +
        \left(
            \esssup_{\mu\in\mathcal B_\eps}V_\eps(\mu)-J_{s,t}
        \right)^-.
    \end{aligned}
    \end{equation}
    Taking expectations and using the convergences~\eqref{EQ:TH:convex_quadratic:L1_3} and~\eqref{EQ:TH:convex_quadratic:L1_1} of the positive and negative parts, we infer
    \begin{equation}
    \label{EQ:TH:convex:fine}
        \lim_{\eps\to 0^+}\E\left[\left|
            \frac1\eps\r_s(\Delta_\eps\pi_t)-J_{s,t}
        \right|\right]
        =0.
    \end{equation}
    Therefore $\mathcal D_s^\r\pi_t=J_{s,t}$. 
    By Remark~\ref{REM:As_or_Ast}, this proves the representation~\eqref{EQ:TH:resilience_convex:thesis} over $\partial^0 \!g$. 
    The bounds~\eqref{EQ:TH:convex_quadratic_bounds_resilience} follow from~\eqref{EQ:TH:convex_quadratic:L1_1_b}.
\end{proof}

\begin{proof}[\textbf{Proof of Corollary
\ref{COR:from_convex_to_coherent}.}]
    Since $g$ is $\P\otimes\ell_1$-a.e. finite, convex, and positively homogeneous in $z$,
    it is $\P\otimes\ell_1$-a.e. sublinear in $z$.
    Hence, by \cite[Theorem~8.24, Corollary~8.25]{Rockafellar+Wets_1998_Variational_analysis}, $\P\otimes\ell_1$-a.e., for all $z,q\in\R^m$
    \begin{equation}
    \label{EQ:COR:sublinear_representation}
        g(\,\cdot\,,z)
        =
        \sup_{q\in\partial_z g(\,\cdot\,,0)} q\cdot z,
        \qquad \qquad 
        g^\ast(\,\cdot\,,q)
        =
        \begin{cases}
            0,
            & q\in\partial_z g(\,\cdot\,,0),\\
            +\infty,
            & q\notin\partial_z g(\,\cdot\,,0).
        \end{cases}
    \end{equation}
    Moreover,~\itemref{IT:GA} is automatically satisfied by Remark~\ref{REM:GA}\itemref{IT:REM:GA:pos_hom}.

    Set $q^\s:=\Theta(\,\cdot\,,\s)$, where $\Theta$ is the predictable selector introduced in Lemma~\ref{LEM:g*}.
    By Lemma~\ref{LEM:g*}\itemref{IT:LEM:g*:selector_subgrad} and by~\eqref{EQ:COR:sublinear_representation},  there exists a $\P\otimes\ell_1$-null set $N\in\mathcal P$ outside which,  \eqref{EQ:COR:sublinear_representation} holds for every $z,q\in\R^m$, together with 
    $q^\s\in\partial_zg(\,\cdot\,,0)$ and $q^\s\cdot\s=g(\,\cdot\,,\s)$.
    For $r\in[0,T]$, set $N_r:=\{\omega\in\Omega:(\omega,r)\in N\}$.
    By Fubini's theorem, $r\mapsto\P(N_r)$ is Borel measurable and $\int_0^T\P(N_r)\d r= (\P\otimes\ell_1)(N) = 0$.
    Therefore, letting $\mathcal T_0$ denote the full $\ell_1$-measure Borel set from Theorem \ref{TH:resilience_convex}, 
    \[ 
        \mathcal T := \mathcal T_0 \setminus \{r\in[0,T):\P(N_r)>0\} 
    \] 
    is a Borel subset of $\mathcal T_0$ with full $\ell_1$-measure. Moreover, for every $t\in\mathcal T$, $\P$-a.s., for all $z\in\R^m$:
    \begin{equation} 
    \label{EQ:COR:pointwise_selector}
        g(t,z) = \sup_{q\in\partial_zg(t,0)}q\cdot z, \qquad
        q^\s_t\in\partial_zg(t,0), \qquad q^\s_t\cdot\s_t
        =g(t,\s_t).
    \end{equation}
    Since $\mathcal T\subseteq\mathcal T_0$, Theorem~\ref{TH:resilience_convex} now yields, for every $t\in\mathcal T$ and $s\in[0,t]$,
    \begin{equation}
    \label{EQ:COR:from_theorem}
        \mathcal D^\r_s\pi_t
        =
        \esssup_{\mu\in\partial^0 \!g}
        \E_\mu\left[
            b_t+g(t,\s_t)
            \,\big|\,\F_s
        \right],
    \end{equation}
    where $\partial^0 \!g$ is defined in~\eqref{EQ:A^0}.

    Fix $\mu\in\partial^0_t \!g$.
    Since $\mu_t\in\partial_zg(t,0)$, the pointwise support-function representation~\eqref{EQ:COR:pointwise_selector} yields $\mu_t\cdot\s_t\leq g(t,\s_t)$, $\P$-a.s.
    Moreover, $\partial^0_t \!g\subseteq\partial^0 \!g$. 
    Consequently, $\E_\mu\left[ b_t+\mu_t\cdot\s_t \big|\F_s \right] \leq \mathcal D^\r_s\pi_t$.
    Taking the essential supremum over $\mu\in\partial^0_t \!g$ yields
    \begin{equation}
    \label{EQ:COR:coherent_upper}
        \esssup_{\mu\in\partial^0_t \!g}
        \E_\mu\left[
            b_t+\mu_t\cdot\s_t
            \,\big|\,\F_s
        \right]
        \leq
        \mathcal D^\r_s\pi_t.
    \end{equation}

    Conversely, fix $\mu\in\partial^0 \!g$ and define the predictable representative $\widetilde\mu
        :=
        \mu\1_{[0,T]\setminus\{t\}}
        +
        q^\s\1_{\{t\}}$.
    Since $\Om\times\{t\}$ is predictable,
    $\widetilde\mu$ is predictable.
    Moreover, $\widetilde\mu=\mu$,
    $\P\otimes\ell_1$-a.e., while
   ~\eqref{EQ:COR:pointwise_selector} gives $\widetilde\mu_t=q^\s_t\in\partial_zg(t,0)$, $\P$-a.s.
    Thus $\widetilde\mu\in\partial^0_t \!g$.
    Since $\widetilde\mu$ and $\mu$ differ only on
    $\Om\times\{t\}$, their stochastic integrals coincide and hence $\Q^{\widetilde\mu}=\Q^\mu$.
    Therefore, by~\eqref{EQ:COR:pointwise_selector},
    \[
        \E_\mu\left[
            b_t+g(t,\s_t)
            \,\big|\,\F_s
        \right]
        =
        \E_{\widetilde\mu}\left[
            b_t+\widetilde\mu_t\cdot\s_t
            \,\big|\,\F_s
        \right]
        \leq
        \esssup_{\nu\in\partial^0_t \!g}
        \E_\nu\left[
            b_t+\nu_t\cdot\s_t
            \,\big|\,\F_s
        \right].
    \]
    Taking the essential supremum over $\mu\in\partial^0 \!g$ and using
   ~\eqref{EQ:COR:from_theorem} proves the reverse inequality to
   ~\eqref{EQ:COR:coherent_upper}, and hence
   ~\eqref{EQ:COR:resilience_coherent}.
\end{proof}

\section{Examples of Resilience Evaluations}
\label{SEC:examples}
In this section, we consider several examples for the abstract formula obtained in Theorem~\ref{TH:resilience_convex}.
We first discuss the case in which  $\pi$ is generated by a BSDE.
We then specialize the driver $g$, which leads to several relevant examples for the outer risk measure $\rho$.
Finally, we show that resilience evaluations may be ill-defined outside our framework.

\subsection{BSDE-induced inner processes}

We consider a general driver $g$ satisfying~\itemref{IT:g_lip} and the associated dynamic risk measure $\rho$ as in the setting of Theorem~\ref{TH:resilience_convex}.
We now specify the dynamics for the inner process $\pi$, assuming that it is generated by a BSDE.
Let $X\in L^2(\F_T)$ and let $(\pi,Z^\pi)$ be the solution of the BSDE
\[
    \pi_u
    =
    X+\int_u^T g^\pi(r,\pi_r,Z^\pi_r)\d r
    -\int_u^T Z^\pi_r\cdot \d W_r,
    \qquad u\in[0,T],
\]
for some driver $g^\pi$ that is $\P\otimes\ell_1$-a.e. uniformly Lipschitz in $(y,z)$ with $g^\pi(\,\cdot\,,0,0)\in L^2_T$, so that $b=-g^\pi(\,\cdot\,,\pi,Z^\pi)\in L^2_T$ and ${\s=Z^\pi\in\mathcal H^2_T}$.
Then Theorems~\ref{TH:resilience_convex} and~\ref{TH:attainment} state that, for $\ell_1$-a.e.\ $t\in[0,T)$ and all $s\in[0,t]$ such that~\itemref{IT:GA} holds, we have 
\[
    \mathcal D^\rho_s\pi_t
    =
    \essmax_{\mu\in\partial^0 \!g}
    \E_\mu\big[
        -g^\pi(t,\pi_t,Z^\pi_t)
        + g(t,Z^\pi_t)
        \big|\F_s
    \big],
\]
where  $\partial^0 \!g$ was defined in~\eqref{EQ:A^0}.
The identity persists under~\itemref{IT:g_quad} if $\pi$ has bounded It\^o coefficients; since quadratic BSDE estimates typically give only $Z^\pi\in\BMO_T$, this must be checked separately, e.g.,\ in the Markovian FBSDE setting of \cite[Appendix~B]{Laeven+Ferrari+Rosazza+Zullino_2025_measuringfinancialresilienceusing}.
Beyond the Markovian case, one may use the identification $Z^\pi_u=D_u\pi_u$ of $Z^\pi$ with the Malliavin derivative $D$ of $\pi$, which reduces bounds on $Z^\pi$ to bounds on the solution of a linear BSDE; see \cite[Sections~5 and~7]{Mastrolia+Possamai+Reveillac_2017_Malliavin_differentiability_BSDEs} 
for conditions ensuring the Malliavin differentiability of Lipschitz and quadratic BSDEs, respectively.

This formula shows explicitly the interaction between the two BSDE drivers.
The driver $g^\pi$ determines the local drift of the inner risk process $\pi$, whereas the outer driver $g$ weights the exposure $Z^\pi_t$ to the Brownian motion through the external risk measure $\rho$.

\begin{exampletext}[Black--Scholes replicating portfolio]
\label{EX:BeS}
Assume for simplicity that $m=1$ and consider a Black--Scholes market with numéraire $B_t=\exp\big(\int_0^t r_u\d u\big)$ and risky asset
\[
    S_t = S_0 + \int_0^t \alpha_uS_u\d u+\int_0^t\varsigma_uS_u\d W_u, \qquad t\geq 0,
\]
where $S_0>0$ and $r,\alpha,\varsigma$ are bounded predictable processes, with $\varsigma>0$ bounded away from zero. 
Set $\vartheta:=(r-\alpha)/\varsigma$.
Then $\vartheta$ is bounded and $W^\vartheta:=W-\int_0^{\cdot}\vartheta_u\d u$ is a $\Q^\vartheta$-Brownian motion by Girsanov's theorem.
Let $X\in L^2(\F_T)$ be the payoff of a European option with exercise time $T>0$. 
Since the market is complete, the replicating price process is
\[
    \pi_t
    :=
    \E_\vartheta\big[B_tB_T^{-1}X\big|\F_t\big],
    \qquad t\in[0,T].
\]
By the martingale representation theorem, there exists a predictable process $Z^\pi$ such that $(\pi,Z^\pi)$ is the solution of the BSDE with parameters $(g^\pi,T,X)$, where the driver is
\[
    g^\pi(\om,u,y,z)
    =
    -r_u(\om) y+\vartheta_u(\om) z, \qquad (\om,u,y,z)\in\Om\times[0,T]\times\R\times\R,
\]
and is uniformly Lipschitz in $(y,z)$ with $g^\pi(\,\cdot\,,0,0)=0$.
The unique self-financing replicating strategy is given by $\Delta:=Z^\pi / (\varsigma S)$, $\eta:=(\pi-\Delta S) / B$, so that $\pi=\eta B+\Delta S$.

Therefore, if the outer dynamic risk measure $\rho$ is induced by a driver $g$ satisfying~\itemref{IT:g_lip}, Theorems~\ref{TH:resilience_convex} and~\ref{TH:attainment} yield, for $\ell_1$-a.e. $t\in[0,T)$ and all $s\in[0,t]$, under the condition \itemref{IT:GA}:
\begin{equation}
\label{EQ:EX:BeS}
\begin{aligned}
    \mathcal D^\rho_s\pi_t
    &=
    \essmax_{\mu\in\partial^0 \!g}
    \E_\mu\big[
        r_t\pi_t-\vartheta_t Z^\pi_t
        +g(t,Z^\pi_t)
        \big|\F_s
    \big]\\
    &=
    \essmax_{\mu\in\partial^0 \!g}
    \E_\mu\big[
        r_t\eta_tB_t + \a_t \Delta_t S_t + g(t,\varsigma_t\Delta_t S_t)
        \big|\F_s
    \big].
\end{aligned}
\end{equation}
The last formula separates three contributions: the risk-free rate $r$ earned on the cash component $\eta B$ of the portfolio, the local rate of return $\a$ of the risky component $\Delta S$, and an additional term that combines the volatility $\varsigma$ of the risky component with the preferences of the evaluator encoded in the driver $g$.

Since $b=r\pi-\vartheta Z^\pi$, at $s=t$ formula~\eqref{EQ:EX:BeS} reads
\[
    \mathcal D^\rho_t\pi_t=r_t\pi_t+\big(g(t,Z^\pi_t)-\vartheta_tZ^\pi_t\big).
\]
The first term is the interest accrued by the value of the portfolio, which is all that a risk-neutral evaluator detects: for the linear driver $g(z)=\vartheta z$, one has $\rho_s=\E_\vartheta[\,\cdot\,|\F_s]$ and $\mathcal D^\rho_t\pi_t=r_t\pi_t$.
The second term measures by how much the risk adjustment of the evaluator departs from the market price for the exposure $Z^\pi_t$.

For computational purposes, assume $r,\a,\varsigma$ constant in time and deterministic, so that $B_t=e^{rt}$ and $S_t=S_0\exp[(\a-\varsigma^2/2)t+\varsigma W_t]$.
Also, let $X=(K-S_T)^+$ be the payoff of a European put option with strike price $K>0$.
Then, for $t\in[0,T)$:
\[
    \pi_t=
    Ke^{-r(T-t)}\mathcal N(d^-_t)-S_t\mathcal N(d^+_t), \qquad \Delta_t=-\mathcal N(d^+_t), \qquad \eta_t=Ke^{-rT}\mathcal N(d^-_t),
\]
where $\mathcal N$ is the standard normal cdf and 
\[
    d_t^{\pm}:=\big[\log(K/S_t) - \left(r\pm\varsigma^2/2\right)(T-t)\big]/\big[\varsigma\sqrt{T-t}\big].
\]
Suppose that $\rho=\mathfrak e^\gamma$ is the entropic risk measure with risk-aversion coefficient $\gamma>0$, so that $g(z)=\gamma z^2/2$ (see Section~\ref{SEC:entropic_RM}).
It follows from the non-negativity of the put price $\pi$ that $b=-g^\pi(\,\cdot\,,\pi,Z^\pi)$ and $\sigma=Z^\pi$ are essentially bounded processes, which motivates the application of Theorem~\ref{TH:resilience_convex} under~\itemref{IT:TH:convex:quad} for the driver $g(z)=\gamma z^2/2$ (equivalently, Theorem \ref{TH:resilience_entropy}).
Indeed, $0\leq\pi_t\leq Ke^{-r(T-t)}$ gives $S_t\mathcal N(d_t^+)\leq Ke^{-r(T-t)}\mathcal N(d_t^-)\leq Ke^{|r|T}$, hence $|\s_t|=|Z^\pi_t|=\varsigma S_t\mathcal N(d^+_t)\leq\varsigma Ke^{|r|T}$ and $|b_t|=|r\pi_t-\vartheta Z^\pi_t|\leq(|r|+|\vartheta|\varsigma)Ke^{|r|T}$.

We then can compute explicitly the conditional expectation appearing in~\eqref{EQ:EX:BeS}, obtaining
\begin{gather} 
    \begin{aligned}
        \mathcal D_s^{\mathfrak e^\gamma} \pi_t 
        &\!=\!
        -\alpha S_s e^{\alpha(t-s)}\mathcal N(D_{s,t}^+) + rKe^{-r(T-t)}\mathcal N(D_{s,t}^-)\\ &\quad + \frac{\gamma}{2}\varsigma^2 S_s^2 e^{(2\alpha+\varsigma^2)(t-s)} \mathcal N^{\,2}_{\varrho} \big(D^{++}_{s,t}, D^{++}_{s,t}\big),
    \end{aligned}\\
    D_{s,t}^{\pm} 
    := d_s^\pm + \frac{t-s} {\sqrt{T-s}}\vartheta, 
    \qquad 
     D^{++}_{s,t} 
     := d_s^+ + \frac{t-s} {\sqrt{T-s}}(\vartheta-\varsigma),
     \qquad 
     \varrho 
     := \frac{t-s}{T-s},
\end{gather}
where $\mathcal N^2_\varrho$ is the bivariate normal cdf with correlation $\varrho$.
If $s=t$, the formula specializes to
\begin{align}
    \mathcal D_t^{\mathfrak e^\gamma}\pi_t
    =\a (\Delta_tS_t)+r(\eta_tB_t)+ \frac\g2\varsigma^2(\Delta_tS_t)^2.
\end{align}
The sign of $\mathcal D_t^{\mathfrak e^\gamma}\pi_t$ has a transparent reading for the option writer, whose liability is precisely $\pi$.
Since $\Delta_t\leq0$ for a put, a rising underlying $\alpha>0$ makes $\alpha(\Delta_tS_t)\leq0$, which contributes to recovery: the liability reduces as the put moves further out of the money.
Since $\eta_t\geq0$ for a put, the term $r(\eta_tB_t)$ has the sign of $r$: assuming $r>0$, it contributes to deterioration, reflecting the fact that the long cash position $\eta_tB_t\geq0$ of the replicating portfolio accrues interest at rate $r$.
The entropic loading $\tfrac\gamma2\varsigma^2(\Delta_tS_t)^2\geq0$ also contributes to deterioration, in proportion to the  risk-aversion $\gamma$.

\end{exampletext}

\subsection{Coherent outer risk measures}
We next provide some examples of sublinear drivers $g$ satisfying the assumptions of Theorem~\ref{TH:resilience_convex}.
A sublinear driver is convex and positively homogeneous, hence it induces a coherent risk measure $\rho$; if it is also uniformly Lipschitz, it satisfies~\itemref{IT:g_lip}, so in what follows we take $b\in L^2_T$.
\begin{exampletext}[Linear driver]
    Let $g(\om,r,z)=\beta_r(\om)\cdot z$, where $\beta$ is predictable and $\P\otimes\ell_1$-essentially bounded.
    Since $g^\ast(\om,r,q)$ is zero at $q=\beta_r(\om)$ and infinite elsewhere, we have $\partial^0\!g=\{\beta\}$ and condition~\itemref{IT:GA} holds with $\psi(x)=x$.
    The induced risk measure is $\rho_r(X)=\E_\beta[X\mid\F_r]$.
    Hence, Theorem~\ref{TH:resilience_convex} gives, for $\ell_1$-a.e. $t\in[0,T)$ and every $s\in[0,t]$,
    \[
        \mathcal D^\rho_s\pi_t
        =\E_\beta[b_t+\beta_t\cdot\s_t\mid\F_s].
    \]
\end{exampletext}

\begin{exampletext}[Sublinear driver]
\label{EX:sublinear}
    Let $C(\om,r)$ be a non-empty, closed, convex, uniformly bounded subset of $\R^m$, such that the multifunction $C$ is $\mathcal P$-measurable in the sense of \cite[Definition~14.1]{Rockafellar+Wets_1998_Variational_analysis}, and set
    \[
        g(\om,r,z)=\sup_{q\in C(\om,r)}q\cdot z.
    \]
    Then, $\partial_zg(\om,r,0)=C(\om,r)$ $\P\otimes\ell_1$-a.e.; see \cite[Corollary 13.1.1]{Rockafellar_1997_Convex_analysis}.
    Since $g^\ast$ is zero on $C$ and infinite outside it, condition~\itemref{IT:GA} holds with $\psi(x)=x$. 
    Moreover,
    \begin{equation}
    \label{EQ:def_AC}
        \partial^0\!g
        =\bigl\{\mu\in\BMO_T:
          \mu_r(\om)\in C(\om,r)
          \ \P\otimes\ell_1\text{-a.e.}\bigr\}.
    \end{equation}
    The induced risk measure is coherent and satisfies $\rho_r(X)=\esssup_{\mu\in\partial^0\!g}\E_\mu[X\mid\F_r]$ for $X\in L^2(\F_T)$.
    Theorem~\ref{TH:attainment} gives, for $\ell_1$-a.e. $t\in[0,T)$ and every $s\in[0,t]$,
    \[
        \mathcal D^\rho_s\pi_t
        =\essmax_{\mu\in\partial^0\!g}
          \E_\mu\bigg[
            b_t+\sup_{q\in C(t)}q\cdot\s_t
            \bigg|\F_s\bigg]
        =\essmax_{\mu\in\partial^0_t\!g}
          \E_\mu[b_t+\mu_t\cdot\s_t\mid\F_s].
    \]
    Economically, $C(t)$ is the set of drifts of the risk factor $W$ that the evaluator regards as plausible at time $t$, as in the $\kappa$-ignorance model of \cite{Chen+Epstein_2002_Ambiguity_risk_asset_returns_continuous_time}, and the resilience evaluation is the drift of $\pi$ under the least favourable plausible drift of the factor.
    For instance, when $m=1$ and $C=[a,c]$, an exposure $\s_t>0$, whose losses increase with the factor, is charged at the highest plausible drift $c$, and an exposure $\s_t<0$ at the lowest one $a$.

    To make the latter explicit, take $m=1$ and $C(\om,r)=[a,c]$ for constants $a\leq c$. Then $g(z)=cz^+-az^-$, and the preceding formula becomes
    \[
        \mathcal D^\rho_s\pi_t
        =\essmax_{\mu\in\partial^0\!g}
          \E_\mu[b_t+c\s_t^+-a\s_t^-\mid\F_s],
        \qquad
        \mathcal D^\rho_t\pi_t=b_t+c\s_t^+-a\s_t^-.
    \]
    Positive and negative values of the volatility coefficient can therefore receive different adjustments. If $a\leq0\leq c$, both adjustments are non-negative; if $a=c$, this reduces to the linear-driver example.
\end{exampletext}

\subsection{Convex outer risk measures}
\label{SEC:EX:convex}
We now provide some examples of convex, non-coherent risk measures, induced by drivers satisfying~\itemref{IT:g_quad} that are not positively homogeneous.
In what follows, we work with a fixed $\pi$ with $b,|\s|\in L^\infty_T$.
\begin{exampletext}[Entropy]
\label{EX:entropy}
    Let $g(z)=\frac{\gamma}{2}|z|^2$, with $\gamma>0$, so that $\rho=\mathfrak e^\gamma$.
    Since $g^\ast(q)=|q|^2/(2\gamma)$, we have $\partial^0\!g=\{0\}$ and condition~\itemref{IT:GA} holds with $\psi(x)=x/(2\gamma)$.
    Theorem~\ref{TH:resilience_convex} therefore gives $\mathcal D^{\mathfrak e^\gamma}_s\pi_t=\E\big[b_t+\frac\gamma2|\s_t|^2\,\big|\,\F_s\big]$ for $\ell_1$-a.e. $t\in[0,T)$ and all $s\in[0,t]$, which agrees with the formula obtained in Theorem~\ref{TH:resilience_entropy} under weaker  assumptions.

    Suppose in addition that the inner process is also an entropic evaluation: $\pi=\mathfrak e^\delta(X)$ for some $\delta>0$ and
    $X\in L^{\mathrm{exp}}(\F_T)$. 
    Denote by $Z^\delta$ the martingale coefficient of the associated BSDE and assume $Z^\delta\in L^\infty_T$.
    Then $\sigma=Z^\delta$ and $b=-\delta|Z^\delta|^2/2$, so
    \[
        \mathcal D^{\mathfrak e^\gamma}_s\big(\mathfrak e^\delta_t(X)\big)
        = 
        \frac{\gamma-\delta}2 \E\big[|Z^\delta_t|^2\big|\F_s\big]
        =-\frac{\gamma-\delta}\delta\,\bm\dot{\mathfrak e}^\delta_{t|s}(X),
    \]
    where the last equality uses the conditional resilience rate of Example~\ref{EX:resilience_rate}.
    It shows that the entropic resilience evaluation of $\mathfrak e_t^\delta(X)$ is a rescaled version of its resilience rate.
    The sign interpretation of Section~\ref{SEC:res_eval_interpretation} illustrates the different roles of $\mathfrak e^\gamma$ and $\mathfrak e^\delta(X)$.
    A more risk-averse outer evaluator than the one already embedded in $\pi$, $\gamma>\delta$, always reads $\pi$ as weakly deteriorating, $\mathcal D^{\mathfrak e^\gamma}_s(\mathfrak e^\delta_t(X))\geq0$: additional risk-aversion, layered on top of an already risk-adjusted position, can only uncover further risk. 
    Conversely, a less risk-averse evaluator, ${\gamma<\delta}$, always perceives a prospect of recovery, $\mathcal D^{\mathfrak e^\gamma}_s(\mathfrak e^\delta_t(X))\leq0$. 
    At the threshold $\gamma=\delta$ the entropic resilience evaluation vanishes, $\mathcal D^{\mathfrak e^\gamma}_s(\mathfrak e^\delta_t(X))=0$: a risk measure cannot detect additional risk in a claim already evaluated through that same risk measure (see Remark~\ref{REM:self_evaluation}).
\end{exampletext}
\begin{exampletext}[Coherent plus quadratic: robust entropy]
    Let $g^0$ and its support set $C$ be as in
    Example~\ref{EX:sublinear}, and fix $\gamma>0$. Consider
    \[
        g(\om,r,z)=g^0(\om,r,z)+\frac\gamma2|z|^2,\qquad (\om,r,z)\in\Om\times[0,T]\times\R^m.
    \]
    Since the quadratic term has derivative zero at the origin,
    $\partial_zg(\om,r,0)=C(\om,r)$. Moreover, inf-convolution
    (see \cite[Theorem~16.4]{Rockafellar_1997_Convex_analysis}) gives
    \[
        g^\ast(\om,r,q)
        =
        \inf_{p\in C(\om,r)}
        \frac{1}{2\g}|q-p|^2
        =
        \frac{1}{2\g}
        \operatorname{dist}^2\big(q,C(\om,r)\big), \qquad (\om,r,q)\in\Om\times[0,T]\times\R^m.
    \]
    Thus condition~\itemref{IT:GA} holds with $\psi(x)=x/(2\gamma)$ for all $0\leq s \leq t <T$.
    The zero-penalty class $\partial^0\!g$ is precisely the class from~\eqref{EQ:def_AC}.
    For $X\in L^\infty(\F_T)$, the induced risk measure has the robust entropic representation
    \begin{equation}
        \rho_r(X)
        \!
        =
        \!
        \essmax_{\mu\in\BMO_T}
            \E_\mu
            \!\!
            \left[
            \!
            X 
            \!
            - 
            \!
            \frac{1}{2\g}
            \!
            \int_r^T
                \!\!\!\!
                \operatorname{dist}^2(\mu_u,C(u))
                \d u
                \bigg|\F_r
            \right]
        \!
        =
            \!
            \frac1\g 
            \!
            \essmax_{p\in \partial^0\!g}
            \log \E_{p}\big[\exp(\gamma X)\big|\F_r\big],
    \end{equation}
    see \cite[Sections~6.1 and~7.3]{Barrieu+ElKaroui_2009_Pricing_hedging_optimally_designing_derivatives_minimization_risk_measures}.
    Theorems~\ref{TH:resilience_convex} and~\ref{TH:attainment} therefore yield, for $\ell_1$-a.e. $t$ and all ${s\in[0,t]}$ such that~\itemref{IT:GA} holds,
    \[
        \mathcal D^\r_s\pi_t
        =
        \essmax_{\mu\in\partial^0\!g}
        \E_\mu\bigg[b_t+\sup_{q\in C(t)}
        q\cdot\s_t+\frac\g 2|\s_t|^2\bigg|\F_s\bigg]
        =
        \essmax_{\mu\in\partial^0_t\!g}
        \E_\mu\bigg[b_t+\mu_t\cdot\s_t+\frac\g 2|\s_t|^2\bigg|\F_s\bigg],
    \]
    where the last equality can be proved as in Corollary~\ref{COR:from_convex_to_coherent}.
\end{exampletext}
\begin{exampletext}[Subquadratic]
    Fix $p\in(1,2]$ and let $a:\Om\times[0,T]\to(0,\infty)$ be predictable and $\P\otimes\ell_1$-essentially bounded. 
    Consider
    \[
        g(\om,r,z)=a_r(\om)|z|^p,\qquad (\om,r,z)\in\Om\times[0,T]\times\R^m.
    \]
    Its convex conjugate is
    \[
        g^\ast(\om,r,q)
        =
        c_r(\om)|q|^{\frac{p}{p-1}},  
        \quad 
        c_r(\om):=\frac{p-1}{p}\,\big(a_r(\om)p\big)^{-\frac{1}{p-1}},
        \quad
        (\om,r,q)\in\Om\times[0,T]\times\R^m,
\] 
    hence $\partial_z g(\om,r,0)=\{0\}$.
    Since $a\leq \|a\|_{L^\infty_T}$, we have $c_r(\om)\geq \frac{p-1}{p}\,(p\|a\|_{L^\infty_T})^{-\frac{1}{p-1}}=:\underline c>0$.
    Hence, setting $p':=p/(p-1)$,
    \[
        g^\ast(\om,r,q)
        \geq
        \underline c |q|^{p'}
        =
        \underline c\,
        \Big(
        \operatorname{dist}^2
        \big(q,\partial_zg(\om,r,0)\big)
        \Big)^{p'/2},
    \]
    so that the condition~\itemref{IT:GA} is satisfied for all times with $\psi(x) = \underline c\, x^{p'/2}$, which is continuous, strictly increasing, and convex since $p'\geq 2$.
    The zero-penalty class is $\partial^0\!g=\{0\}$.
    Theorem~\ref{TH:resilience_convex} therefore yields
    \[
        \mathcal D^\r_s\pi_t
        =
        \E\big[b_t+a_t|\s_t|^p\big|\F_s\big], \qquad \ell_1\text{-a.e. }t\in[0,T), \ \forall\,s\in[0,t].
    \]
    The interpretation is analogous to that of the entropic resilience evaluation: the coefficient $a$ plays the role of a local risk-aversion parameter, while the volatility exposure $\sigma$ is penalized through the power $p$ rather than quadratically. 
    Compared with the entropic case, for $p < 2$ small volatility exposures are penalized relatively more, and large ones relatively less.
\end{exampletext}

\subsection{Counterexamples}
We now provide counterexamples showing that the existence of the resilience evaluation may fail outside our framework.
We consider separately two types of failures: first, risk measures such as Value at Risk and Expected Shortfall (\cite[Ch.~4]{Follmer+Schied_2016_Stochastic_finance}), which do not admit a BSDE representation; second, BSDE-induced risk measures that are not normalized or cash-additive.

\begin{exampletext}[VaR and ES]
    Let $\a\in(0,1)$ and define, for $X\in L^1$,
    \[
        \operatorname{VaR}_\a(X)
        :=\inf\{x\in\R:\P(X\leq x)>\a\},
        \qquad
        \operatorname{ES}_\a(X)
        :=\frac{1}{1-\a}\int_\a^1\operatorname{VaR}_u(X)\d u.
    \]
    Setting $\pi=W$ and using that $W_{t+\eps}-W_t =: \sqrt\eps Z$, where $Z$ has standard normal distribution with cdf $\Phi$ and pdf $\varphi$, we obtain
    \[
        \frac{1}{\eps}\operatorname{VaR}_\a(\pi_{t+\eps}-\pi_t)
        =\eps^{-1/2}\Phi^{-1}(\a),
        \qquad 
        \frac{1}{\eps}\operatorname{ES}_\a(\pi_{t+\eps}-\pi_t)
        =\eps^{-1/2}
        \frac{\varphi\bigl(\Phi^{-1}(\a)\bigr)}{1-\a}.
    \]
    They both diverge as $\eps\to 0^+$, for all $\a\neq 1/2$.
\end{exampletext}

\begin{exampletext}[Necessity of normalization]
    Let us consider a BSDE-induced dynamic risk measure $\rho$ associated with a driver $g$ satisfying all the assumptions in~\itemref{IT:g_lip}, except for $g(\,\cdot\,,0)= 0$.
    Assume instead that $\P\otimes\ell_1\big(g(\,\cdot\,,0)\neq 0\big)>0$.
    Then $\rho$ is not normalized: there exists $s\in[0,T)$ with $\rho_s(0)\neq 0$.
    Indeed, if $\rho_s(0)=0$ for all $s\in[0,T]$, then the solution $(Y,Z)$ of the BSDE with parameters $(g,T,0)$ satisfies $Y\equiv0$; hence the continuous martingale $(Z\tinybullet W)$ has finite variation, so that $Z=0$, and the equation reduces to $\int_s^Tg(r,0)\d r=0$ for all $s$, that is, $g(\,\cdot\,,0)=0$, $\P\otimes\ell_1$-a.e.
    As a consequence, the resilience evaluation through this $\rho$ of even the simplest, constant It\^o process $\pi\equiv0$ is ill-posed.
    Indeed, for all $t\in[s,T)$, we have
    \[
        \big\|\eps^{-1} \rho_s(\pi_{t+\eps}-\pi_t)\big\|_{L^1}
        =
        \eps^{-1} \|\rho_s(0)\|_{L^1}
        \longrightarrow +\infty.
    \]
    The obstruction is exactly the lack of normalization: the null increment is assigned a non-zero risk, which is then amplified by the scaling factor $\eps^{-1}$.
\end{exampletext}

\begin{exampletext}[Failure without outer cash-additivity]
    Let $a>0$ and consider the Lipschitz driver $g(y)=a|y|$ for $y\in\R$.
    The induced dynamic risk measure 
    \[
        \rho_s(X)
        =
        \esssup_{\mu\in\A^a}
        \E\Big[
            e^{\int_s^T \mu_r\d r}X
            \Big|\F_s
        \Big], 
        \qquad X\in L^2(\F_T), \ s\in[0,T],
    \]
    where $\A^a:= \big\{\mu:\Om\times[0,T]\to[-a,a] \text{ predictable}\big\}$, is not cash-additive (see also the discussion in \cite{ElKaroui+Ravanelli_2009_Cash_sub-additive_risk_measures_interest_rate_ambiguity} in the cash-subadditive case).
    Economically, $\rho$ discounts the position at an ambiguous interest rate $\mu_r\in[-a,a]$, as in \cite{ElKaroui+Ravanelli_2009_Cash_sub-additive_risk_measures_interest_rate_ambiguity}; the worst-case discounting depends on the sign of the increment, and this asymmetry is of order $\sqrt\eps$ rather than $\eps$.
    Take the simple It\^o process $\pi=W$ on $\Om\times[0,T]$, and fix $0\leq s\leq t < T$.
    For $\eps\in(0,T-t]$, the increment $\pi_{t+\eps}-\pi_t=W_{t+\eps}-W_t=:\xi_\eps$ is normally distributed and independent of $\F_s$.
    Choosing the admissible control $\underline\mu:
        =a
        \operatorname{sign}(\xi_\eps)
        \1_{(t+\eps,T]}\in\A^a$
    gives
    \[
        \frac1\eps\rho_s(\pi_{t+\eps}-\pi_t)
        \geq
        \frac1\eps\E\left[
            e^{a(T-t-\eps)}\xi_\eps^+
            -
            e^{-a(T-t-\eps)}\xi_\eps^-
            \bigg|\F_s
        \right]
        =
        \frac{
            e^{a(T-t-\eps)}
            -
            e^{-a(T-t-\eps)}
        }{\sqrt{2\pi\eps}},
\]
    where the right-hand side diverges to $+\infty$ as $\eps\to 0^+$.
\end{exampletext}

\appendix

\section{Conditional asymptotic expansions}
\label{SEC:app_taylor}

\begin{lemma}
\label{LEM:Taylor_rv}
    Let $(X_n)_{n\in\N}$ be a sequence of random variables and $f:\R\to\R$ a twice continuously differentiable function.
    Assume the following.
    \begin{enumerate}[label=(\roman*)]
        \item \label{IT:LEM:Taylor_rv:HP_exp} There exist $c>0$, $C\geq0$ such that $|f(x)|\leq C e^{c|x|}$ for all $x\in\R$.
        \item \label{IT:LEM:Taylor_rv:HP_conv} There exists $p>2$ such that $X_n\longrightarrow0$ in $L^p(\Om)$.
        \item \label{IT:LEM:Taylor_rv:HP_unif_bound} There exist $q>p/(p-2)$ and $K\geq 0$ such that $\E\left[e^{qc|X_n|}\right]\le K$ for all $n\in\N$.  
    \end{enumerate}
    Then, for any sub-$\s$-algebra $\mathcal G \subseteq\F$ and all $n\in\N$, we have $\P$-a.s.
    \[
        \E[f(X_n)|\mathcal G]=f(0)+f'(0)\E[X_n|\mathcal G] + \frac12f''(0)\E[X_n^2|\mathcal G]+R_n,
    \]
    where $R_n\in L^1(\mathcal G)$, for $n\in\N$, and  $\|R_n\|_{L^1}=o\left(\|X_n\|^2_{L^p}\right)$, as $n\to\infty$.
\end{lemma}
\begin{proof}
    Let us define $R:\R\to\R$ as the Taylor remainder of order $2$ for $f$:
    \[
        R(x):=f(x)-f(0)-f'(0)x - \frac12f''(0)x^2,\qquad\forall\,x\in\R,
    \]
    and define $R_n:=\E[R(X_n)|\mathcal G]$, for $n\in\N$.
    Thanks to hypotheses~\itemref{IT:LEM:Taylor_rv:HP_exp},\itemref{IT:LEM:Taylor_rv:HP_unif_bound} we have $f(X_n)\in L^1$, while thanks to~\itemref{IT:LEM:Taylor_rv:HP_conv} we have $X_n,X_n^2\in L^1$. 
    Consequently, $R_n\in L^1(\mathcal G)$.
    The statement of the lemma follows from
    \[
        \frac{\E[|R(X_n)|]}{\|X_n\|^2_{L^p}}\longrightarrow 0, \qquad \text{as }n\to\infty,
    \]
    because $\|R_n\|_{L^1}\leq \E[|R(X_n)|]$. 
    
    Fix \(\delta>0\). 
    By Taylor's formula,
    \[
        R(x)
        =
        x^2\int_0^1(1-u)\big(f''(ux)-f''(0)\big)\d u,
    \]
    and therefore
    \[
        \E\!\left[|R(X_n)|\1_{\{|X_n|\leq\delta\}}\right]
        \leq
        \varpi(\delta)\|X_n\|_{L^p}^2,
        \qquad
        \varpi(\delta):=\sup_{|x|\leq\delta}|f''(x)-f''(0)|.
    \]
    Moreover, hypothesis~\itemref{IT:LEM:Taylor_rv:HP_exp} implies $|R(x)|\leq C'e^{c|x|}$ for $x\in\R$ and some \(C'>0\). 
    Hence, by Hölder's and Markov's inequalities,
    \[
    \begin{aligned}
        \E\!\left[|R(X_n)|\1_{\{|X_n|>\delta\}}\right]
        \leq
        C'K^{1/q}\,\big(\P(|X_n|>\delta)\big)^{1-1/q} 
        \leq
        C'K^{1/q}
        \delta^{-p(1-1/q)}
        \|X_n\|_{L^p}^{p(1-1/q)}.
    \end{aligned}
    \]
    Consequently,
    \[
        \frac{\E[|R(X_n)|]}{\|X_n\|_{L^p}^2}
        \leq
        \varpi(\delta)
        +
        C'K^{1/q}\delta^{-p(1-1/q)}
        \|X_n\|_{L^p}^{p(1-1/q)-2}.
    \]
    Since \(p(1-1/q)>2\), taking first \(n\to\infty\) and then \(\delta\to0^+\) proves the claim.
\end{proof}

\begin{proof}[\textbf{Proof of Proposition~\ref{PROP:Taylor_exp}.}]
    Let us set, for all $n\in\N$:
    \[
        M_n:=\E[X_n|\mathcal G], \quad 
        \tilde X_n:=X_n-M_n, \quad 
        V_n:=\E\big[\tilde X^2_n\big|\mathcal G\big], \quad 
        Y_n:=\E\big[e^{\tilde X_n}\big|\mathcal G\big]-1.
    \]
    Since $M_n$ is $\mathcal G$-measurable and $\E\big[\tilde X_n\big|\mathcal G\big]=0$, we have $\E\big[e^{\tilde X_n}\big|\mathcal G\big]=e^{-M_n}\E\big[e^{X_n}\big|\mathcal G\big]$ and $V_n=\V(X_n|\mathcal G)$.
    Therefore, the claim is equivalent to
    \begin{equation}
    \label{EQ:LEM:Ito_thesis}
        \left\|\ln(1+Y_n)-\frac12V_n\right\|_{L^1}=o\left(\|X_n\|^2_{L^p}\right), \qquad \text{ as }n\to\infty.
    \end{equation}
    
    First, $\|\widetilde X_n\|_{L^p}
        \leq2\|X_n\|_{L^p}
        \longrightarrow0$.
    Moreover, conditional Jensen's and Cauchy--Schwarz inequalities yield
    \[
        \E\Big[e^{q|\widetilde X_n|}\Big]
        \leq
        \left(\E\Big[e^{2q|X_n|}\Big]\right)^{1/2}
        \left(\E\Big[e^{2q|M_n|}\Big]\right)^{1/2}
        \leq
        \E\Big[e^{2q|X_n|}\Big]
        \leq K,
    \]
    where we used $e^{2q|M_n|}
        \leq
        \E[e^{2q|X_n|}|\mathcal G]$.
    Lemma~\ref{LEM:Taylor_rv}, applied to $f=\exp$ and $(\widetilde X_n)_{n\in\N}$, therefore gives
    \[
        Y_n=\frac12V_n+R_n,
        \qquad
        \|R_n\|_{L^1}
        =
        o\big(\|X_n\|_{L^p}^2\big).
    \]
    
    Define $h(y):=y-\ln(1+y)$ for $y\geq 0$.
    Observe that $Y_n\geq 0$ since $\E[\tilde X_n|\mathcal G]=0$ and, by conditional Jensen's inequality, $1=\exp\big(\E[\tilde X_n|\mathcal G]\big)\leq \E\big[\exp(\tilde X_n)\big|\mathcal G\big]$.
    Therefore,
    \begin{equation}
    \label{EQ:PROP:almost}
        \left\|\ln(1+Y_n)-\frac12V_n\right\|_{L^1}
        \!\!\!=\left\|-h(Y_n)+R_n\right\|_{L^1}
        \leq \|h(Y_n)\|_{L^1}+\|R_n\|_{L^1},
    \end{equation}
    where we used the definition of $h$.
    
    Choose $r\in(1,2\wedge p/2)$ such that $rp/(p-2r)<q$, which is possible because $q>p/(p-2)$. 
    Since $0\leq h(y)\leq y^r$ for $y\geq 0$, and $e^x-1-x\leq x^2e^{|x|}$ for $x\in\R$, conditional Jensen's and Hölder's inequalities yield, with $t:=p/(p-2r)>1$:
    \begin{gather}
        \|h(Y_n)\|_{L^1}
        \leq
        \E[Y_n^r] 
        \leq
        \E\left[|\widetilde X_n|^{2r}e^{r|\widetilde X_n|}\right] 
        \leq
        \|\widetilde X_n\|_{L^p}^{2r}
        \left(\E\Big[e^{rt|\widetilde X_n|}\Big]\right)^{1/t} \\
        \leq
        K^{1/t}(2\|X_n\|_{L^p})^{2r}
        =
        o(\|X_n\|_{L^p}^2).
    \end{gather}
    Together with \eqref{EQ:PROP:almost} and the estimate on $R_n$, this proves~\eqref{EQ:LEM:Ito_thesis}.
\end{proof}

\section{Technical results on convex duality and changes of measure}
\label{SEC:app_duality}

\begin{proof}[\textbf{Proof of Lemma~\ref{LEM:g*}.}]
    Let $A\in\mathcal P$ with $\P\otimes\ell_1(A)=0$ be such that the properties listed in~\itemref{IT:g_lip} or~\itemref{IT:g_quad} hold in $A^c:=(\Om\times[0,T])\setminus A$.
    On $A\times\R^m$ define $g^\ast=0$, while on $A^c\times\R^m$ define $g^\ast$ as in~\itemref{IT:LEM:g*:def}.
    Then~\itemref{IT:LEM:g*:def} holds by construction.
    
    For every $(\om,r)\in A^c$,
    the map $z\mapsto g(\om,r,z)$ is  continuous on $\R^m$ because it is finite and convex on $\R^m$.
    Hence, for every $(\om,r,q)\in \Om\times[0,T]\times\R^m$,
    \[
        g^\ast(\om,r,q)
        :=\1_{A^c}(\om,r)\sup_{z\in\R^m}\big\{q\cdot z-g(\om,r,z)\big\}
        =\1_{A^c}(\om,r)\sup_{z\in\Q^m}\big\{q\cdot z-g(\om,r,z)\big\}.
    \]
    Since $A\in\mathcal P$, and $(\om,r,q)\mapsto q\cdot z-g(\om,r,z)$ is $\mathcal P\otimes\mathscr B(\R^m)$-measurable for every fixed $z\in\Q^m$, it follows that $g^\ast$ is $\mathcal P\otimes\mathscr B(\R^m)$-measurable.

    Fix $(\om,r)\in A^c$.
    The map $q\mapsto g^\ast(\om,r,q)$ is proper, convex and lower semi-continuous as it is the convex conjugate of the finite and convex function $z\mapsto g(\om,r,z)$ (see, e.g., \cite[Theorem~12.2]{Rockafellar_1997_Convex_analysis}).
    Moreover,~\itemref{IT:LEM:g*:g=g**} holds by continuity of $z\mapsto g(\om,r,z)$.
    Since $g(\om,r,0)=0$, we also have $g^\ast(\om,r,q)\ge q\cdot 0-g(\om,r,0)=0$ for all $q\in\R^m$, so $g^\ast$ is non-negative. 
    This proves the first sentence of~\itemref{IT:LEM:g*:prop_nonnega}.
    Concerning the first bullet point in~\itemref{IT:LEM:g*:prop_nonnega}, if $g$ satisfies~\itemref{IT:g_quad}, then $q\cdot z-g(\om,r,z)\ge q\cdot z-k(1+|z|^2)$ for all $q,z\in\R^m$.
    Taking the supremum over $z\in\R^m$ yields
    \[
        g^\ast(\om,r,q)
        \ge \sup_{z\in\R^m}\big\{q\cdot z-k(1+|z|^2)\big\}
        =\frac{|q|^2}{4k}-k.
    \]
    Instead, if $g$ satisfies~\itemref{IT:g_lip}, then $q\cdot z-g(\om,r,z)\ge q\cdot z-L|z|$ for all $q,z\in\R^m$. 
    Fix $|q|>L$ and set $z_n:=n\,q/|q|$ for $n\in\N$.
    Then, 
    \[
        g^\ast(\om,r,q)
        \geq \sup_{n\in\N}\big\{q\cdot z_n-L|z_n|\big\}
        = \sup_{n\in\N}n(|q|-L)
        =+\infty.
    \]
    This completes the proof of~\itemref{IT:LEM:g*:prop_nonnega}.
    
    The subdifferential $\partial_z g(\om,r,0)$ is a non-empty, closed and convex set because $z\mapsto g(\om,r,z)$ is convex and finite on $\R^m$ (see, e.g., \cite[Theorem~23.4]{Rockafellar_1997_Convex_analysis}).
    By definition of subdifferential, $q\in \partial_z g (\om,r,0)$ if and only if 
    \begin{equation}
    \label{EQ:ggeq0}
        g(\om,r,z)\ge g(\om,r,0)+q\cdot(z-0)=q\cdot z,
        \qquad \forall\,z\in\R^m.
    \end{equation}
    Since $g(\om,r,0)=0$, this is equivalent to $q\cdot z-g(\om,r,z)\le 0$ for all $z\in\R^m$, namely to $g^\ast(\om,r,q)\le 0$. 
    It is then equivalent to $g^\ast(\om,r,q)=0$ because $g^\ast(\om,r,q)\ge 0$.
    
    It remains to prove the uniform boundedness of $\partial_z g (\om,r,0)$.
    If $\partial_zg(\om,r,0)=\{0\}$, then this is trivial.
    Thus fix
    $q\in\partial_z g(\om,r,0)$ such that $q\neq 0$.
    If $g$ satisfies~\itemref{IT:g_lip} then, recalling~\eqref{EQ:ggeq0}:
    \[
        q\cdot z \le g(\om,r,z)\leq|g(\om,r,z)|=|g(\om,r,z)-g(\om,r,0)|\le L|z|,
        \qquad \forall\,z\in\R^m.
    \]
    Choose $z=\lambda q /|q|$ for any $\lambda>0$, then $\lambda |q|\leq L\l$ for all $\lambda>0$, hence $|q|\leq L$.
    Instead, if $g$ satisfies~\itemref{IT:g_quad}, then $q\cdot z \le g(\om,r,z)\leq k(1+|z|^2)$ for all $z\in\R^m$.
    Choose again $z=\lambda q /|q|$ for any $\lambda>0$ to get  
    \[
        \lambda |q|\leq k(1+\lambda^2), \qquad \forall\,\lambda>0 
        \qquad \Longrightarrow \qquad 
        |q|\leq \inf_{\lambda>0}k\Big(\frac1\lambda+\lambda\Big)=2k.
    \]
    This proves~\itemref{IT:LEM:g*:subgrad_0}.

    We now prove~\itemref{IT:LEM:g*:projector}.
    Define
    \[
        \widetilde g(\omega,r,z):=
        \begin{cases}
            g(\omega,r,z), &(\omega,r)\in A^c,\\
            0, &(\omega,r)\in A.
        \end{cases}
    \]
    Since $A\in\mathcal P$, this is a finite Carathéodory integrand, hence a normal integrand by \cite[Example~14.29]{Rockafellar+Wets_1998_Variational_analysis}.
    By \cite[Theorem~14.56]{Rockafellar+Wets_1998_Variational_analysis},
    \[
        K(\omega,r):=\partial_z\widetilde g(\omega,r,0)
        =
        \begin{cases}
            \partial_zg(\omega,r,0), &(\omega,r)\in A^c,\\
            \{0\}, &(\omega,r)\in A
        \end{cases}
    \]
    is a closed-valued, $\mathcal P$-measurable multifunction.
    It is non-empty and convex by~\itemref{IT:LEM:g*:subgrad_0} on $A^c$, and by definition on $A$.

    Set $E:=\Omega\times[0,T]\times\R^m$, $\mathscr E:=\mathcal P\otimes\mathscr B(\R^m)$, and $\overline K(\omega,r,q):=K(\omega,r)$.
    The multifunction $\overline K$ is closed-valued and $\mathscr E$-measurable in the sense of \cite[Definition~14.1]{Rockafellar+Wets_1998_Variational_analysis}. 
    The function
    \[
        F((\omega,r,q),p)
        :=|q-p|^2+\delta_{\overline K(\omega,r,q)}(p)
    \]
    is a normal integrand by \cite[Examples~14.29 and 14.32]{Rockafellar+Wets_1998_Variational_analysis}.
    Since $K(\omega,r)$ is non-empty, closed and convex, the minimizer of $F((\omega,r,q),\cdot)$ is unique. 
    Define
    \[
        \Pi(\omega,r,q)
        :=\argmin_{p\in K(\omega,r)}|q-p|^2.
    \]
    By \cite[Theorem~14.37]{Rockafellar+Wets_1998_Variational_analysis}, $\Pi$ is $\mathscr E$-measurable. 
    By definition, it is the projector onto $\partial_zg(\omega,r,0)$ on $A^c$ and equals $0$ on $A$.

    It remains to prove~\itemref{IT:LEM:g*:selector_subgrad}.
    Define $f:E\times\R^m\to\R$ by $f((\omega,r,z),q):=\widetilde g(\omega,r,q)$.
    This is a finite Carathéodory integrand on $(E,\mathscr E)\times\R^m$, hence a normal integrand by \cite[Example~14.29]{Rockafellar+Wets_1998_Variational_analysis}.
    Apply \cite[Theorem~14.56]{Rockafellar+Wets_1998_Variational_analysis} to $f$ and to the $\mathscr E$-measurable map $x:E\ni(\omega,r,z)\mapsto z\in\R^m$ to infer that the multifunction $\Gamma:E\ni e\mapsto \partial_q f(e,x(e))\subseteq\R^m$ is closed-valued and $\mathscr E$-measurable. 
    Equivalently,
    \[
        \Gamma:\Om\times[0,T]\times\R^m\ni(\omega,r,z)\mapsto \begin{cases}
            \partial_z g(\omega,r,z), &\text{if }(\omega,r)\in A^c,\\
            \{0\}, &\text{if }(\omega,r)\in A.
        \end{cases}
    \]
    By the measurable selection theorem \cite[Corollary~14.6]{Rockafellar+Wets_1998_Variational_analysis}, there exists a $\mathcal P\otimes\mathscr B(\R^m)$-measurable selector $\Theta:\Omega\times[0,T]\times\R^m\to\R^m$ such that $\Theta(\omega,r,z)\in\Gamma(\omega,r,z)$ for all $(\omega,r,z)\in\Omega\times[0,T]\times\R^m$.
    Hence, for $(\om,r)\in A^c$ and all $z\in\R^m$, we have $\Theta(\omega,r,z)\in\partial_z g(\omega,r,z)$, which is equivalent to $g(\,\cdot\,,z)+g^\ast\big(\,\cdot\,,\Theta(\,\cdot\,,z)\big)
        =
        \Theta(\,\cdot\,,z)\cdot z$ by \cite[Theorem~23.5]{Rockafellar_1997_Convex_analysis}.
    
    It remains to prove the stated bounds on $\Theta$.
    Fix $(\omega,r)\in A^c$, $z\in\R^m$ and $q\in\partial_z g(\omega,r,z)$.
    If $q=0$, there is nothing to prove. Assume then that $q\neq0$ and set
    $u:=q/|q|$.
    By definition of subdifferential, for every $\lambda>0$ we have  $g(\omega,r,z+\lambda u)
        \geq
        g(\omega,r,z)+\lambda |q|$.
    Hence, if $g$ satisfies~\itemref{IT:g_lip}, then $\lambda |q|
        \leq
        g(\omega,r,z+\lambda u)-g(\omega,r,z)
        \leq
        L\lambda$ for all $\l>0$,
    and therefore $|q|\leq L$.
    If instead $g$ satisfies~\itemref{IT:g_quad}, then
    \[
\begin{aligned}
        \lambda |q|
        &\leq
        g(\omega,r,z+\lambda u)-g(\omega,r,z)
        \\ &\leq
        k(1+|z+\lambda u|^2)+k(1+|z|^2)
        \leq
        k\big(2+2|z|^2+2\lambda |z|+\lambda^2\big),
    \end{aligned}
\]
    where we used the triangle inequality in the end.
    Choosing
    $\lambda^2:=2+2|z|^2$ and dividing by $\lambda$ yields
    \[
        |q|
        \leq
        k\Big(2|z|+2\sqrt{2+2|z|^2}\Big)
        \leq k\Big(2|z| + 2\sqrt 2 (1+|z|)\Big)
        \leq
        5k(1+|z|),
    \]
    where we used $\sqrt{1+|z|^2}\leq 1+|z|$ and $2\leq 2+2\sqrt 2 \leq 5$.
    By arbitrariness of $q\in\partial_z g(\omega,r,z)$, and since
    $\Theta(\omega,r,z)\in\partial_z g(\omega,r,z)$, this proves the claimed bound and completes the proof of
   ~\itemref{IT:LEM:g*:selector_subgrad}.
\end{proof}

\begin{proof}[\textbf{Proof of Lemma~\ref{LEM:cond_expec_convex}.}]
{\let\qed\relax
\begin{proof}[\textbf{\itemref{IT:LEM:cond_expec_convex:Lexp_BMO}.}]
    Fix $\mu\in\BMO_T$.
    Then  there exists $r>1$ such that $\mathcal E_{T}^\mu\in L^r$, thanks to \cite[Theorem~3.4]{Kazamaki_1994_Continuous_exponential_martingales_BMO}.
    Let $r':=r/(r-1)$. 
    If $X\in L^{\mathrm{exp}}$, then, for every $c>0$, 
    \[
        \E_\mu\big[e^{c|X|}\big] 
        = \E\Big[\mathcal E_{T}^\mu \,e^{c|X|}\Big] 
        \leq \big\|\mathcal E_{T}^\mu\big\|_{L^r} \left(\E\Big[e^{cr'|X|}\Big]\right)^{1/r'} <+\infty. \] Therefore $X\in L^{\mathrm{exp}}(\Q^\mu)$. Since finite exponential moments of every order imply finite polynomial moments of every order, we also have $X\in L^p(\Q^\mu)$ for every $p\in[1,+\infty)$. 
\end{proof}
\begin{proof}[\textbf{\itemref{IT:LEM:cond_expec_convex:stoch_exponen}.}]
    Recall that $\mathcal E(M):=\exp\left(M-\frac12\langle M\rangle\right)$ denotes the stochastic exponential of any continuous martingale $M$ with $M_0=0$, where $\langle M\rangle$ is its quadratic variation process (see, e.g., \cite[Section~1.1]{Kazamaki_1994_Continuous_exponential_martingales_BMO}).
    Then for all $p>1$ and $0\leq a\leq b\leq T$, we have
    \begin{align}
        \E\left[(\mathcal E_b(M))^p|\F_a\right]
        =\E\left[\left.\mathcal E_b(pM) \exp{\left(\Big(\frac{p^2}{2}-\frac p2\Big)\langle M\rangle_b\right)}\right|\F_a\right].
    \end{align}
    Recall that, for any such $M$, $\mathcal E(M)$ is a non-negative local martingale, hence a supermartingale.
    If, in addition, $\langle M\rangle_b$ is bounded,  then Novikov's condition is satisfied and  $\mathcal E(pM)$ is a martingale.
    Therefore we reach
    \[
        \E\left[(\mathcal E_b(M))^p|\F_a\right]\leq \mathcal E_a(pM)\exp{\left(\frac12 p(p-1)\|\langle M\rangle_b\|_{L^\infty}\right)}.
    \]
    The assertion follows by applying this inequality to the martingale  $M :=\int_0^{\,\cdot}\1_{[a,T]}\mu\cdot \d W$, which satisfies $\mathcal E_b(M)=\mathcal E_{a,b}^\mu$, $\mathcal E_a(pM)=1$, and $\langle M\rangle_b\leq K^2(b-a)\leq K^2 T$.
\end{proof}}
\begin{proof}[\textbf{\itemref{IT:LEM:cond_expec_convex:embedding}.}]
    Fix $\mu\in\mathcal U$, $X\in L^p$ and $s\in[0,T]$.
    By the conditional change-of-measure formula, and the conditional H\"older inequality with conjugate exponents $p/u>1$ and $q:=p/(p-u)>1$, we have
    \[
        \E_\mu[|X|^u|\F_s]
        =\E[\mathcal E^\mu_{s,T}|X|^u|\F_s]
        \leq (\E[(\mathcal E^\mu_{s,T})^{q}|\F_s])^{1/q}(\E[|X|^p|\F_s])^{u/p}.
    \]
    The claim follows by first applying Lemma~\ref{LEM:cond_expec_convex}\itemref{IT:LEM:cond_expec_convex:stoch_exponen} to the first factor on the right-hand side, then taking the power $1/u$.
    In particular, if $s=0$, the continuous embedding of $L^p$ into $L^u(\Q^\mu)$ follows.
\end{proof}
\let\qed\relax\end{proof}

\begin{proof}[\textbf{Proof of Proposition~\ref{PROP:dual_representation}.}]
    Under~\itemref{IT:g_lip}, the Lipschitz BSDE result recalled in Section~\ref{SEC:BSDEs} makes the induced risk measure $\rho$ well-defined on $\mathcal X_T=L^2(\F_T)$; under~\itemref{IT:g_quad}, the convex quadratic result makes it well-defined on $\mathcal X_T=L^{\mathrm{exp}}(\F_T)$.
    In addition, $\rho$ is normalized, cash-additive and convex, as recalled at the end of Section~\ref{SEC:BSDEs}.
    
    Before dealing with the dual representation, we note that the conditional expectation appearing in~\eqref{EQ:dual_representation_convex} is well-defined as a $[-\infty,+\infty)$-valued random variable.
    Since $g^\ast\geq 0$, for all $\mu\in\mathcal B$ we have
    \[
        A^\mu:=\int_s^T g^\ast(r,\mu_r)\d r
        \in [0,+\infty],
    \]
    and therefore $X-A^\mu\in[-\infty,+\infty)$, $\Q^\mu$-\text{a.s.}
    Consequently, 
    \[
        \E_\mu[X-A^\mu|\F_s]:=\E_\mu\big[(X-A^\mu)^+\big|\F_s\big] - \E_\mu\big[(X-A^\mu)^-\big|\F_s\big]
    \]
    is well-posed, provided $(X-A^\mu)^+\in L^1(\Q^\mu)$.
    This follows from $(X-A^\mu)^+ \leq X^+$, and from $X\in \mathcal X_T$, which implies $X^+\in L^1(\Q^\mu)$ by~\itemref{IT:LEM:cond_expec_convex:Lexp_BMO} and~\itemref{IT:LEM:cond_expec_convex:embedding} in Lemma~\ref{LEM:cond_expec_convex}.
    
    The dual representation is given in \cite[Theorem~7.4]{Barrieu+ElKaroui_2009_Pricing_hedging_optimally_designing_derivatives_minimization_risk_measures} under assumption~\ref{IT:g_lip} for all $X\in L^2(\F_T)$, and under assumption~\ref{IT:g_quad} for $X\in L^\infty(\F_T)$.
    It remains to extend the representation to $X\in L^{\mathrm{exp}}(\F_T)$ under assumption~\ref{IT:g_quad}.
    Fix $s\in[0,T]$ and let $(Y,Z)$ be the solution to the BSDE with parameters $(g,T,X)$, where $Y\in\mathcal S^{\mathrm{exp}}_T$ and $Z\in \mathcal H^p_T$ for all $p\geq 1$.
    Set
    \[
        D_s(X)
        :=
        \esssup_{\mu\in\BMO_T}
        \E_\mu\left[
            X-\int_s^T g^\ast(r,\mu_r)\d r
            \,\bigg|\,\F_s
        \right].
    \]
    
    We first prove that $D_s(X)\leq Y_s$.
    Fix $\mu\in\BMO_T$.
    The stochastic exponential of $\int_0^{\cdot}\mu\cdot\d W$ satisfies a reverse H\"older inequality by \cite[Theorem~3.4]{Kazamaki_1994_Continuous_exponential_martingales_BMO}. 
    Together with $Z\in\mathcal H_T^p$ for every $p\geq1$, this implies that $g(\,\cdot\,,Z)$ and $\mu\cdot Z$ are integrable with respect to $\Q^\mu\otimes\ell_1$, and that $\int_0^{\,\cdot \,}Z\cdot\d W^\mu$ is a uniformly integrable martingale under $\Q^\mu$, where $W^\mu:=W-\int_0^{\,\cdot\,}\mu \d \ell_1$ is a $\Q^\mu$-Brownian motion by the Girsanov theorem.
    Hence, the solution $Y$ can be written as
    \[
        Y_s = \E_\mu\left[ \left.
        X+\int_s^T\big(g(r,Z_r)-\mu_r\cdot Z_r\big)\d r
        \right|\F_s \right]
        \geq
        \E_\mu\left[ \left.
        X-\int_s^Tg^\ast(r,\mu_r)\d r
        \right|\F_s \right],
    \]
    where we used Fenchel's inequality.
    Taking the essential supremum over $\mu\in\BMO_T$ gives 
    \begin{equation}
    \label{EQ:PROP:dual_exp_left}
        D_s(X)\leq Y_s.
    \end{equation}

    For the reverse inequality, let $X_n:=(-n)\vee(X\wedge n)\in L^\infty(\F_T)$ and denote by $(Y^n,Z^n)\in\mathcal S^\infty_T\times\BMO_T$ the corresponding solution.
    By the dual representation for bounded terminal conditions in \cite[Theorem~7.4]{Barrieu+ElKaroui_2009_Pricing_hedging_optimally_designing_derivatives_minimization_risk_measures},
    \begin{equation}
    \label{EQ:PROP:dual_exp:essmax}
        Y^n_s
        =
        \essmax_{\mu\in\BMO_T}\E_{\mu}\left[
            X_n-\int_s^Tg^\ast(r,\mu_r)\d r
            \bigg|\F_s
        \right],
    \end{equation}
    and the maximum is attained at some $\mu^n\in\BMO_T$.
    Set
    \[
        A^n_s
        :=
        \E_{\mu^n}
        \!
        \bigg[
            \int_s^T
            \!
            g^\ast(r,\mu^n_r)\d r
            \bigg|\F_s
        \bigg], 
        \quad
        H^n_s
        :=
        \E_{\mu^n}
        \!
        \bigg[
            \log
                \frac{\mathcal E_T^{\mu^n}}
                     {\mathcal E_s^{\mu^n}}
            \bigg|\F_s
        \bigg]
        =
        \frac12
        \E_{\mu^n}
        \!
        \left[
            \int_s^T
            \!
            |\mu^n_r|^2\d r
            \bigg|\F_s
        \right].
    \]
    
    The estimate in Lemma~\ref{LEM:g*}\itemref{IT:LEM:g*:def}\itemref{IT:LEM:g*:prop_nonnega} yields
    \begin{equation}
    \label{EQ:PROP:dual_exp:H}
        H^n_s\leq2kA^n_s+2k^2T.
    \end{equation}
    From Lemma~\ref{LEM:g*}\itemref{IT:LEM:g*:projector}, let $q^0:=\Pi(\,\cdot\,,0)$, so that $q^0\in\BMO_T$ and $g^\ast(\,\cdot\,,q^0)=0$. 
    Then~\eqref{EQ:PROP:dual_exp:essmax} and $X_n\geq -X^-$ yield $Y^n_s\geq-\E_{q^0}[X^-\mid\F_s]$.
    Let us recall the conditional entropy inequality
    \[
        \E_{\mu^n}[\xi|\F_s]
        \le
        \log\E\big[e^{\xi}\big|\F_s\big]
        +H^n_s, \qquad\forall\, \xi\in L^\infty(\F_T),
    \]
    which directly follows from \cite[Lemma~2]{Detlefsen+Scandolo_2005_Conditional_dynamic_convex_risk_measures}.
    Using these two inequalities, we reach:
    \begin{gather}
        -\E_{q^0}[X^-\mid\F_s]
        \leq Y^n_s
        \le
        \E_{\mu^n}[X_n^+|\F_s]-A^n_s 
        \le
        \frac1{\lambda}
        \log\E\big[e^{\lambda X_n^+}\big|\F_s\big]
        +\frac{H^n_s}\l-A^n_s\\
        \le
        \frac1{\lambda}
        \log\E\big[e^{\lambda X^+}\big|\F_s\big]
        +
        \frac{2k^2T}{\lambda}
        -
        \Big(1-\frac{2k}{\lambda}\Big)A^n_s,
    \end{gather}
    where we used the maximizer $\mu^n$ in~\eqref{EQ:PROP:dual_exp:essmax} and $X_n\leq X_n^+$ for the second inequality, followed by the conditional entropy inequality applied to $\xi=\l X^+_n$ for some $\l>0$, then $X_n^+\leq X^+$ with~\eqref{EQ:PROP:dual_exp:H} at the end.
    Since $X\in L^{\mathrm{exp}}(\F_T)$, the first term in the last expression is finite $\P$-a.s. 
    Comparing the first and last members of the previous chain of inequalities, and choosing $\lambda>2k$, yields $\sup_{n\in\N}A^n_s<+\infty$, $\P$-a.s.
    Plugging this uniform upper bound into~\eqref{EQ:PROP:dual_exp:H} gives 
    \begin{equation}
    \label{EQ:PROP:dual_exp:sup_H}    
        \sup_{n\in\N}H^n_s<+\infty.
    \end{equation}
    
    Now set $U_n:=|X_n-X|$.
    Then,
    \begin{equation}
    \label{EQ:PROP:dual_exp:Y}
        Y^n_s
        \leq
        \E_{\mu^n}\left[
            X-\int_s^Tg^\ast(r,\mu^n_r)\d r
            \,\bigg|\,\F_s
        \right]
        +
        \E_{\mu^n}[U_n\mid\F_s] \leq
        D_s(X)+\E_{\mu^n}[U_n\mid\F_s].
    \end{equation}
    The stability result \cite[Proposition~7]{Briand+Hu_2008_Quadratic_BSDEs_convex_generators_unbounded_terminal_conditions} gives a subsequence $(n_i)_{i\in\N}$ such that  $Y^{n_i}_s\to Y_s$, $\P$-a.s.
    Moreover, fix $\lambda>0$ and $m\in\N$.
    Since $U_n\wedge m\in L^\infty(\F_T)$, we can apply the conditional entropy inequality to $\l(U_n\wedge m)$, then 
    let $m\to\infty$ by the conditional monotone convergence under $\Q^{\mu^n}$ and under $\P$, yielding:
    \begin{equation}
    \label{EQ:PROP:dual_exp:entr}
        \E_{\mu^n}[U_n\mid\F_s]
        \leq
        \frac{H^n_s}{\lambda}
        +
        \frac1\lambda
        \log\E\big[e^{\lambda U_n}\big|\F_s\big].
    \end{equation}
    Since $U_n\to0$ $\P$-a.s., $U_n\leq|X|$, and $X\in L^{\mathrm{exp}}(\F_T)$, the dominated convergence theorem yields ${\E[e^{\lambda U_n}\mid\F_s]\longrightarrow1}$ in $\P$-probability.
    Moreover, for every $\delta,\lambda>0$,~\eqref{EQ:PROP:dual_exp:entr} gives
    \[
        \P\big(\E_{\mu^n}[U_n|\F_s]>\delta\big)
        \leq\P\Big(\sup_{m\in\N}H^m_s>\frac{\lambda\delta}2\Big)
        +\P\Big(\log\E\big[e^{\lambda U_n}\big|\F_s\big]>\frac{\lambda\delta}2\Big).
    \]
    The second probability vanishes as $n\to\infty$, and the first one is arbitrarily small for $\lambda$ large by~\eqref{EQ:PROP:dual_exp:sup_H}. Hence, $\E_{\mu^n}[U_n\mid\F_s]\to0$ in $\P$-probability.
    By characterization of convergence in probability, we can extract a sub-subsequence $(n_{i_k})_{k\in\N}$, simply denoted by $(k)_{k\in\N}$,
    such that $\E_{\mu^{k}}[U_{k}\mid\F_s]\longrightarrow0$, $\P$-a.s. as $k\to\infty$.
    Therefore, passing to the limit in the first and last members of~\eqref{EQ:PROP:dual_exp:Y},  $\P$-a.s. along the sub-subsequence $(n_{i_k})_{k\in\N}$, yields $Y_s\leq D_s(X)$.
    This inequality, when combined with~\eqref{EQ:PROP:dual_exp_left}, proves the claimed dual representation~\eqref{EQ:dual_representation_convex} for $X\in L^{\mathrm{exp}}(\F_T)$ and completes the proof.
\end{proof}

\begin{proof}[\textbf{Proof of Lemma~\ref{LEM:mu_vs_nu}.}]
    Fix $\mu,\nu\in\BMO_T$. 
    By the Girsanov theorem, $W^\mu:= W - \int_0^{\, \cdot}\mu_r\d r$ is a $\Q^\mu$-Brownian motion (see, e.g., \cite[Section~3.5, Theorem~5.1]{Karatzas+Shreve_1991_Brownian_motion_stochastic_calculus}).
    Let us compute the Radon--Nikodym derivative of $\Q^\nu$ with respect to $\Q^\mu$, rewrite it in terms of $W^\mu$:
    \begin{align}
        R_T:=
        \!
        \frac{\d \Q^\nu}{\d \Q^\mu}
        \!=\!
        \frac{\mathcal E^\nu_T}{\mathcal E^\mu_T}
        =\exp\bigg(
        \!
        \int_0^T(\nu_r-\mu_r)\cdot \d W^\mu_r 
        - \frac12\int_0^T
        \!\!\!
        |\nu_r-\mu_r|^2\d r\bigg) 
        = \mathcal E_T\big((\nu-\mu)\tinybullet W^\mu\big),
    \label{EQ:LEM:mu=nu_RT}
    \end{align}
    \normalsize
    where we used  the identity $|x-y|^2=|x|^2 + |y|^2 - 2x\cdot y$, for $x,y\in\R^m$, to rewrite the Lebesgue integral.
    Since $\mu,\nu\in\BMO_T$, $M:=(\mu\tinybullet W)$ and $N:=\big((\mu-\nu)\tinybullet W\big)$ are BMO-martingales. 
    By \cite[Theorem~3.6]{Kazamaki_1994_Continuous_exponential_martingales_BMO}, $\langle N,M\rangle-N=\big((\nu-\mu)\tinybullet W^\mu\big)$ is a BMO-martingale under $\Q^\mu=\mathcal E_T(M)\P$. 
    Hence, $\mathcal E(\langle N,M\rangle-N)=R$ is a $\Q^\mu$-martingale by \cite[Theorem~2.3]{Kazamaki_1994_Continuous_exponential_martingales_BMO}.
    In particular,  $\E_\mu [R_T|\F_s]=R_s>0$ for any $s\in[0,T]$. 

\textit{\textbf{\itemref{IT:LEM:mu_vs_nu:bound_u}.}}
    Assume that $\mu,\nu\in\BMO_T$ satisfy $|\mu|,|\nu|\leq K$, $\P\otimes\ell_1$-a.e. for some $K>0$, and fix $u>1$, $X\in L^u(\Q^\mu)$, $s\in[0,T]$.
    We evaluate the difference between the conditional expectations of $X$ with respect to $\nu$ and $\mu$.
    Using the conditional change-of-measure formula $\E_\nu[X|\F_s]=\E_\mu[R_TR_s^{-1}X|\F_s]$, and the conditional H\"older inequality, we have
    \begin{align}
        \big|\E_\nu[X|\F_s]-\E_\mu[X|\F_s]\big|^u
        &=\left|\E_\mu\big[(R_TR_s^{-1}-1)X\big|\F_s\big]\right|^u\\
        &\leq \E_\mu[|X|^u|\F_s]\Big(\E_\mu\big[|R_TR_s^{-1}-1|^v\big|\F_s\big]\Big)^{u/v}.
    \label{EQ:LEM:mu=nu:bound_u}
    \end{align}
    Letting $\theta:=\nu-\mu$ and $\Lambda:=\mathcal E\big(\1_{[s,T]}\theta\tinybullet W^\mu\big)$, note that we have $R_TR_s^{-1}=\Lambda_T$, hence the Itô formula for the stochastic exponential yields
    \[
        R_TR_s^{-1}=1+\int_s^T\Lambda_r\theta_r\cdot \d W^\mu_r.
    \]
    By the conditional Burkholder--Davis--Gundy inequality (see, e.g., \cite[Chapter~VII, numbers~91,92]{Dellacherie+Meyer_1982_Probabilities_potential_B}; see also \cite{Lenglart+Lepingle+Pratelli_1980_Presentation_unifiee_certaines_inegalites_theorie_martingales}), there is a constant $c_v>0$ such that
    \begin{align}
        \E_\mu\big[|R_TR_s^{-1}-1|^v\big|\F_s\big] 
        &\leq c_v\E_\mu\bigg[\bigg(\int_s^T\Lambda_r^2|\theta_r|^2 \d r\bigg)^{\!v/2}\bigg|\F_s\bigg]\\
        &\leq c_v\sqrt{\E_\mu\bigg[\bigg(\sup_{r\in[s,T]}\Lambda_r^{2v}\bigg)\bigg|\F_s\bigg]}
        \sqrt{\E_\mu\bigg[\bigg(\int_s^T|\theta_r|^2 \d r\bigg)^{\!v}\bigg|\F_s\bigg]},
        \qquad
    \label{EQ:LEM:mu=nu:bound_u_bis}
    \end{align}
    where we used the conditional Cauchy--Schwarz inequality in the last line.
    Since $(\Lambda_r^{2v})_{r\in[s,T]}$ is a non-negative $\Q^\mu$-submartingale, the conditional Doob inequality yields the existence of a constant $C_v>0$ such that
    \begin{equation}
    \label{EQ:LEM:mu=nu:bound_u_ter}
        \E_\mu\left[\left.\sup_{r\in[s,T]}\Lambda_r^{2v}\right|\F_s\right]
        \leq C_v\E_\mu\left[\left.\Lambda_T^{2v}\right|\F_s\right]
        \leq C_v\exp\left(4v(2v-1)K^2T\right),
    \end{equation}
    where we used Lemma~\ref{LEM:cond_expec_convex}\itemref{IT:LEM:cond_expec_convex:stoch_exponen} for the second inequality, whose proof applies verbatim under $\Q^\mu$, with $W^\mu$ in place of $W$ and $|\theta|\leq 2K$.
    Reinserting~\eqref{EQ:LEM:mu=nu:bound_u_ter} into~\eqref{EQ:LEM:mu=nu:bound_u_bis} and then again into~\eqref{EQ:LEM:mu=nu:bound_u} yields the claim with ${C_{u,K,T}:=\big(c_v^2C_ve^{4v(2v-1)K^2T}\big)^{u/(2v)}>0}$.

\textit{\textbf{\itemref{IT:LEM:mu_vs_nu:bound_infty}.}}
    Fix $s\in[0,T]$ and $X\in L^\infty$, and assume that $\mu,\nu\in\BMO_T$ satisfy $\nu=\mu=0$ on $\Om\times[0,s]$.
    Then $R_s=1$ and, by the conditional change-of-measure formula, $\E_\nu[X|\F_s]=\E_\mu[R_TR_s^{-1}X|\F_s]=\E_\mu[R_TX|\F_s]$.
    Moreover, since $\Q^\nu$ and $\P$ are equivalent, $X\in L^\infty(\Q^\nu)$ with $\|X\|_{L^\infty}=\|X\|_{L^\infty(\Q^\nu)}$.
    Therefore,
    \begin{equation}
    \label{EQ:LEM:mu=nu_quadratic:bound_u}
        \big|\E_\mu[X|\F_s]-\E_\nu[X|\F_s]\big|
        =\left|\E_\mu\big[(1-R_T)X\big|\F_s\big]\right|
        \leq \|X\|_{L^\infty}\E_\mu[|R_T-1||\F_s].
    \end{equation}
    Notice that $R_T>0$, and estimate the last conditional expectation via the Cauchy--Schwarz inequality
    \begin{align}
        \E_{\mu}\big[|R_T-1|\big|\F_s\big]
        &=\E_{\mu}\Big[\big|\sqrt{R_T}-1\big|\,\big|\sqrt{R_T}+1\big|\Big|\F_s\Big]\\
        &\leq \left(\E_\mu\left[\left.\big(\sqrt{R_T}-1\big)^2\right|\F_s\right]\right)^{1/2}\left(\E_\mu\left[\left.\big(\sqrt{R_T}+1\big)^2\right|\F_s\right]\right)^{1/2}.
    \label{EQ:LEM:bound:begin}
    \end{align}
    Concerning the second factor, we compute the square, then use $\E_\mu[R_T|\F_s]=R_s=1$ and Jensen's inequality:
    \begin{equation}
    \label{EQ:LEM:bound_factor_2}
        \E_\mu\Big[\big(\sqrt{R_T}+1\big)^2\Big|\F_s\Big]
        =\E_\mu\Big[R_T+2\sqrt{R_T}+1\Big|\F_s\Big]
        \leq 2 +2\sqrt{\E_\mu\big[R_T\big|\F_s\big]}
        =4.
    \end{equation}
    Concerning the first factor, we use the elementary inequality $(\sqrt x - 1)^2\leq -\ln x + x-1$, for $x>0$, and recall that $\E_\mu[R_T|\F_s]=R_s=1$:
    \begin{equation}
    \label{EQ:LEM:bound_factor_1}
        \E_\mu\Big[\big(\sqrt{R_T}-1\big)^2\Big|\F_s\Big]
        \leq \E_\mu\big[-\ln R_T + R_T - 1\big|\F_s\big]
        =-\E_\mu\big[\ln R_T\big|\F_s\big].
    \end{equation}
    Recall the definition of $R_T$ and use the martingale property for $\big((\mu-\nu)\tinybullet W^\mu\big)$, under $\Q^\mu$, to compute its logarithm:
    \begin{equation}
    \label{EQ:LEM:mu=nu:ln}
        -\E_\mu\big[\ln R_T\big|\F_s\big]
        =\frac12\E_\mu\bigg[\int_0^T|\mu_r-\nu_r|^2\d r \bigg|\F_s\bigg].
    \end{equation}
    Recall that $\nu=\mu$ on $\Om\times[0,s]$, then reinsert~\eqref{EQ:LEM:bound_factor_1} and~\eqref{EQ:LEM:bound_factor_2} into~\eqref{EQ:LEM:bound:begin} and finally~\eqref{EQ:LEM:bound:begin} into~\eqref{EQ:LEM:mu=nu_quadratic:bound_u} to reach the conclusion.

\textit{\textbf{\itemref{IT:LEM:mu_vs_nu:equality}.}}
    Suppose now that $\nu,\mu\in\BMO_T$ satisfy $\nu=\mu$ on $\Om\times(s,a)$, for some $0\leq s< a\leq T$.
    Then $R_s=R_a$.
    As a first step, suppose that $X$ is $\F_a$-measurable, real-valued, and integrable with respect to both $\Q^\nu$ and $\Q^\mu$, then
    \[
        \E_\nu[X|\F_s]
        =\E_\mu[R_T X|\F_s]R_s^{-1} 
        = \E_\mu\big[\E_\mu[R_T X|\F_a]\big|\F_s\big]R_s^{-1}.
    \]
    Since $X$ is $\F_a$-measurable and $R$ is a $\Q^\mu$-martingale, we have \[\E_\mu[R_T X|\F_a]=X\E_\mu[R_T|\F_a]=XR_a.\] 
    Therefore, we reach the claim for an integrable random variable
    \[
        \E_\nu[X|\F_s]
        =\E_\mu[XR_a|\F_s]R_s^{-1}
        =\E_\mu\big[XR_aR_s^{-1}\big|\F_s\big]
        =\E_\mu[X|\F_s].
    \]

    Assume now that $X$ is $\F_a$-measurable and $[0,+\infty]$-valued. 
    For all $n\in\N$, define the bounded (hence integrable) random variable $X_n:=X\wedge n$. 
    By the previous step we have $\E_\nu[X_n|\F_s]=\E_\mu[X_n|\F_s]$.
    By the conditional monotone convergence theorem, we reach 
    \[
        \E_\mu[X|\F_s]
        :=\lim_{n\to\infty}\E_\mu[X_n|\F_s]
        =\lim_{n\to\infty}\E_\nu[X_n|\F_s]
        =\E_\nu[X|\F_s],
    \]
    the equality and the limits holding $\P$-a.s.\ as $[0,+\infty]$-valued random variables.

    Finally, assume that $X$ is $\F_a$-measurable and $[-\infty,+\infty)$-valued such that $X^+$ is integrable with respect to both $\Q^\mu$ and $\Q^\nu$.
    Then the first step gives $\E_\nu[X^+|\F_s]=\E_\mu[X^+|\F_s]$, while the second gives $\E_\nu[X^-|\F_s]=\E_\mu[X^-|\F_s]$, as $[0,+\infty]$-valued random variables.
    Subtracting the conditional expectations of the positive and negative parts, we reach
    \[
        \E_\mu[X|\F_s]:=\E_\mu[X^+|\F_s] - \E_\mu[X^-|\F_s] = \E_\nu[X^+|\F_s] - \E_\nu[X^-|\F_s] =:\E_\nu[X|\F_s],
    \]
    hence the equality between the $[-\infty,+\infty)$-valued conditional expectations.
\end{proof}

\begin{proof}[\textbf{Proof of Lemma~\ref{LEM:h}.}]
    Retain the predictable null set $A$, the modified driver $\widetilde g$, and the multifunction $K(\omega,r)=\partial_z\widetilde g(\omega,r,0)$ from the proof of Lemma~\ref{LEM:g*}.
    In particular, $K$ is non-empty, compact, convex and $\mathcal P$-measurable in the sense of \cite[Definition~14.1]{Rockafellar+Wets_1998_Variational_analysis}; it equals $\partial_zg(\omega,r,0)$ on $A^c$ and $\{0\}$ on $A$.
    Now define the support function of $K$ by
    \[
        h(\omega,r,z)
        :=
        \sup_{q\in K(\omega,r)}q\cdot z
        =
        \begin{cases}
            \ds\sup_{q\in\partial_z g(\omega,r,0)}q\cdot z,
            &\text{if }(\omega,r,z)\in A^c\times\R^m,\\
            0,
            &\text{if }(\omega,r,z)\in A\times\R^m,
        \end{cases}
    \]
    which is a convex normal integrand by \cite[Example~14.51]{Rockafellar+Wets_1998_Variational_analysis}. 
    In particular, $h$ is $\mathcal P\otimes\mathscr B(\R^m)$-measurable by \cite[Corollary~14.34]{Rockafellar+Wets_1998_Variational_analysis}. 
    
    We now prove the statements~\itemref{IT:LEM:h:def},~\itemref{IT:LEM:h:lip} on $A^c$.
    Fix henceforth $(\omega,r)\in A^c$. 

    The statement~\itemref{IT:LEM:h:def} is satisfied by construction.
    
    Since $z\mapsto h(\omega,r,z)$ is a support function, it is sublinear by \cite[Theorem~8.24]{Rockafellar+Wets_1998_Variational_analysis}.
    For $q\in\partial_z g(\om,r,0)$ and $z,z'\in\R^m$ we have
    \[
        q\cdot z
        \leq q\cdot z'+q\cdot(z-z') 
        \leq q\cdot z'+R|z-z'|,
    \]
    where $R:=L$ if $g$ satisfies~\itemref{IT:g_lip} and $R:=2k$ if $g$ satisfies~\itemref{IT:g_quad}, thanks to Lemma~\ref{LEM:g*}\itemref{IT:LEM:g*:def}\itemref{IT:LEM:g*:subgrad_0}.
    Take the supremum over $q$ to have $h(\om,r,z)-h(\om,r,z')\leq R|z-z'|$.
    Exchanging $z$ and $z'$ yields $|h(\om,r,z)-h(\om,r,z')|\leq R|z-z'|$, that is the Lipschitz continuity with constant $R$ for $z\mapsto h(\om,r,z)$.
    This shows~\itemref{IT:LEM:h:lip}.
    
    It remains to prove~\itemref{IT:LEM:h:selector}. 
    Consider the product space $ E:=\Omega\times[0,T]\times\R^m$ endowed with the $\sigma$-algebra $\mathcal P\otimes\mathscr B(\R^m)$, and define the multifunction $\overline K: E\ni(\omega,r,z)\mapsto K(\omega,r)\subseteq\R^m$.
    The multifunction $\overline K$ is closed-valued by Lemma~\ref{LEM:g*}\itemref{IT:LEM:g*:def}\itemref{IT:LEM:g*:subgrad_0} and is $\mathcal P\otimes\mathscr B(\R^m)$-measurable because $\overline K^{-1}(O)=K^{-1}(O)\times\R^m \in \mathcal P\otimes\mathscr B(\R^m)$ for every open $O\subseteq\R^m$.
    Then the function
    \[
        F: E\times\R^m\ni((\om,r,z),q)\longmapsto 
        \begin{cases}
            -q\cdot z, &\text{if }q\in \overline K(\om,r,z),\\
            +\infty, &\text{if }q\notin \overline K(\om,r,z),\\
        \end{cases}
    \]
    is a normal integrand by \cite[Example~14.32]{Rockafellar+Wets_1998_Variational_analysis}.
    By \cite[Theorem~14.37]{Rockafellar+Wets_1998_Variational_analysis}, the argmin multifunction
    \[
        \Gamma: E\ni(\omega,r,z)\longmapsto \argmin_{q\in\R^m}F((\omega,r,z),q)
        =\argmax_{q\in K(\om,r)}q\cdot z
    \]
    is closed-valued and $\mathcal P\otimes\mathscr B(\R^m)$-measurable.
    In addition $\Gamma(\om,r,z)$ is non-empty for any $(\om,r,z)\in E$ because $q\mapsto q\cdot z$ is continuous and $K(\om,r)$ is compact by Lemma~\ref{LEM:g*}\itemref{IT:LEM:g*:def}\itemref{IT:LEM:g*:subgrad_0}.
    Therefore, the measurable selection theorem \cite[Corollary~14.6]{Rockafellar+Wets_1998_Variational_analysis} yields a $\mathcal P\otimes\mathscr B(\R^m)$-measurable map $\hat\mu:\Omega\times[0,T]\times\R^m\to\R^m$ such that $\hat\mu\in\Gamma$ everywhere in $E$, namely
    \[
        \hat\mu(\omega,r,z)\in K(\omega,r)
        \quad\text{and}\quad
        \hat\mu(\omega,r,z)\cdot z
        =
        \max_{q\in K(\omega,r)} q\cdot z
        =
        h(\omega,r,z), \qquad \forall\, (\om,r,z)\in E.
    \]
    Then~\itemref{IT:LEM:h:selector} holds on $A^c$.    
\end{proof}

\section{Deterministic drivers}
\label{SEC:app_determ_driver}

This appendix verifies that condition~\itemref{IT:GA}
is automatically satisfied when the driver is
deterministic and constant in time.

\begin{lemma}
\label{LEM:g_deter}
    Assume that $g$ is deterministic and constant in time, so that $g:\R^m\to\R$, and satisfies either~\itemref{IT:g_lip} or \itemref{IT:g_quad}.
    Then the condition~\itemref{IT:GA} is satisfied for every $0\leq s\leq t<T$.
\end{lemma}
\begin{proof}
    The condition~\itemref{IT:GA} is trivially satisfied if $q\in\R^m$ is such that $g^\ast(q)=+\infty$, thus we will restrict the domain to $\operatorname{dom}(g^\ast):=\{q\in\R^m \ : \ g^\ast(q)<+\infty\}$.
    In view of Remark~\ref{REM:GA}\itemref{IT:REM:GA:minimal}, it is enough to construct a concave function $\phi_0:[0,+\infty)\to[0,+\infty)$ with $\lim_{y\to 0^+}\phi_0(y)=\phi_0(0)=0$, such that
    \begin{equation}
    \label{EQ:LEM:g_deter:phi}
        \operatorname{dist}^2(q,\partial g(0))
        \leq
        \phi_0(g^\ast(q)),
        \qquad q\in\operatorname{dom}(g^\ast).
    \end{equation}
    $\phi_0$ is then non-decreasing.
    Indeed, suppose by contradiction that $\phi_0(x)>\phi_0(y)$ for some $0\le x<y$ and set $\delta:=\phi_0(x)-\phi_0(y)>0$. 
    Writing $y=\tfrac{n}{n+1}x+\tfrac{1}{n+1}x_n$ with $x_n:=y+n(y-x)$ and using concavity gives $\phi_0(x_n)\le\phi_0(y)-n\delta<0$ for $n$ large, contradicting $\phi_0\ge 0$.
    Setting $\phi(y):=\phi_0(y)+y$ for all $y\geq 0$, the function $\psi:=\phi^{-1}$ satisfies condition~\itemref{IT:GA} for all $0\leq s \leq t <T$.
    
    For $q\in\operatorname{dom}(g^\ast)$ and $x\geq0$, let
    \[
        d(q):=\operatorname{dist}(q,\partial g(0)),
        \qquad 
        \Gamma(x):=\{q\in\R^m:g^\ast(q)\leq x\},
        \qquad
        \Phi(x):=
        \max_{q\in\Gamma(x)}
        d^2(q).
    \]
    The definition is well posed since the sublevel sets of $g^\ast$ are non-empty and compact, by Lemma~\ref{LEM:g*}\itemref{IT:LEM:g*:def}\itemref{IT:LEM:g*:prop_nonnega},\itemref{IT:LEM:g*:subgrad_0}, and $d$ is continuous.
    
    Moreover,  $\Phi(x)\longrightarrow0$ as $x\to0^+$.
    Indeed, if not, there exist $\eta>0$, a sequence $x_n\to0^+$, and $q_n\in\Gamma(x_n)$ such that $d^2(q_n)\geq\eta$.
    Since $(x_n)_{n\in\N}$ is bounded, the sequence $(q_n)_{n\in\N}$ is contained in a fixed compact sublevel set of $g^\ast$.
    Hence, along a subsequence, $q_{n_k}\to q$ for some $q\in\R^m$.
    By lower semi-continuity, $g^\ast(q)
        \leq
        \liminf_{k\to\infty}g^\ast(q_{n_k})
        \leq
        \liminf_{k\to\infty}x_{n_k}
        =0$.
    Since $g^\ast\geq0$, we obtain $g^\ast(q)=0$, and therefore $q\in\partial g(0)$.
    Thus $d(q_{n_k}) \longrightarrow d(q)=0$, which is a contradiction.
    
    Furthermore, the estimates in
    Lemma~\ref{LEM:g*}\itemref{IT:LEM:g*:def}\itemref{IT:LEM:g*:prop_nonnega},\itemref{IT:LEM:g*:subgrad_0}
    yield constants $A,B\geq0$ such that $\Phi(x)\leq A+Bx$ for $x\geq 0$.
    Indeed, under~\itemref{IT:g_lip}, both $\Gamma(x)$ and $\partial g(0)$ are bounded by $L$, so that $\Phi(x)\leq4L^2$.
    Under~\itemref{IT:g_quad}, $|q|^2\leq4k(x+k)$ when $q\in\Gamma(x)$, while $\partial g(0)$ is bounded by $2k$, and hence $d^2(q) \leq2|q|^2+8k^2 \leq8kx+16k^2$.
    Thus one may take $(A,B)=(4L^2,0)$ and $(A,B)=(16k^2,8k)$, respectively.

    By \cite[Proposition~2.31]{Rockafellar+Wets_1998_Variational_analysis}, the convexification of $-\Phi$ is the greatest convex function $f:[0,+\infty)\to[-\infty,+\infty]$ such that $f\leq -\Phi$.
    Set $\phi_0:=-f$.
    Then $\phi_0$ is concave and $\Phi\leq\phi_0$.
    
    Fix now $\eps>0$.
    Choose $\delta>0$ such that $\Phi(x)\leq\eps$ for $0\leq x\leq\delta$, and set $M_\eps:=B+A/\delta$.
    Then $-\eps-M_\eps x\leq-\Phi(x)$ for $x\geq 0$.
    By the maximality of $f$, $-\eps-M_\eps x
        \leq
        f(x)
        \leq
        -\Phi(x)
        \leq0$.
    Thus $0\leq\phi_0(x)\leq\eps+M_\eps x$ which implies $\phi_0(0)=0$ and continuity at zero.

    Finally, for $q\in\operatorname{dom}(g^\ast)$, we have $q\in\Gamma(g^\ast(q))$, hence $d^2(q)
        \leq
        \Phi(g^\ast(q))
        \leq
        \phi_0(g^\ast(q))$.
    Thus~\eqref{EQ:LEM:g_deter:phi} holds. 
\end{proof}

\section{Proof of Theorem~\ref{TH:resilience_convex} in the Lipschitz case}
\label{SEC:app_lipschitz}
We here explain in detail how the proof of Theorem~\ref{TH:resilience_convex} under~\itemref{IT:TH:convex:quad} can be adapted to the case~\itemref{IT:TH:convex:lip}, which allows for a more general choice of the inner risk process $\pi$ by restricting the class of convex risk measures $\r$ (imposing stricter conditions on the drivers $g$).

In this setting, dealing with the dual representation for $\r$ is easier thanks to the uniform boundedness of the dual drifts $\mu\in\mathcal B$ at all times.
Thus the main strategy of the proof that we presented in the quadratic case still holds in the Lipschitz case.
The only parts that need further analysis are the ones that were based on the uniform boundedness of $b$ and $\s$.
We review, one by one, the steps of the previous proof, discussing how they need to be modified under the new assumptions.

\begin{proof}[\textbf{Proof of Theorem~\ref{TH:resilience_convex} under~\itemref{IT:TH:convex:lip}.}]
Henceforth assume that $g$ satisfies~\itemref{IT:g_lip}, and $b\in L^2_T$.
Then the dual domain $\mathcal B\subseteq\BMO_T$ for the induced risk measure $\r$ consists of predictable processes $\mu:\Om\times[0,T]\to\R^m$ essentially bounded by $L$, as stated in Proposition~\ref{PROP:dual_representation}.

\begin{enumerate}[label=\textit{Step $\arabic*\!$}, 
leftmargin=*, 
itemsep=1pt, 
topsep=1pt]
    \item is unchanged, with the assumptions guaranteeing that $\Delta_\eps\pi_t\in L^2$.

    \item also remains unchanged.
    In order to apply Lemma~\ref{LEM:mu_vs_nu}\itemref{IT:LEM:mu_vs_nu:equality}, notice that
    \[
        \left(
        \int_t^{t+\eps}(b_r+\s_r\cdot\tilde\mu_r)\d r-\int_s^{t+\eps}g^\ast(r,\tilde\mu_r)\d r
        \right)^+
        \le \int_t^{t+\eps}\big(|b_r|+L|\s_r|\big)\d r.
    \]
    The right-hand side belongs to $L^2$, hence to $L^1(\Q^\mu)\cap L^1(\Q^{\tilde\mu})$ by Lemma~\ref{LEM:cond_expec_convex}\itemref{IT:LEM:cond_expec_convex:embedding} (with $K=L$, $u=1$, $p=2$).
    Therefore the conclusion~\eqref{EQ:TH:convex:step2_bis_quadratic} of the second step remains valid.

    \item  maintains the same structure with the same set $\mathcal C_\eps$, but the final estimate~\eqref{EQ:TH:convex_quadratic:step_J_2_bis} needs changing because $b$ and $\s$ are no longer bounded processes.
    In detail, from~\eqref{EQ:TH:convex_quadratic:step_J_1} and from $b+\sigma\cdot\mu-\frac12 g^\ast(\,\cdot\,,\mu)
        \leq b+\sigma\cdot\mu
        \le |b|+L|\sigma|$, we infer
    \[
        V_\eps(\mu)
        \le
        \E_\mu\left[\left.
        \frac1\eps\int_t^{t+\eps}\big(|b_r|+L|\sigma_r|\big)\d r
        \right|\F_s\right]
        -\frac1{2\eps}G_{t,t+\eps}(\mu)
        -\frac1\eps G_{s,t}(\mu).
    \]
    Hence, the deterministic constant $C_0$ appearing in control~\eqref{EQ:TH:convex_quadratic:step_J_2} for $V_\eps(\mu)$ can be replaced by the random variable
    \[
        U_\eps
        :=
        \esssup_{\mu\in\mathcal B_\eps}
        \E_\mu\left[\left.
        \frac1\eps\int_t^{t+\eps}\big(|b_r|+L|\sigma_r|\big)\d r
        \right|\F_s\right].
    \]
    Repeating the $\F_s$-pasting argument from \textit{Step}~3 of the main proof, with the $\F_s$-measurable bound $U_\eps$ in place of $C_0$, gives the analogue of~\eqref{EQ:TH:convex_quadratic:step_J_2_bis}:
    \begin{equation}
    \label{EQ:REM:lip:step_3}
        \esssup_{\mu\in\mathcal B_\eps}V_\eps(\mu)
        \leq
        \bigg(\esssup_{\mu\in\mathcal C_\eps}V_\eps(\mu)\bigg)
        \vee\left(U_\eps-\frac1{2\sqrt\eps}\right).
    \end{equation}

    \item maintains the projection argument but the upper bound needs to be amended.
    The main difference with the quadratic case is that $\langle Y_t\rangle_\eps$ is only in $L^2$ and not in $L^\infty$, hence Lemma~\ref{LEM:mu_vs_nu}\itemref{IT:LEM:mu_vs_nu:bound_infty} has to be replaced by Lemma~\ref{LEM:mu_vs_nu}\itemref{IT:LEM:mu_vs_nu:bound_u} for $K=L$ and some fixed $u\in(1,2)$, with the conjugate exponent $v=u/(u-1)>2$.
    Specifically, we obtain
    \small
    \begin{equation}
    \label{EQ:REM:lip_variation_u}
        \big|\E_\mu[\langle Y_t\rangle_\eps|\F_s]
        -
        \E_{\tilde\mu}[\langle Y_t\rangle_\eps|\F_s]\big|
        \!
        \leq
        \!
        C'\big(\E_\mu\big[|\langle Y_t\rangle_\eps|^u\big|\F_s\big]\big)^{1/u}
        \Bigg(
        \!
        \E_\mu
        \!
        \Bigg[
        \!
        \bigg(\int_s^{t+\eps}
        \!\!\!\!\!\!\!\!\!
        |\mu_r-\tilde\mu_r|^2\d r\bigg)^{
        \!
        v}\Bigg|\F_s\Bigg]\Bigg)^{
        \!\!
        1/(2v)},
\end{equation}
    \normalsize
    where $C'>0$ depends only on $u,L,T$.
    To control the right-hand side, we write the exponent of the time integral as $v=1+(v-1)$.
    For the integral over $(s,t)$ we still employ the condition \itemref{IT:GA} in form~\eqref{EQ:IT:GA_concave}, while over $(t,t+\eps)$ and for the remaining factor $(\cdots)^{v-1}$  we use instead the bound $|\mu-\tilde\mu|\leq 2L$.
    This yields
    \small
    \begin{align}
        \E_\mu\bigg[\!\bigg(\!\int_s^{t+\eps}
        \!\!\!\!
        |\mu_r-\tilde\mu_r|^2\d r\bigg)^v\bigg|\F_s\bigg]
        \!\leq \!
        \bigg[(1\vee T)\phi\big(G_{s,t}(\mu)\big)+4L^2\eps\bigg]\big(4L^2 T\big)^{\!v-1}
        \!\!
        \leq C''\big(\phi(\sqrt\eps)+\eps\big),
    \end{align}
    \normalsize
    where in the last inequality we used the  increasing monotonicity of $\phi$ with $\mu\in\mathcal C_\eps$, and set $C'':=\big[1\vee T\vee (4L^2)\big]\big(4L^2 T\big)^{v-1}>0$.
    Next, we bound $\big(\E_\mu\big[|\langle Y_t\rangle_\eps|^u\big|\F_s\big]\big)^{1/u}$ by $C'''\sqrt{\E\big[|\langle Y_t\rangle_\eps|^2\big|\F_s\big]}$, thanks to  Lemma~\ref{LEM:cond_expec_convex}\itemref{IT:LEM:cond_expec_convex:embedding} with $p=2$ and $K=L$, where $C'''>0$ depends only on $u,L,T$.
    Combining these estimates with~\eqref{EQ:TH:convex_quadratic:upper_2} and recalling that $\tilde\mu\in\A_s$, we find a constant $\tilde C_1>0$ depending only on $u,L,T$, such that
    \begin{equation}
    \label{EQ:REM:lip:step_4}
        \esssup_{\mu\in\mathcal C_\eps}V_\eps(\mu)
        \le
        \esssup_{\mu\in\A_s}\E_\mu[\langle Y_t\rangle_\eps|\F_s]
        +
        \tilde C_1
        \sqrt{\E\big[|\langle Y_t\rangle_\eps|^2\big|\F_s\big]}
        \tilde \varpi(\eps),
    \end{equation}
    where $\tilde\varpi(\eps):=(\phi(\sqrt\eps)+\eps)^{1/(2v)}$ satisfies $\tilde\varpi(\eps)\to 0$ as $\eps\to 0^+$.
    This is the upper bound that replaces~\eqref{EQ:TH:convex_quadratic:upper_3} in the Lipschitz case.

    \item is adapted by using  Lemma~\ref{LEM:mu_vs_nu}\itemref{IT:LEM:mu_vs_nu:bound_u} with $K=L$ and the same $u\in(1,2)$ from before.
    We then easily reach the lower bound replacing~\eqref{EQ:TH:convex_quadratic:step_lowerbound_2_new}:
    \begin{equation}
    \label{EQ:REM:lip:step_5}
        \esssup_{\mu\in\mathcal B_\eps}V_\eps(\mu)
        \ge
        \esssup_{\mu\in\A_s}\E_\mu[\langle Y_t\rangle_\eps|\F_s]
        -
        \tilde C_2
        \sqrt{\E\big[|\langle Y_t\rangle_\eps|^2\big|\F_s\big]}
        \sqrt\eps,
    \end{equation}
    where $\tilde C_2>0$ depends only on $u,L,T$.

    \item is unchanged, yielding the analogous full $\ell_1$-measure Borel set $\mathcal T_1\subseteq[0,T)$, such that, for $t\in\mathcal T_1$, we have $\langle Y_t\rangle_\eps\longrightarrow Y_t$ in $L^2$, hence, for $s\in[0,t]$,
    \[
        \esssup_{\mu\in\A_s}\E_\mu[\langle Y_t\rangle_\eps|\F_s]
        \longrightarrow
        J_{s,t}, \qquad \text{in }L^1.
    \]
    Moreover, the common factor appearing in both the upper bound~\eqref{EQ:REM:lip:step_4} from \textit{Step} 4 and in the lower bound~\eqref{EQ:REM:lip:step_5} from \textit{Step} 5 is bounded in $L^1$, for sufficiently small $\eps$, by 
    \[
        \left\|\sqrt{\E\big[|\langle Y_t\rangle_\eps|^2\big|\F_s\big]}\right\|_{L^1}
        \leq \left\|\langle Y_t\rangle_\eps\right\|_{L^2}
        \leq \left\|Y_t\right\|_{L^2}+1.
    \]
    Therefore, the right-hand sides of~\eqref{EQ:REM:lip:step_4} and~\eqref{EQ:REM:lip:step_5} converge in $L^1$ to $J_{s,t}$.

    \item keeps the same strategy.
    The negative part can be dealt with exactly as in the quadratic case, using the lower bound~\eqref{EQ:REM:lip:step_5} and the new \textit{Step}~6:
    \begin{equation}
    \label{EQ:REM:negative_part}
        \left(
        \esssup_{\mu\in\mathcal B_\eps}V_\eps(\mu)-J_{s,t}
        \right)^-
        \longrightarrow0,
        \qquad \text{in }L^1.
    \end{equation}
    Analogously, the positive part involving the essential supremum over $\mathcal C_\eps$, is treated as in the quadratic case, using the upper bound~\eqref{EQ:REM:lip:step_4} and the new \textit{Step} 6:
    \begin{equation}
    \label{EQ:REM:lip:C}
        \left(
            \esssup_{\mu\in\mathcal C_\eps}V_\eps(\mu)-J_{s,t}
        \right)^+
        \longrightarrow0,
        \qquad \text{in }L^1.
    \end{equation}
    It remains to show that
    \begin{equation}
    \label{EQ:REM:lip:B_minus_C}
        \left(U_\eps-\frac1{2\sqrt\eps}-J_{s,t}\right)^+
        \longrightarrow 0,
        \qquad\text{in }L^1.
    \end{equation}
    In the quadratic case, the analogous positive part vanishes for sufficiently small $\eps$ by~\eqref{EQ:TH:convex_quadratic:L1_1_b}.
    In the Lipschitz case, this is no longer true as $U_\eps$ is random and $J_{s,t}$ is not essentially bounded.
    We look instead for an $L^2$-lower bound, namely:
    \begin{equation}
    \label{EQ:REM:lip:H_mu0}
        J_{s,t}\geq \E_{\mu^0}[Y_t|\F_s],
    \end{equation}
    where $\mu^0:=\1_{(s,T]}\Pi(\,\cdot\,,0)\in\A_s$.
    Let $\widetilde{\mathcal T}_2\subseteq[0,T]$ be the full $\ell_1$-measure Borel set provided by Lemma~\ref{LEM:integral_average}, such that, for $t\in\widetilde{\mathcal T}_2$
    \begin{equation}
    \label{EQ:REM:lip:L2_conv_tilde}
        \langle\widetilde Y_t\rangle_\eps:=\frac1\eps\int_t^{t+\eps}\big(|b_r|+L|\s_r|\big)\d r \longrightarrow|b_t|+L|\s_t|=:\widetilde Y_t, \qquad \text{ in }L^2.
    \end{equation}
    Use the non-decreasing monotonicity of the positive part and~\eqref{EQ:REM:lip:H_mu0}:
    \begin{equation}
        \left(U_\eps-\frac1{2\sqrt\eps}-J_{s,t}\right)^+
        \leq 
        \left(
            U_\eps-\frac1{4\sqrt \eps}
        \right)^+
        +
        \left(
            -\E_{\mu^0}[Y_t|\F_s]-\frac1{4\sqrt \eps}
        \right)^+,
    \label{EQ:REM:lip:estimate_B_minus_C}
    \end{equation}
    where  we used $(x+y)^+\leq x^++y^+$, for $x,y\in\R$.
    Concerning the first positive part on the right-hand side, Lemma~\ref{LEM:cond_expec_convex}\itemref{IT:LEM:cond_expec_convex:embedding} (with $K=L$, $\mathcal U=\mathcal B_\eps$, $u=1$, and $p=2$) yields $U_\eps
        \leq
        C_{L,T}
        \sqrt{\E\big[|\langle\widetilde Y_t\rangle_\eps|^2\big|\F_s\big]}
        =:Z_\eps$.
    Consequently,
    \begin{equation}
    \label{EQ:REM:lip}
        \left(
        U_\eps-\frac1{4\sqrt\eps}
        \right)^+
        \leq
        \left(
        Z_\eps-\frac1{4\sqrt\eps}
        \right)^+
        \leq Z_\eps.
    \end{equation}
    By~\eqref{EQ:REM:lip:L2_conv_tilde} and the tower property, the family $(Z_\eps)_{\eps>0}$ is uniformly bounded in $L^2$.
    Thus, for every $\delta>0$, Markov's inequality gives, as $\eps\to 0^+$:
    \[
        \P\left(
        \left(
        U_\eps-\frac1{4\sqrt\eps}
        \right)^+>\delta
        \right)
        \leq
        \P\left(
        Z_\eps>\frac1{4\sqrt\eps}+\delta
        \right)
        \leq
        \frac{\E[Z_\eps^2]}
        {\Big(\tfrac{1}{4\sqrt\eps}+\delta\Big)^2}
        \longrightarrow0.
    \]
    Thus, the first positive part converges to $0$ in probability. 
    Moreover,~\eqref{EQ:REM:lip} and the uniform $L^2$-bound on $Z_\eps$ show that it is uniformly integrable. 
    Hence, it converges to $0$ in $L^1$ by Vitali's theorem. 
    The second positive part on the right-hand side of~\eqref{EQ:REM:lip:estimate_B_minus_C} converges to $0$ $\P$-a.s. and is dominated by $\big(-\E_{\mu^0}[Y_t|\F_s]\big)^+\in L^1$, thus it converges to $0$ in $L^1$ by dominated convergence.
    Therefore, \eqref{EQ:REM:lip:B_minus_C} follows.
    
    We can finally conclude as in the quadratic case.
    Since $g(\,\cdot\,,0)=0$ and $g$ is $L$-Lipschitz in $z$ $\P\otimes\ell_1$-a.e., by Fubini's theorem there exists a full $\ell_1$-measure Borel set $\mathcal T_3\subseteq[0,T)$ such that, for every $t\in\mathcal T_3$, $|g(t,z)|\leq L|z|$ for all $z\in\R^m$, $\P$-a.s.
    For $t$ in the full $\ell_1$-measure Borel set
    $\mathcal T:=\mathcal T_1\cap\widetilde{\mathcal T}_2\cap\mathcal T_3$,
    \eqref{EQ:REM:lip:step_3} and $(x\vee y)^+\leq x^++y^+$
    give
    \[
        \bigg(
            \esssup_{\mu\in\mathcal B_\eps}V_\eps(\mu)-J_{s,t}
        \bigg)^+
        \leq
        \bigg(
            \esssup_{\mu\in\mathcal C_\eps}V_\eps(\mu)-J_{s,t}
        \bigg)^+
        +
        \left(
            U_\eps-\frac1{2\sqrt\eps}-J_{s,t}
        \right)^+.
    \]
    The two terms on the right converge to $0$ in $L^1$ by~\eqref{EQ:REM:lip:C} and~\eqref{EQ:REM:lip:B_minus_C}, respectively. 
    The negative part converges to $0$ in $L^1$ by~\eqref{EQ:REM:negative_part}. 
    Combining these convergences with the decomposition~\eqref{EQ:TH:convex_quadratic:L1_3_bis} yields~\eqref{EQ:TH:convex:fine}.
    The well-posedness of the resilience evaluation follows directly from the convergence of the positive and negative parts.
    Finally, Lemma~\ref{LEM:cond_expec_convex}\itemref{IT:LEM:cond_expec_convex:embedding} with $K=L$, $u=1$ and $p=2$, applied to $|Y_t|$ over the controls in $\A_s$, gives
    \[
        |J_{s,t}|
        \leq e^{L^2T/2}
        \big(\E[|Y_t|^2\mid\F_s]\big)^{1/2}.
    \]
    Taking $L^2$-norms and using $|Y_t|\leq|b_t|+L|\s_t|$ $\P$-a.s., proves~\eqref{EQ:TH:convex_lip_bounds_resilience}.
\end{enumerate}
\end{proof}

{\emergencystretch=3em 
\printbibliography[heading=bibintoc]}

\end{document}